\pdfoutput=1
\documentclass[aps,pra,amsmath,amssymb,twocolumn,10pt]{revtex4-2}
\usepackage{mathrsfs}
\usepackage[T1]{fontenc}
\fontencoding{T1}  
\usepackage[utf8]{inputenc}
\usepackage{enumitem}
\usepackage{tikz}
\usepackage{booktabs}
\usepackage{changes}

\newcommand{\vertiii}[1]{{\left\vert\kern-0.25ex\left\vert\kern-0.25ex\left\vert #1 
    \right\vert\kern-0.25ex\right\vert\kern-0.25ex\right\vert}}

\usepackage{amsfonts}
\usepackage{amsmath,amssymb,amsthm}
\usepackage{graphicx}
\usepackage{fancyhdr}
\usepackage[breakable]{tcolorbox}
\usepackage{bbm}
\usepackage{url}
\usepackage{mathtools}
\usepackage{braket}
\usepackage{upgreek}
\usepackage{dsfont}
\usepackage{collectbox}
\usepackage{braket}
\usepackage{quantikz}

\usepackage{float}

\usepackage[plain]{fancyref} %references with object type, no pages
\usepackage[ruled]{algorithm2e}

\usepackage{algpseudocode}

\usepackage{hyperref}

\newtheorem{thm}{Theorem}
\newtheorem*{thm*}{Theorem}

\newtheorem{lemma}[thm]{Lemma}
\newtheorem*{lemma*}{Lemma}
\newtheorem{proposition}[thm]{Proposition}
\newtheorem*{prop*}{Proposition}
\newtheorem{coro}[thm]{Corollary}
\newtheorem{cor}[thm]{Corollary} % Alias por si lo usas en el texto
\newtheorem*{cor*}{Corollary}
\newtheorem{condition}[thm]{Condition}

\newtheorem*{cj*}{Conjecture}

\newtheorem{problem}[thm]{Problem}

\theoremstyle{definition}
\newtheorem{definition}[thm]{Definition}
\newtheorem*{Def*}{Definition}
\newtheorem{rem}[thm]{Remark}
\newtheorem*{rem*}{Remark}

\def\beq{\begin{equation}}
\def\bq{\begin{quote}}
\def\eq{\end{quote}}
\def\ben{\begin{enumerate}}
\def\bit{\begin{itemize}}
\def\eit{\end{itemize}}

\def\r|{\right|}

\newcommand{\stirlingII}[2]{\genfrac{\{}{\}}{0pt}{}{#1}{#2}}

\newcommand{\norm}[1]{\left\|#1\right\|}

\newcommand{\Tr}{\operatorname{Tr}}

\newcommand\be{\begin{equation}}

\usepackage{subcaption}
\usepackage{graphicx}
\usepackage{caption}
\usepackage{ragged2e}
\usepackage{cancel}
\justifying

\newcommand{\Vt}{|\widetilde V(\widetilde G)|}
\newcommand{\wc}{w_{\mathrm c}}
\newcommand{\GC}{\mathcal G_C}
\DeclareMathOperator{\supp}{supp}
\DeclareMathOperator*{\Prob}{Pr}

\begin{document}
    \title{Efficient Quantum Monte Carlo through Cluster Expansions }

\author{\begingroup
\hypersetup{urlcolor=navyblue}
\href{https://orcid.org/0009-0002-1274-4747}{Jorge S\'anchez-Segovia
\endgroup}
}
\email[Jorge S\'anchez-Segovia]{ jorge.sanchezsegovia@estudiante.uam.es}
\affiliation{Instituto de F\'{i}sica Te\'{o}rica UAM/CSIC, C. Nicol\'{a}s Cabrera 13-15, Cantoblanco, 28049 Madrid, Spain}

\author{\begingroup
\hypersetup{urlcolor=navyblue}
\href{https://orcid.org/0000-0002-5889-4022}{\'Alvaro M. Alhambra
\endgroup}
}
\email[\'Alvaro M. Alhambra]{ alvaro.alhambra@csic.es}
 \affiliation{Instituto de F\'{i}sica Te\'{o}rica UAM/CSIC, C. Nicol\'{a}s Cabrera 13-15, Cantoblanco, 28049 Madrid, Spain}

\begin{abstract}
\setlength{\parindent}{0pt}
The sign problem remains a major obstacle to extending the success of Monte Carlo methods to generic quantum many-body systems. A powerful alternative relies on cluster expansions, which provide convergent representations of free energies and related quantities, with applications ranging from tensor network contractions to quantum dynamics and circuits. Within their convergence regime, these expansions often yield deterministic polynomial-time algorithms. However, such methods require exhaustive enumeration of connected clusters, leading to a computational overhead of $(N\varepsilon^{-1})^{\mathcal{O}(\log\Delta)}$, where $N$ is the system size, $\varepsilon$ the desired precision, and $\Delta$ the interaction-graph degree. This becomes particularly severe in long-range systems, where $\Delta\equiv N$, resulting in quasi-polynomial runtimes. Here, we introduce a Markov chain algorithm that samples cluster expansions directly in the space of connected clusters, replacing deterministic enumeration by Monte Carlo sampling. Because sampling occurs in an abstract cluster space rather than a physical configuration space, the conventional exponential sign problem is absent within the convergence regime. For general polymer models, we establish an FPRAS with runtime $\tilde{\mathcal{O}}(N^4\varepsilon^{-3})$. For quantum partition functions at temperatures $\beta\leq\beta_{\mathrm{sampling}}\equiv\mathcal{O}(1)$, this improves to $\tilde{\mathcal{O}}(N^{3+1/D}\varepsilon^{-3})$, where $D$ is the lattice dimension. The framework applies to long-range quantum systems with interactions $\propto1/r^\alpha$, including $2$-local models with $\alpha>2D$.

\end{abstract}

\maketitle

	\section{Introduction}

    Monte Carlo methods are ubiquitous across most computationally-demanding fields,  since they can replace the explicit evaluation of exponentially large sums with stochastic sampling. Notably, in classical many body physics, they have provided efficient schemes for computing equilibrium quantities, study phase transitions and critical phenomena. 

    Extending these methods to quantum many-body systems has long been a central goal in computational physics, but at finite temperature this is often obstructed by the \emph{sign problem} \cite{troyer2005computational}. The alternant sign of weights to sample from ends up leading to an exponentially large number of samples for approximating quantities of interest \cite{troyer2005computational,loh1990sign,kaul2013sign}. As a result, quantum Monte Carlo (QMC) is highly successful in sign-problem-free regimes, but its applicability remains fundamentally limited in generic interacting systems \cite{troyer2005computational}.  

The sign problem also appears in broader contexts beyond equilibrium QMC, including real-time path integrals, finite-density field theories, open quantum systems, and tensor-network contraction problems \cite{zheng2018sign}. While a variety of specialized representations and basis choices can sometimes alleviate the difficulty, and some models can be rigorously shown to be efficient (e.g.\cite{BravyiTerhal,bravyigosset,Crosson2021rapidmixingofpath,Crosson2025classicalsimulation,rayudu2025fastmixingoperatorlooppathintegral}), we cannot provide a general cure on complexity theoretic grounds \cite{Marvian2019,Hangleiter_2020}.

This computational barrier has driven extensive research across multiple fronts. On the quantum side, current efforts focus on algorithms for thermal state preparation, prominently featuring quantum Gibbs sampling (e.g. \cite{chen2023efficient,chen_efficient_2025,rouze_efficient_2026,ding2024efficientquantumgibbssamplers}). In parallel, the search for efficient classical simulation methods has led to approximation schemes based on \emph{cluster expansions} \cite{kotecky1986cluster,fernandez2007cluster}, which have emerged as a particularly powerful algorithmic tool for quantum counting and partition-function problems. Rooted in mathematical physics, these expansions provide a systematic way to organize perturbative corrections around a reference point and often lead to deterministic approximation algorithms for free energies and related quantities in suitable convergence regimes \cite{kotecky1986cluster,fernandez2007cluster,mann2024algorithmic}. They have recently been used in algorithmic applications ranging from quantum partition functions to the contraction of tensor networks \cite{Harrow_2020,Bravyi_2021,Haah_2024, Wild2023,mann2024algorithmic,mann2026efficientalgorithmsweaklyinteractingquantum,midha2025beliefpropagationclustercorrectedtensor}. 

However, most existing approaches remain enumeration-based, and require exhaustive summation over connected clusters. This yields at the very least a $(N\varepsilon^{-1})^{\mathcal{O}(\log \Delta )}$ complexity, where $N$ is the system size, $\varepsilon$ the desired precision and $\Delta$ the degree of the graph. These algorithms thus become prohibitively expensive in higher dimensions and even more so in long-range systems, where their runtime degrades to quasi-polynomial scaling \cite{mann2024algorithmic,Segovia_PRXQuantum2025,zlokapa2026sykthermalexpectationsclassically,zlokapa_rigorous_2026}. 
It was conjectured that Markov Chain Monte Carlo (MCMC) methods can yield better run-times for problems which can be written as a cluster expansion \cite{bergamaschi2026fastmixingalltoallquantum,mann2024algorithmic}. That this is possible is supported by recent numerical works \cite{sim2026stochasticloopcorrectionsbelief, Kai_e26050401}, by the existence of similar results for computing cumulant expansions in the fermionic setting \cite{chen2025convergencecumulantexpansionpolynomialtime} and by the empirical success of the closely related diagrammatic Monte Carlo \cite{DiagMC_Van_Houcke_2010,RossiDiagMC_PhysRevLett.119.045701}.

This conjecture is also supported by the fact that sampling methods for polymer models have been developed in the classical setting, yielding improved run-times \cite{blanca2023fast,chen2021fastalgorithmslowtemperatures}. These works sample from \emph{disconnected} configurations to approximate $Z$ to multiplicative error. To do so, they assume that the weights of the clusters are nonnegative. While these works have very recently been adapted to stoquastic quantum Hamiltonians \cite{mann2026fastalgorithmsstoquasticspin}, they cannot be straightforwardly replicated in the general quantum non-commuting setting: to do this, one needs to consider the sign as observable, recovering the usual \textit{exponential sign problem} of QMC approaches (See App. \ref{app:sign_problem}). Therefore, genuinely new methods need to be developed to apply MCMC techniques to this problem.

In this work, we propose an alternative approach to address this bottleneck: we construct a MCMC algorithm to sample the \emph{connected} cluster expansion directly, yielding an additive-error approximation of $\log Z$. Crucially, the algorithm does this by exploring the space of clusters rather than deterministically enumerating them. 

This sampling approach circumvents the sign problem that typically plagues QMC methods. Because our MCMC chain operates on the space of clusters rather than the physical spin configurations, it does not require the positivity of local weights, but it also does not suffer from variances inflated by vanishingly small average signs. Instead, the variance—and consequently, the runtime—is controlled by the analytic convergence conditions of the expansion. This enables us to compute an $\varepsilon$-additive approximation of $\log Z$ in expected runtime $\mathcal{O}\big( \vert V(G)\vert^4 \varepsilon^{-3} \big)$. This is a low-degree polynomial, which thus overcomes the graph-degree ($\Delta$) dependence that limits deterministic enumeration methods (see Fig. \ref{fig:parte_c}).

\begin{figure*}[t]
    \centering
    \label{fig:illustration_method}
    \begin{subfigure}{0.31\textwidth}
        \centering
        \includegraphics[width=\textwidth]{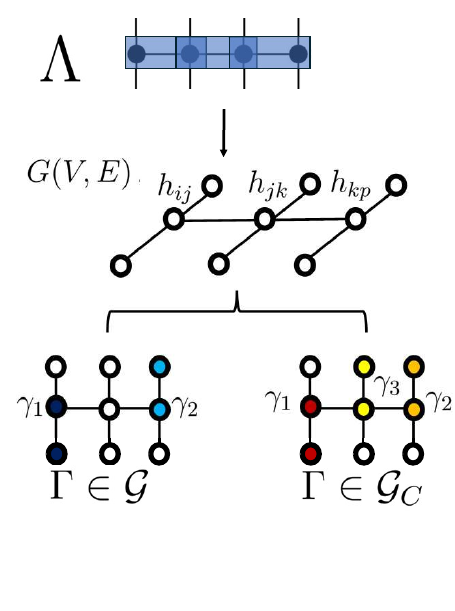}
        \caption{}
        \label{fig:parte_a}
    \end{subfigure}\hfill
    \begin{subfigure}{0.45\textwidth}
        \centering
        \includegraphics[width=\textwidth]{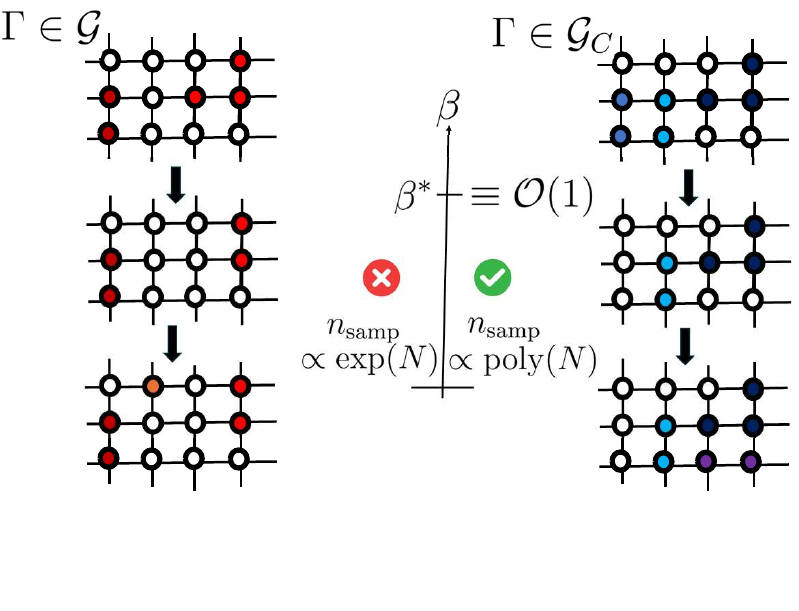}
        \caption{}
        \label{fig:parte_b}
    \end{subfigure}\hfill
    \begin{subfigure}{0.21\textwidth}
        \centering
        \includegraphics[width=\textwidth]{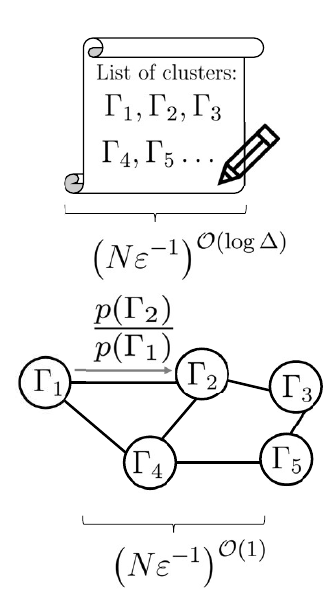}
        \caption{}
        \label{fig:parte_c}
    \end{subfigure}
    
    \caption{\justifying \textbf{Illustration of our sign-problem-free cluster sampling framework.} \textbf{(a)} Polymer model representation on a host graph $G(V,E)\equiv G_{\text{int}}(V,E)$ for a quantum many-body system on a lattice $\Lambda$, where vertices correspond to local Hamiltonian terms $h_i$ and polymers are connected subgraphs of $G(V,E)$. \textbf{Top:} A cluster $\Gamma \in \mathcal{G}_C$ (connected multiset of polymers). \textbf{Bottom:} An admissible configuration $\Gamma \in \mathcal{G}$ (mutually disconnected polymers). Note different polymers are depicted by different colors. \textbf{(b)} Typical approaches such as Glauber dynamics sample physical configurations targeting $Z$, suffering from the exponential sign problem, whereas sampling connected clusters targeting $\log Z$ is sign-problem-free for the convergence regime $\beta \leq \beta ^* \equiv \mathcal{O}(1)$. \textbf{(c)} Illustration of deterministic enumeration ($\Delta$-dependent) against our randomized FPRAS ($\Delta$-independent).}
    \label{fig:figura_conjunta}
\end{figure*}

A central operational requirement of our algorithm is the random generation of these local cluster configurations. While localized cluster sampling strategies have been previously developed for strictly bounded-degree graphs \cite{blanca2023fast}, they cannot be directly adapted to systems with long-range interactions. To bridge this gap, we introduce a tailored spatial sampling strategy designed to efficiently construct local clusters even when interactions decay algebraically. By doing so, our method successfully evaluates $k$-local quantum models with long-range interactions($\propto 1/r^\alpha$) for decay rates $\alpha > kD$ with runtime $\mathcal{O}\big( N^{3+1/D}\varepsilon^{-3} \big)$,  closing the quasipolynomial complexity gap of prior deterministic approaches \cite{Segovia_PRXQuantum2025,tong2026complexityclassificationhightemperaturebosons}.

During the preparation of this manuscript, we became aware of independent work by Putterman, Zlokapa, and Cotler~\cite{putterman2026quantumthermalstateslook}, who similarly recognized that truncated cluster expansions at high temperatures can be estimated via randomized sampling. While their work focuses on sampling combinatorial objects, yielding efficient estimators specifically for long-range Pauli Hamiltonians at high temperatures, our contribution establishes a MCMC approach that applies to general polymer models. We defer a detailed technical comparison of these complementary approaches to Sec. \ref{sec::conclusions}. 

\section{Technical preliminaries}

\subsection{Quantum Hamiltonians and partition functions} \label{sec:quantumHamilton}

We consider quantum systems with $N$ qudits on a $D$-dimensional lattice $\Lambda$ with $\vert \Lambda \vert =N$, governed by Hamiltonians with $k$-body (or $k-$local) terms
	\begin{equation}
		H= \sum_{Z: \vert Z \vert \leq k} h_{Z}. \label{eq:k-body-term}
	\end{equation}
At thermal equilibrium these systems are described by Gibbs states
    \begin{align}
        \rho_\beta= \frac{\mathrm{e}^{-\beta H}}{Z_\beta},
    \end{align}
where $Z_\beta=\Tr[\mathrm{e}^{-\beta H}] $ is the partition function. The main computational problem of interest in that setting \cite{bravyi2021complexity} is to approximate the partition function as follows.
    \begin{problem}[Quantum Partition Function ($k$-QPF)]
Given a $k$-local Hamiltonian $H$ acting on $N$ qudits and a precision parameter $\varepsilon $, compute an estimate $\widehat{\log Z}$ such that
\begin{align}
        \vert \log Z_\beta - \widehat{\log Z_\beta} \vert \leq \varepsilon,
    \end{align}
\end{problem}
This is another way of stating the problem of computing $Z$ with relative error, which in the most general case  this problem is QMA-hard \cite{bravyi2021complexity}. It is thus important to find regimes in which the partition function can be computed efficiently. With this purpose in mind, let us introduce the \textit{principle of inclusion-exclusion}.

    \begin{lemma}(Principle of inclusion-exclusion)
    \label{lemma_inclusion-exclusion}
        Let $f$ be a function defined on the subsets of a finite set $E$, then
        \begin{align*}
            f(E) = \sum_{S \subseteq E} (-1)^{|S|} \sum_{T \subseteq S} (-1)^{|T|} f(T)
        \end{align*}

    \end{lemma}
\noindent For a proof we refer the reader to \cite{mann2024algorithmic}.

It is useful at this point to define the \textit{interaction graph} $G_{\text{int}}(V,E)$, where each vertex corresponds to a Hamiltonian term $h_Z$, and an edge $e \in E(G_{\text{int}})$ connects two vertices if their corresponding terms have overlapping support on the lattice $\Lambda$.
Applying Lem. \ref{lemma_inclusion-exclusion} to the set of vertices $V(G_{\text{int}})$ yields  
\begin{align*}
    Z_\beta &= \text{Tr}\left[\mathrm{e}^{-\beta \sum_{Z \in V(G_{\text{int}})} h_Z}\right] \\       
    &= \sum_{S \subseteq V(G_{\text{int}})} (-1)^{|S|} \sum_{T \subseteq S} (-1)^{|T|} \text{Tr} \left[ \mathrm{e}^{-\beta \sum_{Z \in T} h_Z} \right].
\end{align*}

Crucially, any set of vertices $S$ can be uniquely partitioned into its connected components, which we define as \textit{polymers} $\gamma \subseteq V(G_{\text{int}})$. Let $\Gamma \in \mathcal{G}$ denote a set of these mutually disconnected polymers. Because these polymers do not share any edges in the interaction graph, their corresponding Hamiltonian terms act on disjoint regions of the lattice and therefore commute. This allows the trace to factorize over the components of $\Gamma$

\begin{align*}
    Z_\beta &= \sum_{\Gamma \in \mathcal{G}} \prod_{\gamma \in \Gamma} \left( (-1)^{\|\gamma\|} \sum_{T \subseteq E(\gamma)} (-1)^{|T|} \text{Tr} \left[ \mathrm{e}^{-\beta \sum_{Z \in T} h_Z} \right] \right) \\
    &= \sum_{\Gamma \in \mathcal{G}} \prod_{\gamma \in \Gamma} w_\gamma,
\end{align*}
where we define the polymer weight as
\begin{align}
    w_\gamma := (-1)^{\|\gamma\|} \sum_{T \subseteq E(\gamma)} (-1)^{|T|} \operatorname{Tr} \left[ \mathrm{e}^{-\beta \sum_{Z \in T} h_Z} \right].
\end{align}
This yields a \textit{polymer model} description of quantum partition functions \cite{mann2024algorithmic}(See Fig. \ref{fig:parte_a} for an illustration), which we introduce more formally in the following section, Sec. \ref{subsec_technical_polymermodels}. 

For simplicity of notation, we adopt the normalized trace convention throughout this work, defined such that $\operatorname{Tr}[\mathbb{I}] = 1$. Note that this convention simply evaluates a shifted log-partition function; the standard physical value is exactly recovered at the end of the algorithm by adding the analytic constant $ \log \mathcal{D}_\mathcal{H}$, where $\mathcal{D}_\mathcal{H}$ is the dimension of the full Hilbert space.
	
\subsection{Polymer models}\label{subsec_technical_polymermodels}

	After the example of the previous section, we now introduce \emph{polymer models} formally. This is a convenient way of writing partition functions and other related quantities. In an abstract setting, a polymer model is a triple $(\mathcal{C}, w, \sim)$, where $\mathcal{C}$ is a countable set whose objects are called polymers. The function $w : \mathcal{C} \to \mathbb{C}$ assigns to each polymer $\gamma \in \mathcal{C}$ a complex number $w_\gamma$ called the \emph{weight} of the polymer, and $\sim$ is a symmetric compatibility relation such that each polymer is incompatible with itself. Equivalently, the incompatibility relation $\nsim$ is a symmetric and reflexive relation. A set of polymers is called \textit{admissible} if all the polymers in the set are pairwise compatible. Note that the empty set is admissible.
Let $\mathcal{G}$ denote the collection of all admissible sets of polymers from $\mathcal{C}$. The polymer expansion of a partition function is defined by
	\begin{equation} \label{eq:polymerexp}
		Z(\mathcal{C}, w) = \sum_{\Gamma \in \mathcal{G}} \prod_{\gamma \in\Gamma} w_\gamma,
	\end{equation}
	where $\Gamma $ is an admissible set of polymers.

    More concretely, we focus on \textit{subset polymer models} \cite{blanca2023fast} in which all polymers are vertex-induced subgraphs from a host graph $G=(V,E)$, which in the partition function setting is $G_{\text{int}}(V,E)$. These consist of 
    \begin{itemize}
    \item A polymer set $\mathcal{C} = \mathcal{C}(G)$, where each polymer is a subgraphs in $G$. We denote supp$ (\gamma)$ to the vertices of $G(V,E)$ contained in $\gamma$.
    \item Complex weights $w_\gamma$ assigned to each $\gamma \in \mathcal{C}$. 
    \item A connectivity-based incompatibility relation $\nsim$. Two polymers $\gamma, \gamma' \in \mathcal{C}$ are considered incompatible (denoted $\gamma \nsim \gamma'$) if the union of their vertex sets induces a connected subgraph in $G$. Otherwise, they are compatible, denoted as $\gamma \sim \gamma'$.
\end{itemize}
Let $\norm{\cdot} : \mathcal{C} \to \mathbb{Z}^+$ assign to each polymer
$\gamma \in \mathcal{C}$ a positive integer $\norm{\gamma}$, its \textit{size},
given by the number of vertices of $G(V,E)$ contained in $\text{supp}(\gamma)$.

\subsection{The cluster expansion}

We now define the abstract cluster expansion \cite{kotecky1986cluster,fernandez2007cluster,Mann2021} of a polymer model. Let $(\Gamma, m_\gamma(\Gamma))$ be a multiset of polymers with multiplicities $m_\gamma(\Gamma)$. The incompatibility graph $H_\Gamma$ of $\Gamma$ is the graph with vertex set $\Gamma$ and edges between any two polymers if and only if they are incompatible. $\Gamma$ is called a \textit{cluster} if its incompatibility graph $H_\Gamma$ is connected. A polymer and a cluster are compatible if the polymer is compatible with every polymer in the cluster. 

Let $\mathcal{G}_{\mathcal{C}}$ denote the set of all clusters of polymers from $\mathcal{C}$. The abstract cluster expansion is a formal power series for $\log Z(\mathcal{C}, w)$ in the variables $w_\gamma$, defined by
\begin{align}
\log(Z(\mathcal{C}, w)) = \sum_{(\Gamma, m_\gamma) \in \mathcal{G}_C} \varphi(\Gamma) \prod_{\gamma \in \Gamma} \frac{w_\gamma^{m_\gamma(\Gamma)}}{m_\gamma(\Gamma)!},
\end{align}
where

\begin{align}
\varphi(\Gamma) := \sum_{\substack{S \subseteq E(H_\Gamma) \\ \text{spanning, connected}}} (-1)^{|S|}    
\end{align}
denotes the \textit{Ursell function} of the incompatibility graph $H_\Gamma$. For a derivation of the cluster expansion, we refer the reader to \cite{StatisticalMechanicsofLatticeSystems_friedli_velenik_2017}. 
For a cluster $(\Gamma, m_\gamma)$ we write
\begin{align}
\vert \Gamma \vert = \sum_{\gamma \in \Gamma} m_\gamma(\Gamma),
\qquad
\norm{\Gamma} = \sum_{\gamma \in \Gamma} m_\gamma(\Gamma) \norm{\gamma},
\end{align}
for the number of polymers it contains and for its total size.

Note that in the literature it is more common to consider clusters as ordered tuples of polymers, but our notation is better suited for algorithmic purposes \cite{Mann2021, kotecky1986cluster,Segovia_PRXQuantum2025,StatisticalMechanicsofLatticeSystems_friedli_velenik_2017}.

Throughout this work we focus on polymer models that satisfy particular cases of the following condition. 
\begin{condition}[Cluster Sampling Condition]
\label{con:clustersamplingcondition}
Let $(\mathcal{C}, w, \sim)$ be a polymer model, and let $a, d : \mathcal{C} \to \mathbb{R}^+$ be real-valued functions. The polymer model satisfies the Cluster Sampling condition if
\begin{align}\label{eq:cluster_sampling_condition}
    \sum_{\gamma^* \not\sim \gamma} \vert w_{\gamma^*} \vert^{1/2} \mathrm{e}^{a(\gamma^*) + d(\gamma^*)} \le a(\gamma),
\end{align}
for all polymers $\gamma \in \mathcal{C}$. Throughout this work, we specifically set $a(\gamma)=\norm{\gamma} /2$ and $d(\gamma)= 2 \log 2 \cdot \norm{\gamma}$. 
\end{condition} 
The specific choice of $d(\gamma)= 2 \log 2 \cdot\norm{\gamma}$ and $a= \norm{\gamma} /2$ is motivated by algorithmic efficiency. As detailed in Appendix \ref{app:ursell_calculation_cost}, our algorithm requires computing the Ursell functions $\varphi$. This choice of $d(\gamma)$ guarantees that the expected runtime of this subroutine is bounded by an $\mathcal{O}(1)$ constant. On the other hand the choice of $a(\gamma)$ is motivated in Remark \ref{rem_achoice}.
One could relax this to more general choices of $d(\gamma)$ and $a(\gamma)$, but doing so would introduce a $\mathrm{poly}(\vert V(G) \vert, \varepsilon^{-1})$ overhead to the overall algorithmic complexity (See App. \ref{app_rapidmixingproof},\ref{app:ursell_calculation_cost} for more details).

This condition is a strengthening of the usual condition of convergence of cluster expansions \cite{kotecky1986cluster}. 
\begin{condition}\label{con:kotecky_preiss_criteria}
Let $(\mathcal{C}, w, \sim)$ be a polymer model and let $a(\gamma) : \mathcal{C} \to \mathbb{R}^+$ and $d: \mathcal{C} \to \mathbb{R}^+$.  A polymer model satisfies the Koteck\'y -Preiss condition if 
\begin{align}\label{eq:kotecky_preiss_criteria}
    \sum_{\gamma^* \not\sim \gamma} |w_{\gamma^*}| \mathrm{e}^{a(\gamma^*) + d(\gamma^*)} \le a(\gamma),
\end{align}
for all polymers $\gamma \in \mathcal{C}$. 
\end{condition}
The difference is the square root in the polymer weights, which we require to bound the variance of the Monte-Carlo estimator we consider.
\begin{proposition}
    Con. \ref{con:clustersamplingcondition} $\implies$ Con.\ref{con:kotecky_preiss_criteria}.
\end{proposition}
\begin{proof}
        Since a polymer is self-incompatible ($\gamma \not\sim \gamma$), isolating the $\gamma^* = \gamma$ term in Eq. \eqref{eq:cluster_sampling_condition} implies that $\vert w_\gamma \vert^{1/2} \le a(\gamma) \mathrm{e}^{-a(\gamma)-d(\gamma)} \le 1$. Consequently, $\vert w_\gamma \vert \le \vert w_\gamma \vert^{1/2}$ for all $\gamma \in \mathcal{C}$.
\end{proof}

One of the central results of the theory of polymer models \cite{kotecky1986cluster} states that Con. \ref{con:kotecky_preiss_criteria} is sufficient for the absolute convergence of the cluster expansion.

\begin{lemma}
\label{lem:kotecky_preiss}
Let $(\mathcal{C}, w, \sim)$ be a polymer model such that Con. \ref{con:kotecky_preiss_criteria} holds.
Then, the cluster expansion for $\log(Z(\mathcal{C}, w))$ converges absolutely, $Z(\mathcal{C}, w) \neq 0$, and
\[
\sum_{\substack{\Gamma \in \mathcal{G}_{\mathcal{C}} \\ \Gamma \not\sim \gamma^* }} \left| \varphi(\Gamma) \prod_{\gamma \in \Gamma} \frac{w_{\gamma}^{m_\gamma}}{m_\gamma(\Gamma)!} \right| \mathrm{e}^{\sum_{\gamma \in \Gamma} d(\gamma)} \le a(\gamma^*),
\]
for all polymers $\gamma^* \in \mathcal{C}$.
\end{lemma}

 That is, Con. \ref{con:kotecky_preiss_criteria} yields analytical convergence bounds (and is the one often proven in applications \cite{mann2024algorithmic,Segovia_PRXQuantum2025,zlokapa2026sykthermalexpectationsclassically}), whereas the square-root strengthening in Con. \ref{con:clustersamplingcondition} is required in our algorithmic framework (See Lem. \ref{lemma_finitevariance}). As we see here, while the regimes of validity might differ, most of the existing proofs of one condition can be straightforwardly adapted to the other.

\subsection{MCMC and canonical paths method}

Since we consider sampling from cluster expansions with Markov chains, we now introduce some technical preliminaries on Markov chains and their convergence. Formally, a Markov chain on a finite state space $\mathcal{X}$ is a sequence of random variables $(X_t)_{t=0}^\infty$ satisfying the Markov property: the probability distribution of the next state $X_{t+1}$ depends solely on the current state $X_t$. The dynamics of this process are entirely governed by a transition matrix $P$, where 
\begin{align}
    \Pr &(X_{t+1} = x_{t+1} \mid X_t = x_t, \dots, X_0 = x_0) \\
    &= \Pr(X_{t+1} = x_{t+1} \mid X_t = x_t) \equiv P(x_t, x_{t+1}),
\end{align}
for all $x_t, x_{t+1} \in \mathcal{X}$. A fundamental property of such chains is the existence of a \textit{stationary distribution} $\pi$, defined as a probability distribution over $\mathcal{X}$ that remains invariant under the application of the transition matrix 
\begin{equation}
    \sum_{x \in \mathcal{X}} \pi(x) P(x, x') = \pi(x') \quad \text{for all } x' \in \mathcal{X}.
\end{equation}
If the Markov chain is irreducible (any state can eventually be reached from any other state) and aperiodic (it does not cycle deterministically), the distribution of $X_t$ is guaranteed to converge to a unique stationary distribution $\pi$ as $t \to \infty$ \cite{levin2017markovmixing}.

To quantify the rate of this convergence, we measure the distance between the distribution of $X_t$ and $\pi$. The standard metric used for this purpose is the \textit{total variation distance}. For two probability distributions $\mu$ and $\nu$ on $\mathcal{X}$, the total variation distance $d_{\text{TV}}(\mu, \nu)$ is defined as
\begin{align*}
    d_{\text{TV}}(\mu, \nu) &= \frac{1}{2} \sum_{x \in \mathcal{X}} |\mu(x) - \nu(x)|.
\end{align*}
Let $P^t(x, \cdot)$ denote the distribution of the chain after $t$ steps starting from an initial state $x$. The maximum distance to the stationary distribution at time $t$ is given by
\begin{equation}
    d(t) = \max_{x \in \mathcal{X}} d_{\text{TV}}(P^t(x, \cdot), \pi).
\end{equation}
With this metric established, the \textit{mixing time} $t_{\text{mix}}(\varepsilon_{\text{mix}})$ is defined as the minimum number of steps required to guarantee that the total variation distance is at most $\varepsilon_{\text{mix}}$, regardless of the starting state:
\begin{equation}
    t_{\text{mix}}(\varepsilon_{\text{mix}}) = \min \{ t \geq 0 : d(t) \leq \varepsilon_{\text{mix}} \}.
\end{equation}
To rigorously bound the mixing time $t_{\text{mix}}(\varepsilon_{\text{mix}})$ of a Markov chain on a finite state space $\mathcal{X}$, a powerful method we consider is that of \emph{canonical paths} \cite{levin2017markovmixing}. Let $G_{MC} = (\mathcal{X}, E)$ be the transition graph associated with the chain, where an edge $e = (u, v) \in E$ exists if the transition probability $P(u, v) > 0$. For every pair of states $x, w \in \mathcal{X}$, we define a unique, deterministic sequence of valid transitions—a canonical path $\mathcal{P}_{xw}$—connecting $x$ to $w$. The number of transitions in this path is denoted by $\vert{}\mathcal{P}_{xw}\vert{}$, and the maximum path length across all pairs in the state space is defined as $L= \max_{x,w} \vert{}\mathcal{P}_{xw}\vert{}$. The efficiency of the Markov chain depends on ensuring that no single transition is bottle-necked by these overlapping routes. This is quantified by the maximum \emph{congestion parameter} $\rho$, defined as
\begin{align} \label{eq:rho_congestion}     \rho = \max_{e=(u,v) \in E} \frac{1}{\pi(u) P(u,v)} \sum_{\mathcal{P}_{xw} \ni e} \pi(x) \pi(w) \vert{}\mathcal{P}_{xw}\vert{}
\end{align}
It is often mathematically convenient to decouple the path length from the probabilistic flow accumulation by introducing the unweighted congestion parameter $\bar{\rho}$
\begin{align} \label{eq:bar_rho_congestion}     \bar{\rho} = \max_{e=(u,v) \in E} \frac{1}{\pi(u) P(u,v)} \sum_{\mathcal{P}_{xw} \ni e} \pi(x) \pi(w) ,\end{align}
such that $\rho \le \bar{\rho} L$. Bounding the congestion directly lower-bounds the spectral gap of the transition matrix, which guarantees that the mixing time required to reach a total variation distance of at most $\varepsilon_{\text{mix}}$ scales as \cite{levin2017markovmixing}
\begin{align} \label{eq:mixing_time_bound}     
t_{\text{mix}}(\varepsilon_{\text{mix}}) \le \bar{\rho} L \log\left(\frac{1}{p_{\text{min}} \varepsilon_{\text{mix}}}\right).
\end{align}
In interacting many-body systems and polymer models defined on a host graph $G(V,E)$, transitions in the state space are strictly constrained by the geometry of $G$. The routing distance between elements in this graph, and crucially its diameter $\text{diam}(G)$, dictate the maximum number of local updates required to propagate a deformation across the system. In the context of canonical paths, ensuring that the maximum path length satisfies $L \le \mathcal{O}(\text{diam}(G))$ is the key to obtain tight bounds on the mixing time. 

Throughout this work, we frequently rely on Markov chains to estimate the expected value of various quantities. Let $\pi$ be the stationary distribution of a Markov chain on a state space $\mathcal{X}$, and let $O: \mathcal{X} \to \mathbb{C}$ be an observable of interest. The standard MCMC estimator for the expected value $\mathbb{E}_\pi[O]$ is the sample mean over $ n_s$ states generated by the chain
\begin{equation}
    \widehat{O} = \frac{1}{n_s} \sum_{i=1}^{n_s} O(z_i).
\end{equation}
To bound the error of this estimator and provide rigorous algorithmic guarantees, we use Chebyshev's inequality.

\begin{lemma}[Chebyshev's Inequality]
\label{lemma:chebyshev}
Let $x$ be a random variable with finite expected value $\mu$ and finite variance $\sigma^2$. Then, for any real number $\varepsilon > 0$, the probability that $x$ deviates from its mean by at least $\varepsilon$ is bounded by:
\begin{equation}
    \Pr(|x - \mu| \geq \varepsilon) \leq \frac{\sigma^2}{\varepsilon^2}.
\end{equation}
\end{lemma}
Applying Lem. \ref{lemma:chebyshev} to our MCMC estimator $\widehat{O}$, we require that the probability of the estimation error exceeding a tolerance $\varepsilon$ is bounded by a maximum failure probability $\delta$. If the samples are strictly independent by construction, the variance of the estimator scales exactly as $  \text{Var}(\widehat{O}) = \text{Var}_\pi(O) / n_s$. Consequently, to achieve the desired precision with probability of success $\geq 1-\delta$, the number of required parallel samples scales as
\begin{equation} \label{eq:sample_complexity}
    n_s \geq \text{Var}_\pi(O)\varepsilon^{-2} \delta^{-1}.
\end{equation}
Eq. \eqref{eq:sample_complexity} establishes a fundamental constraint for randomized approximation algorithms: if either the variance of the observable $\text{Var}_\pi(O)$ is excessively large or the desired precision $\varepsilon$ excessively small, the required number of samples $n_s$ diverges, rendering the MCMC approach inefficient. 

\subsection{Truncation and the Algorithmic Polymer Model}
\label{subsection_algorithmicpolymermodel}

A useful quantity for algorithmic purposes is the truncated cluster expansion $T_M(\mathcal{C}, w)$ for $\log Z(\mathcal{C}, w)$
\begin{align}\label{eq::def_T_M}
    T_M(\mathcal{C}, w) := \sum_{\substack{(\Gamma, m_\gamma) \in \mathcal{G}_{\mathcal{C}} \\ \|\Gamma\| \leq M}} \varphi(\Gamma) \prod_{\gamma \in \Gamma} \frac{w_{\gamma}^{m_\gamma}}{m_\gamma(\Gamma)!}.
\end{align}

To compare this with the full cluster expansion via the Koteck\'y-Preiss criterion (Con. \ref{con:kotecky_preiss_criteria}), we introduce the notion of a \textit{connected covering subgraph} $\tilde{G}(\tilde{V}, \tilde{E}) \subseteq G(V, E)$, whose effective volume is characterized by its number of vertices $\vert \tilde{V} \vert$. We define $\tilde{G}$ as the minimal connected subgraph of the underlying host graph such that every polymer $\gamma \in \mathcal{C}$ in the model is incompatible with at least one vertex in the covering set $\tilde{V}$. Thus, we have the following:

\begin{cor}[Truncation Error Bound]\label{cor:truncation_error}
Under Con. \ref{con:kotecky_preiss_criteria}, the truncation error of the estimator $T_M$ is bounded by
\begin{equation}
    \vert \log Z(\mathcal{C}, w) - T_M(\mathcal{C}, w) \vert \leq \mathcal{O}\left( \vert \tilde{V}(\tilde{G}) \vert \exp{(-M)} \right),
\end{equation}
yielding an $\varepsilon$-additive approximation by truncating at $M=\mathcal{O}( \log (\vert \tilde{V}(\tilde{G}) \vert \varepsilon^{-1}) )$.
\end{cor}

\begin{proof}
The truncation error can be bounded in the following way
\begin{align}
    \vert \log Z(\mathcal{C}, w) - T_M & (\mathcal{C}, w)\vert \leq \sum_{\substack{\Gamma \in \mathcal{G}_{\mathcal{C}} \\ \|\Gamma\| > M}}\left \vert \varphi(\Gamma) \prod_{\gamma \in \Gamma} \frac{w_{\gamma}^{m_\gamma}}{m_\gamma(\Gamma)!} \right \vert \nonumber \\
    &\leq \sum_{v \in \tilde{V}(\tilde{G})} \sum_{\substack{\Gamma \in \mathcal{G}_{\mathcal{C}} \\ \|\Gamma\| > M \\ \exists \gamma \in \Gamma:\, \gamma \nsim v}}\left \vert \varphi(\Gamma) \prod_{\gamma \in \Gamma} \frac{w_{\gamma}^{m_\gamma}}{m_\gamma(\Gamma)!} \right \vert \nonumber \\
    &\leq 4^{- M}\sum_{v \in \tilde{V}(\tilde{G})} a(v) \leq  \vert \tilde{V}(\tilde{G}) \vert 4^{- M} ,
\end{align}
where we introduced $d(\gamma)=2 \cdot \log 2 \cdot \norm{\gamma}$ inside the sum and used Lemma \ref{lem:kotecky_preiss} in the last line. 
\end{proof}

\begin{rem}
    In standard polymer models defined on finite-degree graphs, the truncation error is typically expressed directly in terms of the total graph volume $\vert V(G) \vert$, since the covering set scales trivially as $\vert \tilde{V}(\tilde{G}) \vert = \Theta(\vert V(G) \vert)$. However, explicitly isolating the effective volume $\vert \tilde{V}(\tilde{G}) \vert$ via a connected covering set becomes very useful in physical models with dense interactions. For instance, in quantum long-range systems, the abstract interaction graph $V(G)$ consists of all possible $k$-body Hamiltonian terms, leading to a total volume $\vert V(G) \vert \propto N^k$. By contrast, the connected covering set $\tilde{V}(\tilde{G})$ can be naturally chosen to represent the \textit{underlying local model} (e.g., the grid of single-site fields or nearest-neighbor couplings) on the physical lattice. Since every non-local interaction must overlap with this physical skeleton, $\tilde{V}(\tilde{G})$ successfully covers the entire polymer model while guaranteeing that its size scales linearly with the number of physical constituents, $\vert \tilde{V}(\tilde{G}) \vert \propto N$.
\end{rem}

So far, we have established that truncating the cluster expansion to $M= \mathcal{O}(\log (\vert \tilde{V}(\tilde{G})\vert \varepsilon^{-1}))$ achieves an $\varepsilon$-additive approximation. While discarding large polymers is standard in deterministic cluster expansion methods \cite{Mann2021,mann2024algorithmic}, our randomized algorithms additionally require avoiding clusters whose statistical weight is vanishingly small (See Eq.~\eqref{eq:mixing_time_bound}). This motivates the concept of the \textit{algorithmic polymer model}. Before
defining it, we make explicit the form in which the Kotecký--Preiss criterion is
verified in practice. In all models we consider, Con.
\ref{con:kotecky_preiss_criteria} is established not for the weights themselves
but for an explicit majorant
\begin{align}\label{eq:majorant}
    \vert w_\gamma \vert \;\leq\; \bar{w}_\gamma .
\end{align}
In standard finite-range models
one may simply take $\bar{w}_\gamma = \mathcal{O}(\mathrm{e}^{-c\norm{\gamma}})$,
whereas for the long-range Hamiltonians of Sec. \ref{sec::aplication} the
majorant also absorbs the algebraic decay of the couplings (see App.
\ref{app:kp_longrange}). Having such an explicit majorant is what allows the
algorithm to decide which polymers are negligible.

\begin{definition}[Algorithmic Polymer Model]\label{def::algorithmic_polymer_model}
    Let $\kappa: \mathcal{C} \to \mathbb{R}^+$. Given a polymer model
    $\mathcal{P}=(\mathcal{C}, w, \sim)$, we define its algorithmic truncation as
    $\mathcal{P}_{\text{alg}}=(\mathcal{C}_{\text{alg}}, w, \sim, M, \kappa)$,
    where the truncated set of polymers is restricted to
    \begin{equation}
        \mathcal{C}_{\text{alg}} = \{ \gamma \in \mathcal{C} \mid
        \norm{\gamma} \leq M \text{ and } \vert w_\gamma \vert > \kappa(\gamma) \},
    \end{equation}
    and the clusters assembled from them are retained up to twice that size,
    \begin{equation}
        \mathcal{G}_{\text{alg}} = \{ \Gamma \in \mathcal{G}_{\mathcal{C}_{\text{alg}}}
        \mid \norm{\Gamma} \leq 2M \}.
    \end{equation}
\end{definition}

\begin{lemma}[Algorithmic Truncation Bound]\label{lem:algorithmic_truncation}
    Let $\mathcal{P}=(\mathcal{C}, w, \sim)$ satisfy Con.
    \ref{con:kotecky_preiss_criteria} with majorant $\bar{w}$. Then its
    algorithmic truncation, with size limit and weight threshold
    \begin{align}
        M = \mathcal{O}\!\left(\log \left(2\vert \tilde{V}(\tilde{G}) \vert
        \varepsilon^{-1}\right)\right),
        \qquad
        \kappa(\gamma) = \frac{\varepsilon}{2 \vert \tilde{V}(\tilde{G})\vert}\,
        \bar{w}_\gamma ,
    \end{align}
    provides an $\varepsilon$-additive approximation for $\log Z$,
    \begin{align}
        \vert \log Z (\mathcal{C}, w) - T_M (\mathcal{C}_{\text{alg}}, w) \vert
        \le \varepsilon .
    \end{align}
\end{lemma}

\begin{proof}
    By the triangle inequality,
    \begin{align}
        \vert\log Z (\mathcal{C}, w) - T_M (\mathcal{C}_{\text{alg}}, w) \vert
        &\leq \vert \log Z (\mathcal{C}, w) - T_M (\mathcal{C}, w) \vert \nonumber \\
        &\quad + \vert T_M (\mathcal{C}, w)- T_M (\mathcal{C}_{\text{alg}}, w) \vert .
    \end{align}
    The first term is bounded by $\varepsilon/2$ by the choice of $M$ (Cor.
    \ref{cor:truncation_error}). Every cluster contributing to the second term
    contains at least one discarded polymer, for which $\vert w_\gamma \vert \leq
    \kappa(\gamma)$, while all its remaining weights obey $\vert w_\gamma \vert
    \leq \bar{w}_\gamma$. Extracting the threshold fraction from the discarded
    polymer,
    \begin{align}
        \vert T_M (\mathcal{C}, w) -& T_M (\mathcal{C}_{\text{alg}}, w) \vert
        \\
        &\leq \frac{\varepsilon}{2\vert \tilde{V}(\tilde{G})\vert}
        \sum_{\Gamma \in \mathcal{G}_{\mathcal{C}}} \left \vert \varphi(\Gamma)
        \prod_{\gamma \in \Gamma} \frac{\bar{w}_\gamma^{\,m_\gamma(\Gamma)}}
        {m_\gamma(\Gamma)!} \right \vert \leq \frac{\varepsilon}{2},
    \end{align}
    where the last step is Con. \ref{con:kotecky_preiss_criteria} applied to the
    majorant, which bounds the unconstrained cluster sum by $\vert
    \tilde{V}(\tilde{G}) \vert$. Combining both terms yields the claim.
\end{proof}

The truncation above discards every polymer whose weight falls below
$\kappa(\gamma)$. For the Markov chain of Sec.
\ref{subsec::cluster_markov_chain} to mix in polynomial time we additionally
need this threshold not to be super-polynomially small, since the mixing bound
depends on the minimal stationary probability through $\log(1/p_{\text{min}})$
(see App. \ref{app_rapidmixingproof}).

\begin{condition}[Regularity of Weights]\label{ass::weight_regularity}
    The majorant of Eq. \eqref{eq:majorant} is bounded from below by
    \begin{align}
        \bar{w}_\gamma \;\geq\; \left( \frac{\varepsilon^{p}}
        {\vert \tilde{V}(\tilde{G}) \vert^{k}} \right)^{\norm{\gamma}} ,
    \end{align}
    for some global constants $k,p > 0$.
\end{condition}

This is a mild requirement, inp ractice satisfied by all the models we consider: it only
excludes weights decaying faster than any polynomial in $\varepsilon$ and
$\vert \tilde{V}(\tilde{G}) \vert^{-1}$.

Without loss of generality, we assume throughout this work that the state space of the polymer model is explicitly lower-truncated by the weight threshold $\kappa(\gamma)$. For any general polymer model, this truncation can be rigorously enforced by allocating $\varepsilon/2$ of the global error budget to bound the aggregate weight of the discarded polymers. Because this truncation merely rescales the target precision $\varepsilon$ by a constant factor without altering the asymptotic complexity bounds, we consider models strictly bounded by $\kappa(\gamma)$ by default to ease notation. Likewise, all cluster sums from here on are taken over
$\mathcal{G}_{\text{alg}}$, i.e. restricted to $\norm{\Gamma} \leq 2M$, and we
keep writing $T_M$ for the truncated expansion $T_M(\mathcal{C}_{\text{alg}},w)$
of the algorithmic model.

\section{Sampling the cluster expansion with MCMC}
We now construct our family of Monte Carlo algorithms to estimate the cluster expansion to additive error $\varepsilon$ for polymer models. Note that we can rewrite the truncated cluster expansion $T_M$ as
	
	\begin{align}\label{eq_Tmestimator}
		T_M =& \sum_{\substack{\Gamma \in \mathcal{G}_{\text{alg}} \\ \norm{\Gamma} \leq M}}
       \varphi(\Gamma) \prod_{\gamma \in \Gamma}
       \frac{w_\gamma^{m_\gamma(\Gamma)}}{m_\gamma(\Gamma)!}
   = A_M \cdot \mathbb{E}_{p(\Gamma)}[\xi(\Gamma)]
	\end{align}
	where we defined the cluster distribution
    \begin{align} \label{eq:targetp}
        p(\Gamma) = \begin{cases}
    \dfrac{\prod_{\gamma \in \Gamma}\left\vert w_\gamma^{m_\gamma(\Gamma)}\right\vert}{A_M}
      & \Gamma \in \mathcal{G}_{\text{alg}} \\[4pt]
    0 & \text{otherwise}
\end{cases}
    \end{align}
 with normalizing constant
 \begin{align}
        A_M = \sum_{\Gamma \in \mathcal{G}_{\text{alg}}}
      \prod_{\gamma \in \Gamma} \left\vert w_\gamma^{m_\gamma(\Gamma)} \right\vert,
    \end{align}
and the observable
    \begin{align}
        \xi(\Gamma)=s (\Gamma) \cdot \frac{\vert   \varphi(\Gamma) \vert}{\prod_{\gamma \in \Gamma}m_\gamma(\Gamma)!}  ,
    \end{align}
where $s (\Gamma)$ is the phase/sign associated to the cluster 
    \begin{align}
        s (\Gamma)=\frac{  \varphi(\Gamma)  }{\vert  \varphi(\Gamma) \vert }  \cdot \prod_{\gamma \in \Gamma} \frac{ w_\gamma^{m_\gamma(\Gamma)}}{ \left \vert w_\gamma^{m_\gamma(\Gamma)} \right \vert }.
    \end{align}
In this section, we construct, through a Markov chain over the cluster configurations, a sampler $S_{\text{MC}}$ with target distribution $p(\Gamma)$. Using it, we obtain estimates $\widehat{\xi}$ and $\widehat{A_M}$ such that their product approximates $T_M$ with an error bounded by $\varepsilon$
\begin{align}
    \vert \widehat{\xi} \cdot \widehat{A_M} - T_M \vert \leq \varepsilon.
\end{align}
By Lem. \ref{lem:algorithmic_truncation}, the output of our algorithm is a
$2\varepsilon$-close approximation to $\log Z$ provided
$M = \mathcal{O}(\log(\vert \tilde{V}(\tilde{G})\vert \varepsilon^{-1}))$.

To construct a valid Markov chain, we require a strictly positive target distribution. We achieve this in Eq. \ref{eq:targetp} by absorbing any complex phases or alternating signs into the cluster observable $\xi(\Gamma)$, effectively sampling from a distribution $p(\Gamma)$ proportional to the absolute weights $\vert w_{\gamma}\vert$. 

In standard QMC methods, such absolute-value transformations typically lead to an algorithmic bottleneck known as the \textit{sign problem} \cite{signpbm_Pan_2024}. This manifests as an exponential blow-up in the required number of samples because the target quantity is only recovered via reweighting, which involves dividing by an exponentially small normalization ratio \cite{signpbm_Pan_2024} (corresponding to a vanishing mean sign). 

By estimating the truncated expansion directly as a product $T_M = A_M \cdot \mathbb{E}_{p(\Gamma)}[\xi(\Gamma)]$ (Eq.~\eqref{eq_Tmestimator}), we circumvent any division by an exponentially small denominator. However, for this product estimator to be efficient, we must ensure that neither the normalization constant $A_M$ nor the statistical variance of the observable $\xi(\Gamma)$ diverge substantially.

The absolute convergence guaranteed by Con. \ref{con:clustersamplingcondition} imposes strict mathematical control over both factors. We first formalize this intrinsic variance control for the observable.

\begin{lemma}\label{lemma_finitevariance}
    Let the polymer model $(\mathcal{C}, w, \sim)$ satisfy Con. \ref{con:clustersamplingcondition}. Then
    \begin{align}
        \mathbb{E}_{p(\Gamma)}[\vert\xi(\Gamma)\vert^2] \leq \mathcal{O}(\vert \tilde{V}(\tilde{G}) \vert /A_M),
    \end{align}
    
\end{lemma}
\noindent The proof is shown in App. \ref{app:variance_convergence}.
    
While Lemma \ref{lemma_finitevariance} ensures that the statistical fluctuations of our estimator are well-behaved, we must also guarantee that the normalization constant $A_M$ itself does not introduce a severe overhead. To bound the truncation error and establish $A_M$ to be extensive ($\propto \vert \tilde{V}(\tilde{G}) \vert $) we introduce a mild assumption on the polymer models that we study.
 
\begin{condition}[Uniform Local Weights]\label{ass::uniform_local_weights}
There exists a known integer $r\ge1$, independent of the system size and of the
target accuracy, such that every polymer of non-vanishing weight has size
$\|\gamma\|\ge r$. Moreover, for every vertex $v\in\widetilde V(\widetilde G)$
there exists a local polymer $\gamma_v\in\mathcal C$ such that
$v\in\supp(\gamma_v)$, of size exactly $\|\gamma_v\|=r$, and whose weight
satisfies $|w_{\gamma_v}|\geq w_c>0$ for some global constant
$w_c\equiv\mathcal O(1)$. These representative polymers are unique to each
vertex, meaning $\gamma_v\neq\gamma_u$ for all $v\neq u$.
\end{condition}
 
In practice, this holds without loss of generality for physical systems (e.g., those containing single-site terms or non-zero local fields, See Sec. \ref{sec::aplication}). While this requirement guarantees optimal scaling, it can be mathematically relaxed at the cost of a polynomial overhead in the algorithmic runtime (see Appendix \ref{app_rapidmixingproof}).
This implies that $A_M$ scales extensively with the effective volume.

\begin{cor}
Let the polymer model satisfy Con. \ref{ass::uniform_local_weights} and write
\begin{equation}
  B_r\;:=\!\!\sum_{\gamma\,:\,\|\gamma\|=r}\!\!|w_\gamma| .
  \label{eq:Br}
\end{equation}
Then
\begin{equation}
  \wc\Vt\;\le\;B_r\;\le\;A_M .
  \label{eq:Brchain}
\end{equation}
\end{cor}
 
\begin{proof}
By Con. \ref{ass::uniform_local_weights} the representatives $\{\gamma_v\}_{v\in\widetilde V(\widetilde G)}$
are pairwise distinct polymers of size exactly $r$ with $|w_{\gamma_v}|\ge\wc$,
so they contribute separately to \eqref{eq:Br} and
$B_r\ge\wc\Vt$. Since every polymer of size $r$ is itself a cluster of
$\GC$ of total size $r\le M$, the sum \eqref{eq:Br} is a subsum of the one
defining $A_M$, whence $B_r\le A_M$.
\end{proof}

Under this lower bound, Con. \ref{con:kotecky_preiss_criteria} further guarantees that the finite-order truncation of $A_M$ has an extensive upper bound and converges exponentially fast to $A_M$. 
\begin{lemma}[Convergence and Extensivity of $A_M$]\label{lemma_Aextensive}
    Let the polymer model $(\mathcal{C}, w, \sim)$ satisfy Con. \ref{con:kotecky_preiss_criteria}. Then
    \begin{equation}
        \sum_{\substack{(\Gamma, m_\gamma) \in \mathcal{G}_{C} \\ \Gamma \nsim \gamma^*}}\prod_{\gamma \in \Gamma}\left( \vert w_\gamma \vert \mathrm{e}^{ d(\gamma) } \right)^{m_\gamma (\Gamma)}\leq a(\gamma^*).
    \end{equation}
    Additionally, given $M'\leq M$, the $M'$-th order truncated expansion of $A_M$, defined as 
    \begin{align}
        A_{M'} =\sum_{\substack{(\Gamma, m_\gamma) \in \mathcal{G}_{\text{alg}} \\ \norm{\Gamma} \leq 2M'}} \prod_{\gamma \in \Gamma}  \left\vert w_\gamma^{m_\gamma(\Gamma)} \right\vert,
    \end{align}

    \begin{align}
        A_M = \sum_{\Gamma \in \mathcal{G}_{\text{alg}}}
      \prod_{\gamma \in \Gamma} \left\vert w_\gamma^{m_\gamma(\Gamma)} \right\vert,
    \end{align}

    converges exponentially to $A_M$
    \begin{equation}
        \vert A_M - A_{M'} \vert \leq \mathcal{O}\left( \vert \tilde{V}(\tilde{G}) \vert  \cdot\exp{(-M')} \right).
    \end{equation}
    Furthermore, if Con. \ref{ass::uniform_local_weights} is satisfied, the normalization constant $A_M$ is strictly extensive and scales proportionally to the effective volume:
    \begin{equation}
        A_M = \Theta\left(\vert \tilde{V}(\tilde{G}) \vert\right).
    \end{equation}
\end{lemma}
\noindent    The proof is shown in App. \ref{app:A_convergence}.

\subsection{Main result}

We now construct and prove the convergence of our central algorithm. 
Before constructing the MCMC scheme through the cluster space, let us remark that our Markov chain transitions require generating individual polymers locally. For now, let us assume access to a perfect sampling oracle that outputs a polymer connected to a specific vertex $v \in V(G)$ according to the distribution

	\begin{align}
		\nu_v(\gamma)=
		\begin{cases}
			\vert w_\gamma \vert \quad \forall \, \gamma \neq \emptyset \\
			1- \sum_{\gamma \nsim v}\vert w_\gamma \vert \quad \gamma = \emptyset.\\
		\end{cases}
		\label{eq::target_polymersampler_distribution}
	\end{align}
    Provided Con. \ref{con:kotecky_preiss_criteria} is satisfied, $\nu_v(\gamma)$ constitutes a valid, properly normalized probability distribution. That such a perfect sampler exists is a weak assumption: polymer samplers have been constructed for models of interest, such as those in finite-degree graphs $G(V,E)$  \cite{chen2021fastalgorithmslowtemperatures,blanca2023fast}, and for the long-range case of interest here we independently construct it explicitly in Lem. \ref{lem:longrange_polymer_sampler}.

We now state one of our main technical results, where access to such a sampler
is assumed, and we denote by $C_w$ its expected cost per call.
\begin{thm}\label{thm:sampling_cexp}
        Let the polymer model $(\mathcal{C}, w, \sim)$ defined on a graph $G(V,E)$ satisfy both  Con. \ref{con:clustersamplingcondition} and Con. \ref{ass::uniform_local_weights} and assume access to a perfect polymer sampler of $\nu_v (\gamma)$. There exists a randomized algorithm to approximate $\log Z(\mathcal{C}, w)$ to error $\varepsilon$ with success probability $\geq 1-\delta $ in poly$(\vert V(G) \vert,   \varepsilon^{-1}, \log \delta^{-1})$ runtime  

        \begin{align}
            \mathcal{\tilde{O}}\left( C_w \cdot \text{diam}(\tilde{G}) \cdot \vert \tilde{V}(\tilde{G}) \vert^3\cdot \varepsilon^{-3} + \vert V(G) \vert \right) 
        \end{align}
 which is at most
\begin{align}
            \mathcal{\tilde{O}}\left( C_w \cdot \vert V(G) \vert^4\cdot \varepsilon^{-3} \right) .
        \end{align}
    \end{thm}

\begin{proof}(Thm. \ref{thm:sampling_cexp}) We divide the proof in two: We first assume that we have access to an exact sampler $S_{\text{exact}}$ that samples from $p(\Gamma)$ and analyze the runtime of estimating $\xi$ and $A_M$ with the desired accuracy. Later, we will substitute $S_{\text{exact}}$ by $S_{\text{MC}}$, the Cluster Markov chain presented in Sec. \ref{subsec::cluster_markov_chain} and analyze the full runtime of the algorithm. 

   To compute $\log Z$ up to additive error $\varepsilon$ it suffices to obtain a randomized estimator of the truncated cluster expansion $T_M$ to order $M= \log (2 \vert \tilde{V}(\tilde{G}) \vert / \varepsilon)$ (See Cor. \ref{cor:truncation_error})

    \begin{equation}
         \vert \log Z - \widehat{T_M} \vert \leq \vert \log Z - T_M \vert + \vert \widehat{T_M} - T_M \vert \leq \varepsilon/2 + \varepsilon/2.
    \end{equation}
    To do this task we use our estimator: we compute $\xi(\Gamma)$ and $A_M$ with the sampling oracle $S_{\text{exact}}$ such that

    \begin{align}
      \nonumber  \vert  \widehat{T_M} - T_M \vert =&  \left \vert \widehat{A_M} - A_M \right \vert\cdot \mathbb{E}_{p(\Gamma)}[\xi(\Gamma)] + A_M \cdot \left \vert \widehat{\xi}    - \mathbb{E}[\xi(\Gamma)] \right \vert \\
      \nonumber  & +\left \vert \widehat{A_M} - A_M \right \vert \cdot \left \vert \widehat{\xi}    - \mathbb{E}[\xi(\Gamma)] \right \vert \\
        \leq & \varepsilon_A \cdot  \mathbb{E}[\xi(\Gamma)] + \varepsilon_\xi \cdot A_M + \varepsilon_A \cdot \varepsilon_\xi.
    \end{align}
    so  by computing $A_M$ up to additive error $\varepsilon_A =\varepsilon A_M/8\vert \tilde{V}(\tilde{G})\vert $ and $\xi$ up to $\varepsilon/8A_M$ we provide an accurate enough estimation

    \begin{align}
        \vert  \widehat{T_M} - T_M \vert \leq&  \frac{\varepsilon\mathbb{E}[\xi]}{8\vert \tilde{V}(\tilde{G})\vert /A_M}+\frac{\varepsilon}{8} + \frac{\varepsilon^2}{64\cdot (\vert \tilde{V}(\tilde{G})\vert /A_M)\cdot A_M} \nonumber \\
        \leq& \varepsilon/2,
    \end{align}
    where we used Lem. \ref{lemma_Aextensive}.

\textit{Estimation of $\xi$ and $A_M$}
\begin{itemize}
    \item \textit{Estimation of $\xi$:} We make $n_\xi$ calls to $S_{\text{exact}}$, obtaining a sample of clusters $\{\Gamma_i\}_{i=1}^{n_\xi}$ and compute the following estimator of $\xi$

    \begin{align}
        \widehat{\xi}= \frac{1}{n_\xi} \sum_{i} \xi(\Gamma_i).
    \end{align}
    Because the observable $\xi(\Gamma)$ is complex-valued, we apply Lem. \ref{lemma:chebyshev} independently to its real and imaginary components. For any complex random variable we have $\operatorname{Var}(\operatorname{Re}\xi) \le \mathbb{E}[(\operatorname{Re}\xi)^2] \le \mathbb{E}[\vert\xi\vert^2]$, and similarly for $\operatorname{Im}\xi$. Thus, paying a factor of 2 in the total sample count $n_\xi$ we can safely formulate our sample complexity directly in terms of $\mathbb{E}[\vert\xi\vert^2]$. Therefore, to have probability of failure smaller than $\delta_\xi$
    \begin{align}
        \Pr( \vert \widehat{\xi(\Gamma)}   - \mathbb{E}[\xi (\Gamma)] \vert \geq \varepsilon_\xi = \varepsilon/8A_M ) \leq \frac{64 \mathbb{E}[\vert \xi \vert^2(\Gamma)] A_M^2}{\varepsilon^2 n_\xi},
     \end{align}

    we need $n_\xi \geq \mathcal{O}(\vert \tilde{V}(\tilde{G}) \vert^2  \varepsilon^{-2} \delta_\xi^{-1} )$ samples, where we use Lem. \ref{lemma_finitevariance}. 

    Note that within each of these samples we need to call the polymer sampler, with a $\mathcal{O}(C_w)$ cost, and to compute the Ursell function $\varphi$. This can be done in $\mathcal{O}(\vert \tilde{V}(\tilde{G}) \vert /A_M)\equiv \mathcal{O}(1)$ runtime by Lem. \ref{lemm::calculation_ursell} and under Con. \ref{ass::uniform_local_weights}, yielding a total runtime for estimating $\xi$ of $\mathcal{O}(C_w \vert \tilde{V}(\tilde{G}) \vert^2  \varepsilon^{-2} \delta_\xi^{-1} )$. 

    \item \textit{Annealing for $A_M$:} We employ the annealing schedule of Sec. \ref{subsec:eff_estimat_annealing} to estimate $A_M$ up to $\varepsilon_A = \varepsilon A_M/8 \vert \tilde{V}(\tilde{G}) \vert $. We need to obtain $n_{\text{ann}}\geq \mathcal{\tilde{O}}(\vert \tilde{V}(\tilde{G}) \vert^2 \varepsilon^{-2} \delta_A^{-1}  )$ samples, added to a calculation of the anchor $A_r$ of cost $\mathcal{O}(\vert V(G) \vert)$. For the probability of failure to be less than $\delta_A$, this procedure yields a total runtime of $\mathcal{O}\left(C_w \vert \tilde{V}(\tilde{G}) \vert^2 \varepsilon^{-2} \delta_A^{-1} + \vert V(G) \vert \right)$.
    
    \item \textit{Total Estimate:} We provide our estimator $\widehat{T}_M =\widehat{\xi} \cdot \widehat{A_M}$, yielding an $\varepsilon$-additive error for $\log Z$ in runtime
    \begin{align}
        \mathcal{O}\left(C_w \cdot  \vert \tilde{V}(\tilde{G}) \vert^2 \varepsilon^{-2} \left( \delta_A^{-1}+ \delta_{\xi}^{-1}\right)+ \vert V(G) \vert \right).
    \end{align}
    
\end{itemize}

\textit{Mixing and Bias Control:} 
Until now we had assumed that we have no bias and that we have access to perfect samples of $p(\Gamma)$ with $S_{\mathrm{exact}}$. We now replace the exact sampling oracle $S_{\mathrm{exact}}$ by $S_{\text{MC}}$, the Cluster Markov chain of Sec. \ref{subsec::cluster_markov_chain}. To relate its runtime with the one of the sampling oracle $S_{\mathrm{exact}}$ we perform a coupling argument  \cite{chen2021fastalgorithmslowtemperatures}. For each time one needs to call the sampling oracle there exists a randomized algorithm $S_{\text{MC}}$ that outputs an $\varepsilon_{\mathrm{mix}}$-approximate sample from the target distribution, so we choose

\begin{align}
    \varepsilon_{\mathrm{mix}} = \frac{\delta_{\text{mix}}}{ n_{\text{ann}}  +n_{\xi}},
\end{align}
as we need to run $n_{\text{ann}}$ and $ n_{\xi}$ samples for the annealing procedure and estimation of $\xi$ respectively.

By the definition of total variation distance, we may couple the outputs of $S_{\text{MC}}$ and $S_{\mathrm{exact}}$ so that, at each sampling step, the two samples coincide with probability at least $1 - \varepsilon_{\mathrm{mix}}$. We run the algorithm simultaneously using both $S_{\text{MC}}$ and $S_{\mathrm{exact}}$ under this coupling. Let $B_{\text{mix}}$ be the event that at least one of the $n_{\text{ann}} +n_{\xi}$ samples generated during the execution does not match its coupled exact counterpart. By a union bound the total probability of error due to imperfect mixing is 
\begin{align}
    \Pr(B_{\text{mix}}) \leq \left(n_{\text{ann}}   +n_{\xi} \right) \cdot \varepsilon_{\mathrm{mix}} = \delta_{\text{mix}}.
\end{align}

\textit{Error Analysis and Complexity:}
If our overall algorithm fails to output an $\varepsilon$-approximation, it must logically be true that either the Markov chain failed to couple, or the exact samples themselves produced a bad estimate. Applying the union bound directly to these isolated failure events yields
\begin{align}
\label{eq_proof_ProbSuccess}
    \Pr(\text{failure}) &\le \Pr(\text{failure using } S_{\text{exact}}) + \Pr(B_{\text{mix}}) \nonumber \\
    &\le (\delta_\xi + \delta_A) + \delta_{\text{mix}}.
\end{align}
Consequently, the total success probability of the algorithm is strictly bounded from below by taking the complement
\begin{align}
    \Pr(\text{Success}) \ge 1 - (\delta_\xi + \delta_A + \delta_{\text{mix}}).
\end{align}
By symmetrically distributing the error budget such that $\delta_{\text{mix}} = \delta_A = \delta_\xi \leq \delta/3$, which requires setting the total variation threshold to $\mathcal{O}(\varepsilon_{\text{mix}}) = \mathcal{O}( \vert \tilde{V}(\tilde{G}) \vert^{-2} \varepsilon^2 \delta)$, we guarantee a total success probability of at least $1 - \delta$. Considering, by Thm. \ref{thm_mixingtime_clusterexp}, the costs of the mixing time $t_{\text{ {mix}}}\leq \mathcal{\tilde{O}}(\vert \tilde{V}(\tilde{G}) \vert \text{diam}(\tilde{G}) \varepsilon^{-1}\log\varepsilon_{\text{ mix}}^{-1}) \equiv  \mathcal{\tilde{O}}( \vert \tilde{V}(\tilde{G}) \vert \varepsilon^{-1} \text{diam}(\tilde{G})  )$, and the costs of sampling $\xi$ and annealing, the total runtime is

\begin{align}
            \mathcal{\tilde{O}}\left( C_w \cdot \vert \tilde{V}(\tilde{G}) \vert^3 \cdot \text{diam}(\tilde{G}) \cdot \varepsilon^{-3} \cdot \delta^{-1} + \vert V(G) \vert \right) ,
\end{align}
yielding an estimate of $\log Z(\mathcal{C}, w)$ with additive error $\varepsilon$ with probability of success $\geq 1-\delta$.

To exponentially improve this polynomial dependence on $\delta^{-1}$, we can instead run the base algorithm with a fixed constant confidence (e.g., bounding its failure probability by $1/3$, which removes the $\delta^{-1}$ factor from the runtime above). Running this constant-confidence algorithm $K = \mathcal{O}(\log \delta^{-1})$ times with independent randomness and returning the median of the outputs boosts the success probability back to any target $1 - \delta$. The argument is standard and we give it in App.\ref{app:amplification}. This yields a final total runtime of
\begin{align}
    \mathcal{\tilde{O}}\left(  C_w \cdot \vert \tilde{V}(\tilde{G}) \vert^3 \cdot \text{diam}(\tilde{G}) \cdot \varepsilon^{-3}  \cdot \log \delta^{-1} + \vert V(G) \vert \right),
\end{align}
concluding the proof.

\end{proof}

\subsection{An efficient Cluster Markov chain} \label{subsec::cluster_markov_chain}

The main requirement for the algorithm is to obtain samples from the target distribution $p(\Gamma)$. To achieve this, we construct the \textit{Cluster Markov chain} $S_{\text{MC}}$ with transition matrix $P_{\text{MC}}$, which explores the configuration space of connected clusters $\mathcal{G}_{\text{alg}}$.

\textit{Cluster dynamics:} Given a current cluster configuration $\Gamma_t$, we generate the next state $\Gamma_{t+1}$ by mutually exclusively applying one of the following transition rules:
\begin{itemize}
    \item \textit{1. Addition:} With probability $p_{\text{add}}=1/2$, we sample a vertex $v \in \text{supp}(\Gamma_t)$ uniformly, i.e.\ with probability $1/\vert\text{supp}(\Gamma_t)\vert$ and obtain a connected polymer using our local sampler $\nu_v(\gamma)$. We propose the new cluster $\Gamma_{t+1} = \Gamma_t \cup \{\gamma\}$ and accept it with the filter $f_{\text{add}} = m_\gamma(\Gamma_{t+1})/n_\gamma(\Gamma_{t+1})$, where $n_\gamma(\Gamma)$ is the number of polymers in a cluster $\Gamma$.
    
    \item \textit{2. Deletion:} With probability $p_{\text{del}}=1/2$, we sample a polymer $\gamma \in \Gamma_t$ with probability $\propto m_\gamma(\Gamma_t)/n_\gamma(\Gamma_t)$. We propose its removal, $\Gamma_{t+1} = \Gamma_t \setminus \{\gamma\}$, provided the resulting cluster remains connected. This move is accepted with the filter $f_{\text{del}} =
b_{\Gamma_{t+1}}(\gamma)/\vert \text{supp}(\Gamma_{t+1})\vert$, where
$b_\Gamma(\gamma)$ is the root count defined in Eq. \eqref{eq:rootcount} below. 
\end{itemize}
Crucially, the transition rules inherently confine the Markov chain to the algorithmic polymer model $\mathcal{P}_{\text{alg}}$ introduced in Sec. \ref{subsection_algorithmicpolymermodel}: a proposed polymer is rejected whenever its size exceeds $\norm{\gamma} > M$, whenever its weight falls below the threshold $\vert w_\gamma \vert \leq \kappa(\gamma)$, or whenever the resulting cluster would exceed $\norm{\Gamma_{t+1}} > 2M$. The state space is
thus rigorously truncated to $\mathcal{G}_{\text{alg}}$, which by Lem. \ref{lem:algorithmic_truncation} still gives an $\varepsilon$-additive approximation to $\log Z(\mathcal{C}, w)$.

A given polymer $\gamma$ may be proposed from any vertex of
$\text{supp}(\Gamma)$ that is incompatible with it. We therefore denote by
\begin{align}\label{eq:rootcount}
    b_\Gamma(\gamma) = \#\left\{ v \in \text{supp}(\Gamma) \;:\; \{v\} \nsim \gamma \right\}
\end{align}
the number of admissible roots, so that
$f_{\text{del}} \in [0,1]$. Note that $b_\Gamma(\gamma)$ counts incompatibility rather than overlap: two polymers whose supports are disjoint but joined by an edge of $G$ are incompatible, and the local sampler $\nu_v$ returns precisely
the polymers incompatible with $v$. 

Let $\Gamma' = \Gamma \cup \{\gamma\} \in \mathcal{G}_{\text{alg}}$. Then, we have
\begin{align}
    \frac{P_{MC}(\Gamma \rightarrow \Gamma')}{P_{MC}(\Gamma' \rightarrow \Gamma)}
    &= \frac{\frac{1}{2}\cdot \frac{b_{\Gamma}(\gamma)}{\vert \text{supp}(\Gamma)\vert}
       \cdot \vert w_\gamma \vert \cdot \frac{m_\gamma(\Gamma')}{n_\gamma(\Gamma')}}
      {\frac{1}{2} \cdot \frac{m_\gamma(\Gamma')}{n_\gamma(\Gamma')}
       \cdot \frac{b_{\Gamma}(\gamma)}{\vert \text{supp}(\Gamma)\vert}} \\
    &= \frac{p(\Gamma \cup \{\gamma\})}{p(\Gamma)} = \vert  w_\gamma \vert,
\end{align}
so this Markov chain is detailed-balanced and has $p(\Gamma)$ as the target
distribution. Furthermore, the chain is ergodic. It is intrinsically aperiodic because rejected proposals (such as attempting to delete a polymer that would disconnect the cluster) create self-loops, meaning $P_{\text{MC}}(\Gamma \rightarrow \Gamma) > 0$. It is also irreducible, as any valid cluster $\Gamma \in \mathcal{G}_C$ can transition to any other valid cluster $\Gamma' \in \mathcal{G}_C$ in a finite number of steps (e.g., see the canonical path designed in App. \ref{app_rapidmixingproof} which connects any pair of clusters with steps of the Cluster Markov chain) \cite{levin2017markovmixing}. Note that the double truncation of Def. \ref{def::algorithmic_polymer_model} is what makes
this possible: with clusters capped at $M$, a single polymer of size $M$ could
neither grow nor shrink, since the empty cluster is excluded from
$\mathcal{G}_C$.

By constructing $\Gamma_{t+1}$ directly from the local boundaries of $\Gamma_t$, we avoid an unnecessary and inefficient exhaustive exploration of the host graph $G=(V,E)$. This approach contrasts with most combinatorial sampling literature—such as generating independent sets or colorings via Glauber dynamics \cite{levin2017markovmixing,chen2021fastalgorithmslowtemperatures,blanca2023fast}—where transitions are defined by local updates at vertices chosen entirely independently of the current state. In our setting, transitions are inherently \textit{state-dependent} to strictly enforce the connectivity of the cluster. As previously motivated, this severe structural dependence is precisely what renders standard dynamic coupling techniques mathematically intractable \cite{blanca2023fast, chen2021fastalgorithmslowtemperatures, pathcoupling_BubleyDyer1997PathCoupling}, making the geometric routing perspective of canonical paths the natural framework to rigorously analyze the convergence of $S_{\text{MC}}$.

\begin{figure*}[t]
            \centering
            \includegraphics[width=\textwidth]{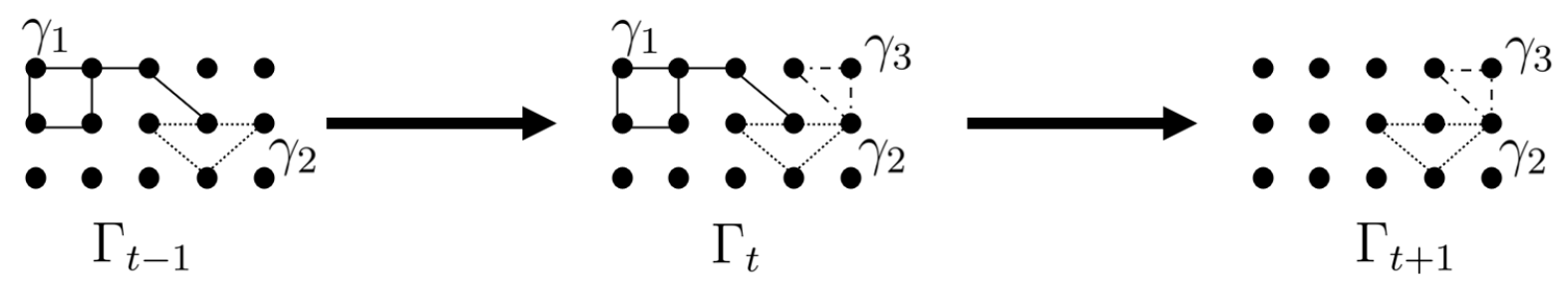}
            \caption{Illustration of the Cluster Markov chain $P_{MC}$. Here the clusters are depicted in the lattice $\Lambda$ rather that in the graph $G(V,E)$ as in Fig. \ref{fig:parte_a}. }
            \label{fig:ilustraion_MC}
\end{figure*}

With this, we can state the main result of this section, which guarantees that we can obtain cluster samples efficiently.
\begin{thm}
\label{thm_mixingtime_clusterexp}
Suppose that a polymer model $(\mathcal{C}, w, \sim)$ satisfies Con. \ref{con:kotecky_preiss_criteria} and Con. \ref{ass::uniform_local_weights}. Then the Cluster Markov chain $P_{\text{MC}}$ mixes in $\text{\emph{poly}}(\vert \tilde{V}(\tilde{G}) \vert, \varepsilon^{-1}, \log \varepsilon_{\text{mix}}^{-1} )$ runtime. More concretely 

\begin{align}
    t_{\text{mix}}(\varepsilon_{\text{mix}}) \leq \mathcal{\tilde{O}}( \vert \tilde{V}(\tilde{G}) \vert \varepsilon^{-1} \text{\emph{diam}}(\tilde{G})  ).
\end{align}

\end{thm}
\begin{proof}
We provide a sketch of the proof here; the complete mathematical analysis is deferred to App. \ref{app_rapidmixingproof}. The core of the argument relies on designing a unique, deterministic sequence of valid transitions—a canonical path $\mathcal{P}_{XY}$—connecting any initial cluster $X$ to any target cluster $Y$. Crucially, since our Cluster Markov chain strictly operates via single-polymer additions or deletions, this geometric routing must be constructed incrementally. We divide $\mathcal{P}_{XY}$ into three distinct phases:

\begin{itemize}
    \item \textit{Phase 1: Deconstruction}. We systematically reduce the initial cluster $X$, polymer by polymer, according to a deterministic ordering until it is reduced to a single, uniquely defined root polymer $\gamma_X \in X$. Due to the structural truncation bound $\vert X\vert \le \mathcal{O}(\log(\vert \tilde{V} \vert \varepsilon^{-1}))$, this phase requires at most $\mathcal{O}(\log(\vert \tilde{V} \vert \varepsilon^{-1}))$ steps.
    
    \item \textit{Phase 2: Translation (The Ladder).} From the root $\gamma_X$, we construct a unique sequence of local polymer additions and deletions—a geometric ``ladder''—that transports the state across the graph to the designated root polymer $\gamma_Y \in Y$ of the target cluster. Because this propagation is strictly bounded by the underlying topology of the host graph $G=(V,E)$, this translation requires at most $\mathcal{O}(\text{diam}(\tilde{G}))$ steps for any pair of roots.
    
    \item \textit{Phase 3: Construction.} Finally, we incrementally build the target cluster $Y$ from its root $\gamma_Y$ by sequentially adding polymers in a unique deterministic order, mirroring the reverse of the deconstruction process.
\end{itemize}

The maximum length of any canonical path is thereby robustly bounded as $L \le \mathcal{O}(\text{diam}(\tilde{G}) + \log(\vert  \tilde{V} \vert \varepsilon^{-1}))$. To successfully bound the total congestion $\bar{\rho}$ we employ   Con. \ref{con:kotecky_preiss_criteria} and Con. \ref{ass::uniform_local_weights}, yielding by Eq. \eqref{eq:mixing_time_bound}

\begin{align}
    t_{\text{mix}}(\varepsilon_{\text{mix}}) \leq \mathcal{\tilde{O}}( \vert \tilde{V}(\tilde{G}) \vert \varepsilon^{-1} \text{diam}(\tilde{G})).
\end{align}

\end{proof}

    \subsection{Efficient estimation of $A_M$ via annealing}\label{subsec:eff_estimat_annealing}

    With a protocol to sample from $p(\Gamma)$ established, evaluating the truncated cluster expansion $T_M$ requires computing the normalizing constant $A_M$ (See Eq.~\eqref{eq_Tmestimator}). We estimate this quantity via an annealing procedure. By introducing a continuous parameter $\theta \ge 0$, we define a family of deformed polymer weights
\begin{align}
    w_{\gamma}(\theta) &= w_{\gamma} \mathrm{e}^{-\theta \norm{\gamma}}
\end{align}
for all $\gamma \in \mathcal{C}$. This defines a parameterized polymer model $(\mathcal{C}, w(\theta))$ that strictly recovers the original model at $\theta = 0$. Because the deformed weights are monotonically decreasing with respect to $\theta$, any model satisfying Con. \ref{con:kotecky_preiss_criteria} at $\theta=0$ automatically satisfies it for all $\theta > 0$. 
Consequently, this weight deformation induces a family of $\theta$-dependent target distributions, denoted as $p_{\theta}(\Gamma)$, governed by the corresponding parameterized normalizing constant
\begin{align}
    A_M(\theta) &= \sum_{(\Gamma, m_\gamma) \in \mathcal{G}_{\text{alg}} }\prod_{\gamma \in \Gamma}\left\vert w_\gamma^{m_\gamma(\Gamma)}(\theta) \right\vert.
\end{align}
In standard annealing procedures, we establish a reference point—often called an \emph{anchor}—for our target function. This anchor serves as a known baseline, allowing us to use randomized algorithms (such as MCMC) to estimate the desired function relative to this reference. In our case we rewrite $A_M$ as

    \begin{align}
        A_M&= \frac{A_M( 0)}{A_M(1)} \cdot  \frac{A_M(1)}{A_M(2)} \dots  \frac{A_M(\ell-1)}{A_M(\ell)} \cdot A_M (\ell) \\&= \mathbb{E}[W] \cdot A_M (\ell)
    \end{align}
Now to obtain a proper estimator $\widehat{A_M}$ of $A_M$ we divide the problem into first calculating the anchor, $A_M(\ell)$, and then finding a random variable $W(\Gamma)$ such that $\mathbb{E}[W]$ reproduces the value of the telescopic product, and an efficient way of getting an estimation $\widehat{W}$.

If Con. \ref{con:kotecky_preiss_criteria} holds, Lem. \ref{lemma_Aextensive} reduces in this case to
\begin{equation}
  \left|A_M(\theta)-A_{M'}(\theta)\right|\;\le\;\Vt\,\mathrm{e}^{-(M'+1)\theta} .
  \label{eq:Lem16rho}
\end{equation}
By As. \ref{ass::uniform_local_weights} no cluster of non-vanishing weight has size below $r$, so the
lowest-order truncation is the $r$-th one and it is exactly
\begin{equation}
  A_r(\theta)\;=\;\mathrm{e}^{-r\theta}B_r .
  \label{eq:Aranchor}
\end{equation}
Taking $M'=r$ in Eq. \eqref{eq:Lem16rho} we may therefore approximate $A_M(\ell)$ by
its lowest-order term, with
\begin{equation}
  \left|A_M(\ell)-A_r(\ell)\right|\;\le\;\Vt\,\mathrm{e}^{-(r+1)\ell}
  \;\le\;\varepsilon_{\mathrm{anchor}}
  \label{eq:anchorerr}
\end{equation}
for any schedule length
\begin{equation}
  \ell\;\ge\;\frac{1}{r+1}\,
  \log\!\left(\Vt\,\varepsilon_{\mathrm{anchor}}^{-1}\right).
  \label{eq:ellchoice}
\end{equation}
By Eq. \eqref{eq:Aranchor} the anchor is the single scalar $B_r$ rescaled by a known
factor, so any estimate of $B_r$ transfers to $A_r(\ell)$ with the same relative
error. For finite degree graphs it can be calculated deterministically in $\mathcal{O}( \vert V(G)\vert)$. In the later case of long-range quantum models $\vert V(G)\vert$ grows $\propto N^k$, where $N$ is the system size and $k$ the maximum support of any $h_Z$ in the physical lattice. This step is a potential overhead that we will avoid later with probabilistic techniques (See Sec. \ref{subsec_eff_alg_lrpartfunctions}).

To estimate the ratios $\frac{A_M(i)}{A_M(i+1)}$, we define a sequence of independent random variables $W_i$ and let $W= \prod_i W_i$. Let $\Gamma$ be a cluster sampled from the cluster distribution $p_i(\Gamma)$ as done in Sec. \ref{subsec::cluster_markov_chain}. We define
\begin{equation}
    W_i = \prod_{\gamma \in \Gamma}  \left \vert \frac{w_{\gamma}(i)}{w_{\gamma}(i+1)} \right \vert .
\end{equation}
Note at this point that we need samples of the modified distribution $p_\theta (\Gamma)$ rather than of $p(\Gamma)$. Given our chain we can obtain them with no concerning overheads (See Lem. \ref{lem_app_AnnealingFilter}).

\begin{lemma}\label{lemm:Wi_expectedvalues_v2}
    For $0 \le i \le \ell - 1$, if Con. \ref{con:kotecky_preiss_criteria} is
    satisfied,
    \begin{align}
        \mathbb{E}_{p_{i+1}(\Gamma)}[W_i] &= \frac{A_M(i)}{A_M(i+1)}, \nonumber \\
        \mathbb{E}_{p_{i+1}(\Gamma)}[W_i^2] &= \frac{A_M(i-1)}{A_M(i+1)} .
    \end{align}
    Consequently,
    \begin{align}
        \mathbb{E}[W] = \frac{A_M(0)}{A_M(\ell)}
        \quad \text{and} \quad
        \mathbb{E}[W^2] = \frac{A_M(-1)A_M(0)}{A_M(\ell-1)A_M(\ell)} .
    \end{align}
\end{lemma}
\noindent The proof can be found in App. \ref{app_annealingproofs}.

Therefore, using Eqs. \eqref{eq:Brchain}, \eqref{eq:Aranchor} and
\eqref{eq:ellchoice},
\begin{equation}
  \mathbb E[W]=\frac{A_M(0)}{A_M(\ell)}
  \le\frac{1}{\wc}\,\mathrm{e}^{r\ell}
  \le\frac{1}{\wc}
  \left(\frac{\Vt}{\varepsilon_{\mathrm{anchor}}}\right)^{\!\frac{r}{r+1}} ,
  \label{eq:EW}
\end{equation}
while the second moment obeys the relative bound
\begin{equation}
  \frac{\mathbb E[W^2]}{\mathbb E[W]^2}
  =\frac{A_M(-1)}{A_M(0)}\cdot\frac{A_M(\ell)}{A_M(\ell-1)}
  \;\le\;\frac{1}{\wc}\;=\;\mathcal O(1),
  \label{eq:EW2}
\end{equation}
where $A_M(-1)\leq\sum_{\Gamma\in\GC}\prod_\gamma|w_\gamma^{m_\gamma(\Gamma)}|
\mathrm{e}^{\|\Gamma\|}\le \Vt$ follows from Lem. \ref{lemma_Aextensive}.

\begin{lemma}\label{lemma_estimationA}

Let $(\mathcal{C}(G), w)$ be a polymer model defined on a graph $G = (V, E)$ satisfying Con. \ref{con:kotecky_preiss_criteria}. For any target accuracy $\varepsilon_A > 0$ and failure probability $\delta_A > 0$, there exists an annealing algorithm that outputs an estimator $\widehat{A_M}$ for the normalizing constant $A_M$ such that

\begin{align}
    \Pr(\vert A_M - \widehat{A_M} \vert \ge \varepsilon_A) \le \delta_A.
\end{align}
The algorithm evaluates this estimator using an annealing schedule of length $\ell = \mathcal{O}(\log(\vert \tilde{V}(\tilde{G}) \vert \varepsilon_A^{-1}))$. The total computational complexity is bounded by a deterministic anchor evaluation running in $\mathcal{O}(\vert V(G) \vert)$ time, plus the generation of $ \ell \cdot n_A \leq \mathcal{\tilde{O}}\left( \vert \tilde{V}(\tilde{G}) \vert^2\varepsilon_A^{-2} \delta_A^{-1} \right)$ samples from the Cluster Markov chain.
    
\end{lemma}

\begin{proof}
By simulating the Cluster Markov chain (Sec.\ref{subsec::cluster_markov_chain})
  at $\theta=0$ and filtering its output as in Lem.\ref{lem_app_AnnealingFilter}, we
sample from each intermediate distribution $p_{i+1}(\Gamma)$ and construct
empirical estimators $\widehat W_i$ for each step of the cooling schedule. The
product of these independent estimates over the sequence $i=0,\dots,\ell-1$
yields an estimator $\widehat W$ that approximates the total ratio
$A_M(0)/A_M(\ell)\equiv\mathbb E[W]$.

The lowest-order mass, by contrast, is evaluated \emph{exactly}. Writing
\begin{equation}
  B_r\;=\!\!\!\sum_{\gamma\in\mathcal C\,:\,\norm{\gamma}=r}\!\!\!|w_\gamma|,
  \qquad
  A_r(\ell)\;=\;\mathrm e^{-r\ell}B_r ,
  \label{eq:Branchor}
\end{equation}
where $r=\mathcal O(1)$ denotes the minimal polymer size, the sum in
\eqref{eq:Branchor} is finite and we compute it by direct enumeration of the polymers
of size $r$. As their size is fixed, each of these weights costs $\mathcal O(1)$,
so the anchor is obtained in $\mathcal O(\vert V(G)\vert)$ time. Multiplying the annealed
ratio by this anchor provides our final estimator
\begin{equation}
  \widehat{A_M}\;=\;\widehat W\cdot\mathrm{e}^{-r\ell}\,B_r ,
  \label{eq:Ahat}
\end{equation}
so that the only two sources of error are the truncation of $A_M(\ell)$ at its
lowest order and the statistical error of $\widehat W$:
\begin{align}
  \left|A_M-\widehat{A_M}\right|
  &\le\left|A_r(\ell)-A_M(\ell)\right|\cdot\mathbb E[W]
  \label{eq:e1}\\
  &\quad+A_M(\ell)\left|\mathbb E[W]-\widehat W\right|
  \label{eq:e2}\\
  &\quad+\left|A_r(\ell)-A_M(\ell)\right|\cdot
  \left|\mathbb E[W]-\widehat W\right|
  \label{eq:e3}\\
  &\le\varepsilon_{\mathrm{anchor}}\mathbb E[W]
  +A_M(\ell)\varepsilon_W+\varepsilon_{\mathrm{anchor}}\varepsilon_W ,
  \label{eq:e4}
\end{align}
where $\varepsilon_{\mathrm{anchor}}:=\left|A_r(\ell)-A_M(\ell)\right|$ is given
by Eq.~\eqref{eq:anchorerr}.

We first pick $\varepsilon_{\mathrm{anchor}}$ such that
$\varepsilon_{\mathrm{anchor}}\mathbb E[W]\le\varepsilon_A/4$. Using Eq.
\eqref{eq:EW},
\begin{equation}
  \varepsilon_{\mathrm{anchor}}\,\mathbb E[W]\;\le\;
  \frac{1}{\wc}\,
  \varepsilon_{\mathrm{anchor}}^{\frac{1}{r+1}}\,\Vt^{\frac{r}{r+1}} ,
  \label{eq:epsanchor}
\end{equation}
so we choose
\begin{equation}
  \varepsilon_{\mathrm{anchor}}
  =\left(\frac{\wc\,\varepsilon_A}{4}\right)^{r+1}\Vt^{-r} ,
  \label{eq:epsanchorchoice}
\end{equation}
which by Eq. \eqref{eq:ellchoice} fixes the length of the annealing schedule at
$\ell=\mathcal O\!\left(\log\!\left(\Vt\varepsilon_A^{-1}\right)\right)$.

It remains to allocate the statistical budget. Setting
$\varepsilon_W=\varepsilon_A/(2A_M(\ell))$ bounds \eqref{eq:e2} by
$\varepsilon_A/2$, and since $\varepsilon_{\mathrm{anchor}}\le A_M(\ell)/2$ by
construction, \eqref{eq:e3} is bounded by $\varepsilon_A/4$; the three
contributions add to $\varepsilon_A$. Equivalently, because
$\mathbb E[W]=A_M/A_M(\ell)$, what the annealing must deliver is the
\emph{relative} precision
\begin{equation}
  \eta_W\;=\;\frac{\varepsilon_W}{\mathbb E[W]}\;=\;
  \frac{\varepsilon_A}{2A_M}\;=\;\Theta\!\left(\frac{\varepsilon_A}{\Vt}\right).
  \label{eq:etaW}
\end{equation}

It remains to fix the number of samples. Because the precision required of the
annealing is relative, Eq.~\eqref{eq:etaW}, the relative second-moment bound of
Eq.~\eqref{eq:EW2} controls the sample count directly. By
Lem.\ref{lemma:chebyshev},
\begin{align}
  \Prob\!\left(\left|\widehat W-\mathbb E[W]\right|
  \ge\eta_W\,\mathbb E[W]\right)
  &\;\le\;\frac{1}{n_A\,\eta_W^{2}}\,
  \frac{\mathbb E[W^{2}]}{\mathbb E[W]^{2}}
  \notag\\[2pt]
  &\;\le\;\frac{1}{\wc}\cdot\frac{1}{n_A\,\eta_W^{2}} ,
  \label{eq:cheb}
\end{align}
so a failure probability below $\delta_A$ is guaranteed as soon as
\begin{equation}
  n_A\;\ge\;\frac{1}{\wc}\,\frac{1}{\eta_W^{2}\,\delta_A}
  \;=\;\mathcal O\!\left(\Vt^{2}\varepsilon_A^{-2}\delta_A^{-1}\right),
  \label{eq:nA}
\end{equation}
where the last equality uses Eq.~\eqref{eq:etaW}. 

Collecting the three ingredients, the algorithm runs the Cluster Markov chain, draws $n_A$ independent copies of the product $W$ ---each of them
requiring $\ell$ annealed samples, and
evaluates the anchor $B_r$ deterministically. By Eqs.~\eqref{eq:e4},
\eqref{eq:epsanchorchoice} and \eqref{eq:nA}, this returns
$|A_M-\widehat{A_M}|\le\varepsilon_A$ with probability at least $1-\delta_A$.

\end{proof}

\section{Application: Sampling quantum  (long-range) thermal equilibrium}\label{sec::aplication}
Having established a general MCMC framework without an exponential sign problem for estimating abstract cluster expansions, we now demonstrate its power by evaluating the thermal properties of interacting quantum many-body systems. We focus on long-range quantum Hamiltonians satisfying
\begin{align}
    \sum_{Z \ni e=(i,j)} \norm{h_Z}^{1/2} \leq \left( \frac{g}{(1+d_e)^{\alpha}}\right)^{1/2},
    \label{eq:longrange_definition}
\end{align}
where $g>0$ and $\alpha>0$ control the strength and decay of the interactions. This condition is slightly stronger than the usual definition in the literature, which requires only $\sum_{Z \ni e} \norm{h_Z} \leq g/(1+d_e)^{\alpha}$ \cite{kuwahara2020longrange,Segovia_PRXQuantum2025}, but it should still cover all physically relevant models of long-range interactions. Specifically, it is satisfied by all 2-local long-range models with power-law decay
\begin{align}
    \norm{h_{ij}} \leq \frac{g}{(1+d_{ij})^{\alpha}},
    \label{eq:2local}
\end{align}
and any finite-range $k-$local model.
In particular, it covers the van der Waals and dipolar interactions that arise in Rydberg atom arrays and dipolar quantum gases \cite{chen_continuous_2023, schauss2012rydberg}, ion traps \cite{Richerme2014}, as well as most long-range spin models studied in the literature \cite{defenu_longrange}. Thus, while formally more restrictive for general $k$-body interactions, our assumption encompasses most of the relevant long-range Hamiltonians.

A geometric quantity which plays an important role in our calculations is
\begin{equation}
u = 2^\alpha \max_{i} \sum_{j\neq i} \frac{1}{1+d_{ij}^\alpha},
\label{eq:u}
\end{equation}
which is $\mathcal{O}(1)$ if $\alpha>D$.

\subsection{Convergence of the cluster expansion for long-range interactions}

As introduced in Sec. \ref{sec:quantumHamilton}, the partition function of a quantum system admits an abstract polymer model representation where the polymers are sets of overlapping Hamiltonian terms $h_Z$, or equivalently subgraphs in the interaction graph $G_{\text{int}}(V,E)$.

Recall that each vertex in $V(G_{\text{int}})$ corresponds to a Hamiltonian term $h_Z$. In the general case of a $k$-local long-range quantum Hamiltonian acting on $N$ qudits, the total number of terms can grow as $|V(G_{\text{int}})| \le \mathcal{O}(N^k)$. Notice, however, that the effective covering volume defined in Sec. \ref{subsection_algorithmicpolymermodel} can be tightly bounded as $|\tilde{V}(\tilde{G}_{\text{int}})| \le N$ by restricting to the finite range terms of the model. This distinction is crucial, as it leads to significantly better algorithmic scaling in our sampling schedule.

The primary hurdle is guaranteeing that this long-range polymer model remains strictly within the convergence regime of the cluster expansion. The following theorem ensures that if the physical interactions decay sufficiently fast, the standard Koteck\'y-Preiss criterion (Con. \ref{con:kotecky_preiss_criteria}) is satisfied at high temperatures.
	
\begin{lemma}\label{lem::koteckypreiss_longrange}
    Let $H$ be a quantum Hamiltonian as defined in Eq.~\eqref{eq:longrange_definition} such that $\alpha > D$. Then, Con. \ref{con:kotecky_preiss_criteria} is satisfied for $\beta \leq \beta^* = (2g)^{-1}\log(1+ ( 24\sqrt{\mathrm e}\,uk)^{-1})$, and the truncated cluster expansion converges as
    \begin{align}
        \vert\log Z_\beta - T_M \vert \leq N \left\vert \frac{\beta}{\beta^*} \right\vert^M.
    \end{align}
\end{lemma}
  \noindent The full proof, largely based on \cite{Segovia_PRXQuantum2025}, is deferred to App. \ref{app:kp_longrange}.

While the standard Koteck\'y-Preiss criterion guarantees the absolute convergence of the series, bounding the statistical variance and the mixing time of our MCMC algorithm requires the slightly more stringent Cluster Sampling condition (Con. \ref{con:clustersamplingcondition}). 

\begin{coro}\label{coro::clustersampling_longrange}
    Let $H$ be a quantum Hamiltonian as defined in Eq.~\eqref{eq:longrange_definition} such that $\alpha > 2D$. Then, Con. \ref{con:clustersamplingcondition} is satisfied for $\beta \leq \beta^{**} = (2g)^{-1}\log(1+ ( 24\sqrt{\mathrm e}\,uk)^{-2})$.
\end{coro}
 \noindent The proof is provided in App. \ref{app:clustersamplingcondition_longrange}.

\subsection{Efficient algorithm for quantum partition functions}\label{subsec_eff_alg_lrpartfunctions}
  
Building on the convergence guarantees established above, we now present the main algorithmic result of this section: a polynomial-time algorithm for evaluating long-range quantum partition functions at high temperatures.

    \begin{thm}\label{thm:algorithm_clusterexp_lr}
        Let us consider a quantum long-range Hamiltonian as introduced in Sec.\ref{sec:quantumHamilton}, with decay exponent $\alpha > kD$. Assuming the inverse temperature satisfies $\beta \leq \beta_{\text{sampling}} = c_\beta g^{-1}$ for constant $c_\beta = \mathcal{O}(1)$ (Eq.~\eqref{eq:def_betasampling}), there exists a randomized algorithm that outputs an estimate of $\log Z_\beta$ with additive error $\varepsilon$ and success probability $\geq 1-\delta$ in runtime bounded by

        \begin{align}
            \mathcal{\tilde{O}}(N^{3+1/D} \varepsilon^{-3}).
        \end{align}
    \end{thm}
\begin{proof}
    To apply our general cluster sampling framework (Thm. \ref{thm:sampling_cexp}) to the evaluation of quantum log-partition functions, we must first ensure that the underlying polymer model satisfies the weight regularity condition (Con. \ref{ass::uniform_local_weights}). This can be systematically achieved by applying a local energy shift to the Hamiltonian terms
    \begin{equation}
        h_Z \rightarrow h_Z - (2+ w_{Z})\norm{h_Z}\mathbb{I},
    \end{equation}
    where $w_{Z} > 0$ is a tunable $\mathcal{O}(1)$ constant, and $\mathbb{I}$ is the identity operator. Crucially, we explicitly define our covering set $\tilde{V}(\tilde{G})$ to consist exclusively of the $\Theta(1)$-norm local skeleton of the model (i.e., single-site fields and nearest-neighbor couplings). Because these local terms do not suffer from spatial decay, their norms are uniformly bounded away from zero ($\norm{h_{Z}} = \Omega(1)$). Consequently, this spectral shift guarantees that all single-term polymers $\gamma_v = \{Z_v\}$ associated with this local covering set possess absolute weights strictly bounded away from zero ($\vert w_{\gamma_v} \vert \geq w_{Z} \beta \norm{h_{Z_v}} = \Omega(1)$). Furthermore, since each vertex $v\equiv Z \in \tilde{V}$ uniquely corresponds to a distinct local interaction $Z_v$, these base polymers are mutually distinct ($\gamma_v \neq \gamma_u$ for $v \neq u$). This rigorous construction explicitly verifies and safely satisfies the lower-bound requirement of Con. \ref{ass::uniform_local_weights}. This transformation simply amounts to redefining the local coupling constant bound $g \rightarrow g(3 + w_Z)$ within the subsequent convergence criteria, with its subsequent threshold temperature shift. 

    With these theoretical guarantees in place, the general MCMC framework of Thm.\ref{thm:sampling_cexp} establishes the mathematical foundation for a randomized algorithm. However, realizing the strictly optimized polynomial runtime stated above requires overcoming two significant computational bottlenecks inherent to long-range interacting graphs:

    \begin{enumerate}
        \item We need a local polymer sampler that operates in $\mathcal{O}(1)$ expected time despite the dense long-range tails. We develop this tool using a radial importance sampling strategy in Sec. \ref{subsec::Polymer_sampler}.
        \item The deterministic evaluation of the annealing anchor $A_r(\ell)$ naively requires $\mathcal{O}(\vert V(G)\vert) = \mathcal{O}(N^k)$ operations. This dense scaling would act as an artificial bottleneck, compromising the strict polynomial efficiency of the overall algorithm. We bypass this overhead by introducing a randomized anchor estimator (detailed in Sec. \ref{subsec::anchorestimation}) that runs in $\mathcal{O}(N^2 \varepsilon_{\text{anchor}}^{-1} \delta_{\text{anchor}}^{-1})$ expected time, with a failure probability bounded by $\delta_{\text{anchor}}$. As established in the proofs of Lem. \ref{lemma_estimationA} and Thm. \ref{thm:sampling_cexp}, the target precision for this estimator must scale as $\varepsilon_{\text{anchor}} = \mathcal{O}(\varepsilon_A^2 N^{-1}) \equiv \mathcal{O}(\varepsilon^2 N^{-1})$. The introduction of this randomized subroutine adds a new failure mode to the algorithm. However, this is seamlessly absorbed into the global success probability analysis (Eq. \eqref{eq_proof_ProbSuccess}) by simply setting $\delta_{\text{anchor}} = \mathcal{O}(\delta)$, which secures the global target success rate without inducing any additional asymptotic overhead. This yields a total cost of $\mathcal{\tilde{O}}(N^3 \varepsilon^{-2} \delta^{-1})$ for the anchor calculation. Note that for the 2-local case a deterministic $\mathcal{O}(N^2)$ calculation is both simpler and more efficient.
    \end{enumerate}

    Beyond overcoming these local subroutines, the global mixing time of the algorithm relies on the geometric routing of the Markov chain. By explicitly identifying the connected covering subgraph $\tilde{G}(\tilde{V}, \tilde{E})$ with the underlying local model of the physical system, its topology is strictly that of a $D$-dimensional lattice. This physical equivalence allows us to bound the abstract routing diameter as $\text{diam}(\tilde{G}) = \mathcal{O}(N^{1/D})$. 
    
    Substituting these three bounds into Thm.\ref{thm:sampling_cexp} yields
$\mathcal{\tilde{O}}(N^{3+1/D} \varepsilon^{-3} \delta^{-1})$, so the anchor
contribution $\mathcal{\tilde{O}}(N^{3} \varepsilon^{-2} \delta^{-1})$ being
subleading. The anchor estimator is part of the base algorithm, so the median
amplification of Thm.\ref{thm:sampling_cexp} applies unchanged and replaces
$\delta^{-1}$ by $\log \delta^{-1}$, concluding the proof.
\end{proof}

\subsection{Overcoming volume overheads via radial importance sampling}
\label{subsec:importance_sampling}

While our Markov chain framework successfully avoids the exhaustive global enumeration of clusters, applying it directly to long-range quantum systems presents us with a difficulty: their dense connectivity means that even generating a single \textit{local} polymer sample, or evaluating the first-order anchor $A_r$, becomes inefficient if one must exhaustively explore the entire interaction space. 

By sampling interaction distances from a heavy-tailed distribution
$\Pr(R) \propto R^{-\tilde{\alpha}}$, we can effectively restrict the local degree
to a bounded quantity $\Delta_R = \mathcal{O}(R^{D(k-1)})$ (Eq. \eqref{eq:DeltaR}), which for two-body
interactions reads $\mathcal{O}(R^{D})$. The radial exponent $\tilde{\alpha}$ is
a free algorithmic parameter, denoted with a tilde to distinguish it from the
physical decay exponent $\alpha$. It must lie in an admissible window, specified
in App. \ref{app:longrangesampler}, which is non-empty precisely when $\alpha > kD$. Every
subroutine below uses this same choice.

This dimensional reduction allows us to generate exact local samples without fully exploring the dense, all-to-all tails of the system, and is the key to converting deterministic inefficient searches into efficient randomized estimators that run in polynomial time.

\subsubsection{A Long-Range Polymer Sampler}\label{subsec::Polymer_sampler}
In order to apply Thm. \ref{thm:sampling_cexp} we need a local polymer sampler.
We construct one by instantiating the radial importance sampling of
Sec. \ref{subsec:importance_sampling}: the first step of the algorithm draws a
maximum interaction length as
\begin{equation}
    \Pr(R) = \frac{R^{-\tilde{\alpha}}}{\zeta(\tilde{\alpha})},
    \label{eq:radial_distribution}
\end{equation}
and the exploration is then confined to terms of spatial extent at most $R$.
With it, we show the following.

\begin{lemma}\label{lem:longrange_polymer_sampler}
		Given a Hamiltonian as in Sec. \ref{sec:quantumHamilton}, with $\alpha > kD$, and $\beta \leq \beta_{\text{sampling}}$, such that Lem. \ref{lem::koteckypreiss_longrange} holds, there exists a randomized algorithm for (perfect) sampling of polymers according to $\nu_{v} (\gamma)$ in runtime $\mathcal{O}(1)$.
	\end{lemma}

\noindent Both the algorithm and the proof is shown in App. \ref{app:longrangesampler}, but we sketch it here.

The scheme is inspired by the finite-range sampler of \cite{blanca2023fast},
which grows a polymer by a breadth-first search visiting the neighbours of each
explored term uniformly. That construction relies on having a bounded degree $\Delta$, which
keeps both the cost of one step and the explored boundary under control, and it
samples polymers containing a prescribed vertex. Neither survives in the
long-range regime, where every term interacts with the whole system and the
boundary of any polymer is the entire graph.

We overcome this with three main upgrades. First, we draw a maximum interaction radius $R$ from
distribution given by Eq. \eqref{eq:radial_distribution} and restrict the search to terms of spatial extent at most $R$, since the
resulting subgraph has bounded degree for each fixed $R$. Second, for a given $R$ the
exploration is done as follows: each candidate term is accepted with a
probability proportional to its interaction strength, so that weak and distant terms are visited correspondingly rarely. Third, we grow the polymer from a seed incompatible with $v$, so that the output is a polymer connected to $v$ (rather
than a polymer which contains $v$ it as done in \cite{blanca2023fast}), which is what the cluster chain requires.

Each of these biases the law of the generated polymer, and a final acceptance
filter is designed to undo the three of them at once: the bias of the weighted
search, the multiplicity of seeds that could have produced the same polymer, and
the conditioning on the radius. Once these corrections are applied the output
law is exactly $\nu_v(\gamma)$, so the sampler is perfect rather than
approximate.

Two requirements then constrain the parameters: the filter must be a valid
probability, and the expected cost of the search must be finite. Both of these impose constraints on the temperature scale $\beta \leq \beta_{\text{sampling}}$ and the decay exponent, requiring $\alpha > kD$.

\subsubsection{Randomized Anchor Estimation}\label{subsec::anchorestimation}
Having established that the local polymer updates and weight computations require strictly $\mathcal{O}(1)$ expected time, the remaining computational bottleneck of the algorithm lies in the global annealing procedure (Sec. \ref{subsec:eff_estimat_annealing}). Specifically, to calibrate the annealing parameters, one must evaluate the global anchor quantity $A_r(\ell)$, which involves a sum over all possible Hamiltonian terms $h_Z$ of the system. Since an exact deterministic computation of this global sum is already expensive, we introduce a randomized estimation strategy.
\begin{lemma}\label{lemma::longrange_anchor}
    Let $\alpha \geq kD$. There exists a randomized algorithm that computes an estimator $\widehat{A_r}(\ell)$ of the anchor of the annealing procedure of Sec. \ref{subsec:eff_estimat_annealing}

    \begin{align*}
        A_r (\ell) = \mathrm{e}^{-\ell} \sum_{Z} \vert w_Z \vert,
    \end{align*}
such that $\vert A_r (\ell) -  \widehat{A_r} (\ell) \vert \leq \varepsilon_{\text{anchor}}$ with probability at least $1-\delta_{\text{anchor}}$, in expected runtime $\mathcal{O}(N^2 \varepsilon_{\text{anchor}}^{-1} \delta_{\text{anchor}}^{-1})$.
    
\end{lemma}

\begin{proof}
    Let us recall the radial distribution $\Pr(t) \propto t^{-\tilde{\alpha}}$ of Eq. \eqref{eq:radial_distribution} and the random variable $S(t)= \Pr(t)^{-1}\, \mathrm{e}^{-\ell}\sum_{Z: d(Z) = t} \vert w_Z \vert$, where we recall the definition $d(Z)= \max_{(i,j)}d_{ij}$. Note the following

    \begin{align}
        \mathbb{E}_{\Pr(t)}[S(t)] = \sum_t  \mathrm{e}^{-\ell}\sum_{Z: d(Z) = t} \vert w_Z \vert = A_r (\ell).
    \end{align}
    From this observation we develop a randomized protocol to compute $A_r (\ell)$. Recall Sec. \ref{subsec:eff_estimat_annealing}, to see whether we can estimate $S(t)$ efficiently by sampling from $\Pr(t)$ with success probability $\geq 1-\delta_{\text{anchor}}$

    \begin{align}\label{eq::chebysev_anchorlr}
        \Pr[ & \vert A_r(\ell) -  \widehat{A_r} (\ell) \vert \geq \varepsilon_{\text{anchor}}  ] \\
        &\leq \frac{ \mathbb{E}_{\Pr(t)}[S^2(t)]}{n_{\text{anchor}} \cdot \varepsilon_{\text{anchor}}^2} \leq \delta_{\text{anchor}}.
    \end{align}
The second moment of the estimator is bounded by
\begin{align}\label{eq:AlphaTildeConstrain3}
    \mathbb{E}_{\Pr(t)}[S^2(t)] &\leq \mathcal{O} \left( \sum_t \mathcal{O}(t^{\tilde{\alpha}}) \cdot  N \varepsilon_{\text{anchor}}  \frac{t^{2(D-1)}}{(1+t)^{2\alpha}} \right) \nonumber \\
    &\leq \mathcal{O} \left( N \varepsilon_{\text{anchor}} \sum_t t^{\tilde{\alpha}-2(\alpha+1-D)}  \right) \nonumber \\
    &\leq \mathcal{O}(N \varepsilon_{\text{anchor}}).
\end{align}
The sum over $t$ converges whenever $\tilde{\alpha} < 2(\alpha+1-D)-1$, which is implied by the admissible window of Sec. \ref{subsec:importance_sampling}. The first inequality in Eq. \eqref{eq:AlphaTildeConstrain3} relies on the following upper bound for $S(t)$
\begin{align}
    S(t) &\leq \mathcal{O}(t^{\tilde{\alpha}}) \left(\mathrm{e}^{\beta g} -1 \right) \mathrm{e}^{-\ell} \sum_{Z: d(Z) = t} \frac{\norm{h_Z}}{g} \nonumber \\
    &\leq \mathcal{O}(t^{\tilde{\alpha}}) \left( \mathrm{e}^{\beta^{**} g} -1 \right) \mathrm{e}^{-\ell} \sum_{i \in \Lambda} \sum_{j : d_{ij}=t} \sum_{Z \ni (i,j)} \frac{\norm{h_Z}}{g} \nonumber \\
    &\leq \mathcal{O}(t^{\tilde{\alpha}}) \left(\mathrm{e}^{\beta^{**} g} -1 \right) \mathrm{e}^{-\ell} N t^{D-1} \cdot \frac{1}{(1+t)^{\alpha}} \nonumber \\
    &\leq C'\mathcal{O}(t^{\tilde{\alpha}}) \varepsilon_{\text{anchor}}^{1/2} N^{1/2}  \frac{t^{D-1}}{(1+t)^{\alpha}}, 
\end{align}
where, for a fixed reference site $i \in \Lambda$, we upper-bounded the number of target sites $j$ at distance $d(Z) = t$ by its geometric surface scaling ($\propto t^{D-1}$), with $C'$ being a positive geometric constant.

Operationally, at each randomized step, the algorithm must deterministically evaluate $S(t)$. The expected runtime of this subroutine is
\begin{align}
    \mathbb{E}_{\Pr(t)}[\text{time to compute } S(t)] \leq& \mathcal{O}\left( N \sum_t \Pr(t) t^{D(k-1)} \right) \nonumber \\ \leq& \mathcal{O}(N).
\end{align}
Here, enumerating the $k$-local terms with $d(Z)=t$ containing a fixed site requires choosing the remaining $k-1$ sites inside the ball of radius $t$, giving $\mathcal{O}(t^{D(k-1)})$ terms, the same count that bounds $\Delta_R$. The sum therefore converges for $\tilde{\alpha} > D(k-1)+1$, which is precisely the lower end of the admissible window. 

By substituting the variance bound into Chebyshev's inequality (Eq. \eqref{eq::chebysev_anchorlr}), we find that the required number of samples is
\begin{align}
    n_{\text{anchor} } \geq \mathcal{O}(N \varepsilon_{\text{anchor}}^{-1} \delta_{\text{anchor}}^{-1}).
\end{align}
Multiplying this sample count by the expected cost per sample yields the total algorithmic complexity for the randomized anchor estimation
\begin{align}
    \mathbb{E}&[\text{runtime of randomized anchor}] \nonumber \\
    &\leq n_{\text{anchor} } \cdot \mathbb{E}[\text{time to compute } S(r)] \nonumber \\
    &\leq \mathcal{O}(N^2\varepsilon_{\text{anchor}}^{-1} \delta_{\text{anchor}}^{-1}),
\end{align}
which concludes the proof.
\end{proof}

\subsection{Application to finite-range interactions}
\label{subsec:finite_range}
It is important to emphasize that our general framework for long-range Hamiltonians inherently encompasses standard short-range quantum models as a special case. While the machinery of radial importance sampling is strictly necessary to overcome the dense connectivity of long-range interactions, our cluster sampling algorithm simplifies significantly when restricted to the short-range regime. 

Consider a $k$-local Hamiltonian on a $D$-dimensional lattice $G$, where interactions are strictly bounded by a finite cut-off distance (e.g., nearest-neighbor interactions). In this regime, the interaction graph has a strictly bounded degree $\Delta = \mathcal{O}(1)$, and the effective volume scales trivially as $|\tilde{V}(\tilde{G})| \propto N$. 
Consequently, the radial importance sampling subroutines are no longer required:
\begin{enumerate}
    \item The deterministic computation of the anchor $A_r(\ell)$ requires only $\mathcal{O}(N \cdot \Delta) = \mathcal{O}(N)$ runtime.
    \item The local polymer sampler can perfectly generate subgraphs in $\mathcal{O}(1)$ time by simply exploring the local finite-degree neighborhood, matching the structure of previous works \cite{blanca2023fast}.
\end{enumerate}

Applying Theorem \ref{thm:sampling_cexp} to this simplified setting immediately yields the following formal result:

\begin{coro}[Complexity for short-range quantum models]
\label{coro:short_range_complexity}
For any finite-range quantum system defined on a $D$-dimensional lattice at high temperatures ($\beta \le  \beta_{\text{sampling}} \equiv \mathcal{O}(1)$), there exists a randomized algorithm that outputs an estimate of $\log Z_\beta$ with additive error $\varepsilon$ with expected runtime

\begin{align}
     \mathcal{\tilde{O}}\left( N^{3+1/D} \varepsilon^{-3} \right).
\end{align}

\end{coro}
In this case, the improvement we achieve is that of a fixed polynomial scaling independent of the degree, unlike exhaustive ennumeration approaches \cite{Mann2021,Haah_2024}.

\section{Further applications}
\label{sec_futherapplications}
The randomized approach considered here strictly improves upon the structural restrictions of prior deterministic frameworks to compute cluster expansions, such as those in~\cite{mann2024algorithmic,Mann2021,zlokapa_rigorous_2026}. By replacing exhaustive enumeration with spatial importance sampling, we completely bypass the $\left(\vert V(G) \vert \varepsilon^{-1}\right)^{\mathcal{O}(\log \Delta)}$ computational bottlenecks, allowing these methods to extend efficiently to dense and long-range regimes. Thus, our results have applications in any domain where polymer models with non-positive weights appear.  Immediate applications where this is the case include:

\begin{itemize}
    \item \textit{Tensor-network contractions and beyond-BP corrections.}
Tensor networks are among the most widely used tools for the classical
    simulation of quantum many-body systems, and their reach is largely set by the cost of contracting them. Belief propagation (BP) provides a cheap approximate contraction whose error is controlled by the contribution of closed loops, and recent works have shown that loop and cluster expansions systematically correct BP and, in suitable regimes, provide analytic guarantees for approximate tensor-network contractions~\cite{midha2025beliefpropagationclustercorrectedtensor,evenbly2026loopseries,gray2026loopclusterexpansions}.

    Because these loop-based corrections naturally form a cluster expansion, our sampling framework applies to them directly. By doing so, we completely avoid the exponential dependence on the underlying graph degree $\Delta$. This is a relevant structural advantage: BP is primarily deployed in complex geometries—such as three-dimensional lattices, random graphs, and dense circuit architectures—where standard contraction techniques are unfeasible. Consequently, deterministic enumeration might become a bottleneck exactly where BP techniques are most needed. Our cluster sampling approach directly avoids this algorithmic limitation, expanding the computational reach of the method. This is conditional on the expansion converging.
    
   A previously considered approach is whether Monte Carlo can be used for the contraction of the tensors themselves ~\cite{schuch2008mctensorcontractions}. The difficulty of doing that is then
    governed by the sign structure of its entries ~\cite{chen2025signproblemtn}. Our framework suggests a
    complementary route: whenever the loop cluster expansion around the BP fixed
    point converges absolutely, cluster sampling estimates the correction in
    provable polynomial time and with a variance controlled by the convergence
    criterion itself, rather than by the positivity of the tensor entries. Closest to our proposal is the recent Ref.~\cite{sim2026stochasticloopcorrectionsbelief}, which also estimates loop corrections to the partition function via stochastic sampling. Their method applies to pairwise Markov random fields with symmetric edge potentials (e.g. the 2D ferromagnetic Ising model), where the partition function admits a particularly simple loop expansion. They identify extending the approach to more general systems as an open problem and note that their scheme can suffer from a sign problem. In contrast, our method comes with analytical guarantees whenever the cluster expansion converges, independently of the sign structure of the Hamiltonian.

    \item \textit{Output of short-depth circuits.} The first example of this is the usage of Barvinok's method in \cite{Bravyi_2021} to compute expectation values of extenside observables close to the identity, phrased in the language of polymer models in \cite{mann2024algorithmic}. Additionally, very recent work also writes the sampling probabilities of both noisy and noiseless circuits within this framework \cite{nelson2026polynomialtimeclassicalsamplernoisy,coudron2026polynomialtimeadditiveerrorestimationoutput}. With our approach, the polynomial runtime of the approximation schemes might not depend on the circuit depth or other parameters as long as the convergence condition holds.

    \item \textit{Short-time quantum dynamics.}
    Cluster expansions have been used to simulate quantum dynamics at short times ($t \leq t^*\equiv \mathcal{O}(1)$) with rigorous error control~\cite{Wild2023}. Our framework naturally extends to such settings, and the same sampling machinery can be adapted to evaluate complex amplitudes.

    \item \textit{Weakly-interacting systems at arbitrary temperatures.} 
    Deterministic cluster expansions have established fully polynomial-time approximation schemes for weakly-interacting quantum spin systems~\cite{mann2026efficientalgorithmsweaklyinteractingquantum}.
    
    \item \textit{Low-temperature quantum perturbations.} 
    Similarly, deterministic expansions combined with contour representations can approximate partition functions of low-temperature quantum spin systems that act as stable perturbations of classical models~\cite{Helmuth_2023}.

      \item \textit{High-temperature long-range bosonic systems.}
    Deterministic cluster-expansion algorithms have recently been developed for long-range bosonic Hamiltonians at sufficiently high temperatures~\cite{tong2026complexityclassificationhightemperaturebosons}. We conjecture that our stochastic cluster sampling approach can similarly improve the complexity and break quasipolynomial gaps in these settings.
\end{itemize}

\section{Conclusions}
\label{sec::conclusions}

A primary open question is whether our algorithmic framework can be extended to the long-range regime $\alpha > D$. We anticipate that this is fundamentally achievable, given that the standard Koteck\'y--Preiss criterion remains valid at sufficiently high temperatures in this regime~\cite{Segovia_PRXQuantum2025}. Currently, reaching this threshold is hindered by two technical bottlenecks: our Markov chain dynamics strictly require the more stringent Con.\ref{con:clustersamplingcondition}, and our spatial polymer sampler intrinsically requires $\alpha > kD$ (while the full $\alpha > D$ is achieved in \cite{putterman2026quantumthermalstateslook}). We conjecture that by modifying the transition rules of the chain to explicitly exploit the combinatorial structure of the Ursell functions, one could construct an algorithm whose convergence depends solely on the relaxed Con.\ref{con:kotecky_preiss_criteria}.  

Our approach shares deep conceptual similarities with Diagrammatic Monte Carlo (DiagMC) \cite{DiagMC_Van_Houcke_2010}. Both paradigms avoid sampling physical configurations with potentially negative or complex weights, and instead directly sample terms in an analytical expansion: DiagMC typically targets Feynman diagrams, often to simulate highly correlated fermionic systems, whereas our method evaluates polymers on a discrete lattice. Although DiagMC is highly efficient in practice, rigorous performance guarantees remain limited: while there are partial analytical convergence results for specific diagrammatic schemes and regimes \cite{Rossi_2017, Rossi_PhysRevLett.121.130405, RossiDiagMC_PhysRevLett.119.045701}, to the best of our knowledge general bounds on Markov chain mixing times or convergence radii are still lacking~\cite{chen2025convergencecumulantexpansionpolynomialtime}. The mathematical tools developed here for rigorously bounding cluster MCMC mixing times, as well as those of related works for fermions~\cite{chen2025convergencecumulantexpansionpolynomialtime}, could offer a way to establish formal convergence regimes in DiagMC algorithms.

From the practical side, an immediate next step is to benchmark this randomized cluster-sampling engine against established numerical techniques, such as linked-cluster expansions (NLCE). While NLCE evaluates finite-temperature and dynamical observables in the thermodynamic limit by deterministically enumerating and exactly solving finite subgraphs~\cite{nlce_review,QuantumDynamicsNLCE}, our stochastic approach bypasses the cost of exact enumeration. It would be interesting to benchmark these methods against one another, particularly for long-range models~\cite{chan2026efficientclassicalsimulationtwodimensional, Kai_PhysRevB.94.075156, Kai_duft2026exactfastseriesexpansions}. 

Our framework introduces an alternative variation to previous QMC approaches. Many of these QMC methods rely on sampling physical configurations—such as imaginary-time worldlines or stochastic series expansion (SSE) operator strings—and thus frequently succumb to an exponential sign problem when the underlying weights lack positivity \cite{troyer2005computational,signpbm_Pan_2024}. Our approach operates entirely within an abstract expansion space. By reformulating quantum partition functions and observables as convergent cluster expansions, we sample \emph{clusters} rather than physical states. In this regime, the analytic convergence criterion inherently bounds the variance, circumventing exponential blow-ups and systematically yielding an efficient MCMC algorithm even in the presence of non-positive weights. This hints at a natural pathway to tackle quantum problems conjectured to be classically tractable, such as Quantum Max-Cut on bipartite graphs \cite{rayudu2026spectralgapleeyanghamiltonians,wong2026leeyangtensorshamiltoniancomplexity, mann2026fastalgorithmsstoquasticspin}, by mapping them onto convergent expansions.

\textit{Note added-} Recent concurrent work of Putterman, Zlokapa, and Cotler~\cite{putterman2026quantumthermalstateslook} achieves related results for the computation of the partition function of long-range Hamiltonians at high temperature. They use a direct importance-sampling estimator of combinatorial structures that successfully reaches the $\alpha > D$ regime, and is particularly efficient for Pauli Hamiltonians. They achieve a polynomial scaling of the runtime $\mathcal{O}(N^3/\varepsilon^3)$, compared to our $\mathcal{O}(N^{3+1/D}/\varepsilon^3)$. 

Our contribution establishes a basis-independent algorithmic framework by constructing a Markov chain directly on the abstract polymer space.
The main advantage of this is that our approach works for general polymer models, beyond high-temperature partition functions, and thus has a number of further applications (see Sec. \ref{sec_futherapplications}). 
The price to pay is that our algorithm is restricted to work at $\alpha > k D$ in the long-range setting. This is required here for the estimator to have a bounded variance, and for the polymer sampler to be well-defined. Future work will be targeted at addressing that gap, potentially combining ideas from both frameworks.

\begin{acknowledgments}
The authors thank Aram Harrow and Cambyse Rouz\'e for very useful suggestions. 
The work of J.S.S. was supported by the Spanish Research Council, CSIC, through the JAE-PREDOC 2024 program.
A.M.A.  acknowledges support from the Spanish Agencia Estatal de Investigacion through the grants  ``PID2023-150847NA-I00'', ``EUR2025-164823'' and ``PCI2024-153448'', funded by
MICIU/AEI/10.13039/501100011033 and co-funded by the EU. 
This project was funded within the QuantERA II Programme that has received funding from the EU’s H2020 research and innovation programme under the GANo. 101017733.
This publication has also been funded within the framework of the R\&D\&I Project CEX2025-001574-S, funded by MICIU/AEI/10.13039/501100011033. The research presented in this publication falls within the research line ``Quantum Matter and Fields''.

\textbf{AI Usage Acknowledgement} 
The core of the project, and the main ideas and proofs (such as using a modified Kotecky-Preiss condition to bound a Monte-Carlo estimator) were conceived and initially developed without AI usage. Both Claude Opus 5 and ChatGPT-6 Astra were used for polishing technical definitions and proofs of many of the intermediate lemmas, and in particular for the analysis of the mixing time of the Markov chain. They were also used to construct the counterexample of App. \ref{app:sign_problem}. Both models were also used in editing and revising the text. We thank Yair Lavi, who with help of AI pointed to two previous innacuracies in Lemmas \ref{lemma_Aextensive} and \ref{lem:longrange_polymer_sampler} in v1.
\end{acknowledgments}

	\bibliographystyle{unsrturl}
	\bibliography{bibliography}

\begin{thebibliography}{10}

\bibitem{troyer2005computational}
Matthias Troyer and U.-J. Wiese.
\newblock Computational complexity and fundamental limitations to fermionic
  quantum monte carlo simulations.
\newblock {\em Physical Review Letters}, 94(17):170201, 2005.
\newblock \href {https://doi.org/10.1103/PhysRevLett.94.170201}
  {\path{doi:10.1103/PhysRevLett.94.170201}}.

\bibitem{loh1990sign}
E.~Y. Loh, J.~E. Gubernatis, D.~J. Scalapino, S.~R. White, R.~T. Scalettar, and
  R.~L. Sugar.
\newblock Sign problem in the numerical simulation of many-electron systems.
\newblock {\em Physical Review B}, 41(13):9301--9307, 1990.
\newblock \href {https://doi.org/10.1103/PhysRevB.41.9301}
  {\path{doi:10.1103/PhysRevB.41.9301}}.

\bibitem{kaul2013sign}
Rahul~K. Kaul and Anders~W. Sandvik.
\newblock The sign problem in quantum monte carlo simulations.
\newblock {\em Physics Reports}, 603:1--44, 2015.
\newblock \href {https://doi.org/10.1016/j.physrep.2015.08.001}
  {\path{doi:10.1016/j.physrep.2015.08.001}}.

\bibitem{zheng2018sign}
B.-X. Zheng et~al.
\newblock The sign problem in quantum monte carlo simulations.
\newblock {\em Journal of Physics: Condensed Matter}, 30(47):473001, 2018.
\newblock \href {https://doi.org/10.1088/1361-648X/aae5e3}
  {\path{doi:10.1088/1361-648X/aae5e3}}.

\bibitem{BravyiTerhal}
Sergey Bravyi and Barbara Terhal.
\newblock Complexity of stoquastic frustration-free hamiltonians.
\newblock {\em SIAM Journal on Computing}, 39(4):1462--1485, 2010.
\newblock \href {https://arxiv.org/abs/https://doi.org/10.1137/08072689X}
  {\path{arXiv:https://doi.org/10.1137/08072689X}}, \href
  {https://doi.org/10.1137/08072689X} {\path{doi:10.1137/08072689X}}.

\bibitem{bravyigosset}
Sergey Bravyi and David Gosset.
\newblock Polynomial-time classical simulation of quantum ferromagnets.
\newblock {\em Phys. Rev. Lett.}, 119:100503, Sep 2017.
\newblock URL: \url{https://link.aps.org/doi/10.1103/PhysRevLett.119.100503},
  \href {https://doi.org/10.1103/PhysRevLett.119.100503}
  {\path{doi:10.1103/PhysRevLett.119.100503}}.

\bibitem{Crosson2021rapidmixingofpath}
Elizabeth Crosson and Aram~W. Harrow.
\newblock Rapid mixing of path integral {M}onte {C}arlo for 1{D} stoquastic
  {H}amiltonians.
\newblock {\em {Quantum}}, 5:395, February 2021.
\newblock \href {https://doi.org/10.22331/q-2021-02-11-395}
  {\path{doi:10.22331/q-2021-02-11-395}}.

\bibitem{Crosson2025classicalsimulation}
Elizabeth Crosson and Samuel Slezak.
\newblock Classical {S}imulation of {H}igh {T}emperature {Q}uantum {I}sing
  {M}odels.
\newblock {\em {Quantum}}, 9:1788, July 2025.
\newblock \href {https://doi.org/10.22331/q-2025-07-09-1788}
  {\path{doi:10.22331/q-2025-07-09-1788}}.

\bibitem{rayudu2025fastmixingoperatorlooppathintegral}
Chaithanya Rayudu and Jun Takahashi.
\newblock Fast mixing of operator-loop path-integral quantum monte carlo for
  stoquastic xy hamiltonians, 2025.
\newblock URL: \url{https://arxiv.org/abs/2509.21683}, \href
  {https://arxiv.org/abs/2509.21683} {\path{arXiv:2509.21683}}.

\bibitem{Marvian2019}
Milad Marvian, Daniel~A. Lidar, and Itay Hen.
\newblock On the computational complexity of curing non-stoquastic
  hamiltonians.
\newblock {\em Nature Communications}, 10(1):1571, Apr 2019.
\newblock \href {https://doi.org/10.1038/s41467-019-09501-6}
  {\path{doi:10.1038/s41467-019-09501-6}}.

\bibitem{Hangleiter_2020}
Dominik Hangleiter, Ingo Roth, Daniel Nagaj, and Jens Eisert.
\newblock Easing the monte carlo sign problem.
\newblock {\em Science Advances}, 6(33), August 2020.
\newblock URL: \url{http://dx.doi.org/10.1126/sciadv.abb8341}, \href
  {https://doi.org/10.1126/sciadv.abb8341} {\path{doi:10.1126/sciadv.abb8341}}.

\bibitem{chen2023efficient}
Chi-Fang Chen, Michael~J. Kastoryano, and András Gilyén.
\newblock An efficient and exact noncommutative quantum gibbs sampler, 2023.
\newblock \href {https://arxiv.org/abs/2311.09207} {\path{arXiv:2311.09207}}.

\bibitem{chen_efficient_2025}
Chi-Fang Chen, Michael Kastoryano, Fernando G. S.~L. Brandão, and András
  Gilyén.
\newblock Efficient quantum thermal simulation.
\newblock {\em Nature}, 646(8085):561--566, October 2025.
\newblock URL: \url{https://www.nature.com/articles/s41586-025-09583-x}, \href
  {https://doi.org/10.1038/s41586-025-09583-x}
  {\path{doi:10.1038/s41586-025-09583-x}}.

\bibitem{rouze_efficient_2026}
Cambyse Rouzé, Daniel Stilck~França, and Álvaro~M. Alhambra.
\newblock Efficient thermalization and universal quantum computing with quantum
  {Gibbs} samplers.
\newblock {\em Nature Physics}, pages 1--5, April 2026.
\newblock URL: \url{https://www.nature.com/articles/s41567-026-03246-y}, \href
  {https://doi.org/10.1038/s41567-026-03246-y}
  {\path{doi:10.1038/s41567-026-03246-y}}.

\bibitem{ding2024efficientquantumgibbssamplers}
Zhiyan Ding, Bowen Li, and Lin Lin.
\newblock Efficient quantum gibbs samplers with kubo–martin–schwinger
  detailed balance condition.
\newblock {\em Communications in Mathematical Physics}, 406(3), February 2025.
\newblock URL: \url{http://dx.doi.org/10.1007/s00220-025-05235-3}, \href
  {https://doi.org/10.1007/s00220-025-05235-3}
  {\path{doi:10.1007/s00220-025-05235-3}}.

\bibitem{kotecky1986cluster}
Roman Koteck{\'y} and David Preiss.
\newblock Cluster expansion for abstract polymer models.
\newblock {\em Communications in Mathematical Physics}, 103:491--498, 1986.
\newblock \href {https://doi.org/10.1007/BF01211762}
  {\path{doi:10.1007/BF01211762}}.

\bibitem{fernandez2007cluster}
Roberto Fern{\'a}ndez and Aldo Procacci.
\newblock Cluster expansion for abstract polymer models. new bounds from an old
  approach.
\newblock {\em Communications in Mathematical Physics}, 274(1):123--140, 2007.
\newblock \href {https://doi.org/10.1007/s00220-007-0279-2}
  {\path{doi:10.1007/s00220-007-0279-2}}.

\bibitem{mann2024algorithmic}
Ryan~L. Mann and Romy~M. Minko.
\newblock Algorithmic cluster expansions for quantum problems.
\newblock {\em PRX Quantum}, 5(1), January 2024.
\newblock URL: \url{http://dx.doi.org/10.1103/PRXQuantum.5.010305}, \href
  {https://doi.org/10.1103/prxquantum.5.010305}
  {\path{doi:10.1103/prxquantum.5.010305}}.

\bibitem{Harrow_2020}
Aram~W. Harrow, Saeed Mehraban, and Mehdi Soleimanifar.
\newblock Classical algorithms, correlation decay, and complex zeros of
  partition functions of quantum many-body systems.
\newblock {\em Proceedings of the 52nd Annual ACM SIGACT Symposium on Theory of
  Computing}, Jun 2020.
\newblock URL: \url{http://dx.doi.org/10.1145/3357713.3384322}, \href
  {https://doi.org/10.1145/3357713.3384322}
  {\path{doi:10.1145/3357713.3384322}}.

\bibitem{Bravyi_2021}
Sergey Bravyi, David Gosset, and Ramis Movassagh.
\newblock Classical algorithms for quantum mean values.
\newblock {\em Nature Physics}, 17(3):337–341, January 2021.
\newblock URL: \url{http://dx.doi.org/10.1038/s41567-020-01109-8}, \href
  {https://doi.org/10.1038/s41567-020-01109-8}
  {\path{doi:10.1038/s41567-020-01109-8}}.

\bibitem{Haah_2024}
Jeongwan Haah, Robin Kothari, and Ewin Tang.
\newblock Learning quantum hamiltonians from high-temperature gibbs states and
  real-time evolutions.
\newblock {\em Nature Physics}, 20(6):1027–1031, March 2024.
\newblock URL: \url{http://dx.doi.org/10.1038/s41567-023-02376-x}, \href
  {https://doi.org/10.1038/s41567-023-02376-x}
  {\path{doi:10.1038/s41567-023-02376-x}}.

\bibitem{Wild2023}
Dominik~S. Wild and \'Alvaro~M. Alhambra.
\newblock Classical simulation of short-time quantum dynamics.
\newblock {\em PRX Quantum}, 4:020340, Jun 2023.
\newblock URL: \url{https://link.aps.org/doi/10.1103/PRXQuantum.4.020340},
  \href {https://doi.org/10.1103/PRXQuantum.4.020340}
  {\path{doi:10.1103/PRXQuantum.4.020340}}.

\bibitem{mann2026efficientalgorithmsweaklyinteractingquantum}
Ryan~L. Mann and Gabriel Waite.
\newblock Efficient algorithms for weakly interacting quantum spin systems.
\newblock {\em Physical Review A}, 114(1), 2026.
\newblock URL: \url{http://dx.doi.org/10.1103/45rz-42gr}, \href
  {https://doi.org/10.1103/45rz-42gr} {\path{doi:10.1103/45rz-42gr}}.

\bibitem{midha2025beliefpropagationclustercorrectedtensor}
Siddhant Midha and Yifan~F. Zhang.
\newblock Beyond belief propagation: Cluster-corrected tensor network
  contraction with exponential convergence.
\newblock {\em PRX Quantum}, 7:033010, Jul 2026.
\newblock URL: \url{https://link.aps.org/doi/10.1103/sdlx-d3tn}, \href
  {https://doi.org/10.1103/sdlx-d3tn} {\path{doi:10.1103/sdlx-d3tn}}.

\bibitem{Segovia_PRXQuantum2025}
Jorge S\'anchez-Segovia, Jan T.~Schneider, and \'Alvaro M.~Alhambra.
\newblock High-temperature partition functions and classical simulatability of
  long-range quantum systems.
\newblock {\em PRX Quantum}, 6:040366, Dec 2025.
\newblock URL: \url{https://link.aps.org/doi/10.1103/cww3-j1vd}, \href
  {https://doi.org/10.1103/cww3-j1vd} {\path{doi:10.1103/cww3-j1vd}}.

\bibitem{zlokapa2026sykthermalexpectationsclassically}
Alexander Zlokapa and Bobak~T. Kiani.
\newblock Syk thermal expectations are classically easy at any temperature,
  2026.
\newblock URL: \url{https://arxiv.org/abs/2602.22619}, \href
  {https://arxiv.org/abs/2602.22619} {\path{arXiv:2602.22619}}.

\bibitem{zlokapa_rigorous_2026}
Alexander Zlokapa.
\newblock A rigorous quasipolynomial-time classical algorithm for {SYK} thermal
  expectations, April 2026.
\newblock arXiv:2604.21089 [quant-ph].
\newblock URL: \url{http://arxiv.org/abs/2604.21089}, \href
  {https://doi.org/10.48550/arXiv.2604.21089}
  {\path{doi:10.48550/arXiv.2604.21089}}.

\bibitem{bergamaschi2026fastmixingalltoallquantum}
Thiago Bergamaschi.
\newblock Fast mixing of all-to-all quantum systems at high temperatures, 2026.
\newblock URL: \url{https://arxiv.org/abs/2606.26090}, \href
  {https://arxiv.org/abs/2606.26090} {\path{arXiv:2606.26090}}.

\bibitem{sim2026stochasticloopcorrectionsbelief}
Gi~Beom Sim, Tae~Hyeon Park, Kwang~S. Kim, Yanmei Zang, Xiaorong Zou, Hye~Jung
  Kim, D.~ChangMo Yang, Soohaeng~Yoo Willow, and Chang~Woo Myung.
\newblock Stochastic loop corrections to belief propagation for tensor network
  contraction, 2026.
\newblock URL: \url{https://arxiv.org/abs/2603.08427}, \href
  {https://arxiv.org/abs/2603.08427} {\path{arXiv:2603.08427}}.

\bibitem{Kai_e26050401}
Patrick Adelhardt, Jan~A. Koziol, Anja Langheld, and Kai~P. Schmidt.
\newblock Monte carlo based techniques for quantum magnets with long-range
  interactions.
\newblock {\em Entropy}, 26(5), 2024.
\newblock URL: \url{https://www.mdpi.com/1099-4300/26/5/401}, \href
  {https://doi.org/10.3390/e26050401} {\path{doi:10.3390/e26050401}}.

\bibitem{chen2025convergencecumulantexpansionpolynomialtime}
Hongrui Chen, Cambyse Rouzé, Jielun Chen, Jiaqing Jiang, Samuel~O. Scalet,
  Yongtao Zhan, Garnet Kin-Lic Chan, Lexing Ying, and Yu~Tong.
\newblock Convergence of the cumulant expansion and polynomial-time algorithm
  for weakly interacting fermions, 2025.
\newblock URL: \url{https://arxiv.org/abs/2512.12010}, \href
  {https://arxiv.org/abs/2512.12010} {\path{arXiv:2512.12010}}.

\bibitem{DiagMC_Van_Houcke_2010}
Kris Van~Houcke, Evgeny Kozik, N.~Prokof’ev, and B.~Svistunov.
\newblock Diagrammatic monte carlo.
\newblock {\em Physics Procedia}, 6:95–105, 2010.
\newblock URL: \url{http://dx.doi.org/10.1016/j.phpro.2010.09.034}, \href
  {https://doi.org/10.1016/j.phpro.2010.09.034}
  {\path{doi:10.1016/j.phpro.2010.09.034}}.

\bibitem{RossiDiagMC_PhysRevLett.119.045701}
Riccardo Rossi.
\newblock Determinant diagrammatic monte carlo algorithm in the thermodynamic
  limit.
\newblock {\em Phys. Rev. Lett.}, 119:045701, Jul 2017.
\newblock URL: \url{https://link.aps.org/doi/10.1103/PhysRevLett.119.045701},
  \href {https://doi.org/10.1103/PhysRevLett.119.045701}
  {\path{doi:10.1103/PhysRevLett.119.045701}}.

\bibitem{blanca2023fast}
Antonio Blanca, Sarah Cannon, and Will Perkins.
\newblock Fast and perfect sampling of subgraphs and polymer systems.
\newblock 2023.
\newblock URL: \url{https://arxiv.org/abs/2202.05907}, \href
  {https://arxiv.org/abs/2202.05907} {\path{arXiv:2202.05907}}.

\bibitem{chen2021fastalgorithmslowtemperatures}
Zongchen Chen, Andreas Galanis, Leslie~A. Goldberg, Will Perkins, James
  Stewart, and Eric Vigoda.
\newblock Fast algorithms at low temperatures via markov chains.
\newblock {\em Random Structures \& Algorithms}, 58(2):294--321, 2021.
\newblock URL: \url{https://onlinelibrary.wiley.com/doi/abs/10.1002/rsa.20968},
  \href
  {https://arxiv.org/abs/https://onlinelibrary.wiley.com/doi/pdf/10.1002/rsa.20968}
  {\path{arXiv:https://onlinelibrary.wiley.com/doi/pdf/10.1002/rsa.20968}},
  \href {https://doi.org/10.1002/rsa.20968} {\path{doi:10.1002/rsa.20968}}.

\bibitem{mann2026fastalgorithmsstoquasticspin}
Ryan~L. Mann.
\newblock Fast algorithms for stoquastic spin systems, 2026.
\newblock URL: \url{https://arxiv.org/abs/2608.19489}, \href
  {https://arxiv.org/abs/2608.19489} {\path{arXiv:2608.19489}}.

\bibitem{tong2026complexityclassificationhightemperaturebosons}
Xin-Hai Tong and Tomotaka Kuwahara.
\newblock Toward a complexity classification of high-temperature bosons:
  Computational tractability and power-law clustering, 2026.
\newblock URL: \url{https://arxiv.org/abs/2509.25572}, \href
  {https://arxiv.org/abs/2509.25572} {\path{arXiv:2509.25572}}.

\bibitem{putterman2026quantumthermalstateslook}
Harald Putterman, Alexander Zlokapa, and Jordan Cotler.
\newblock When quantum thermal states look classical, 2026.
\newblock URL: \url{https://arxiv.org/abs/2607.28536}, \href
  {https://arxiv.org/abs/2607.28536} {\path{arXiv:2607.28536}}.

\bibitem{bravyi2021complexity}
Sergey Bravyi, Anirban Chowdhury, David Gosset, and Pawel Wocjan.
\newblock Quantum hamiltonian complexity in thermal equilibrium.
\newblock {\em Nature Physics}, 18(11):1367--1370, Nov 2022.
\newblock \href {https://doi.org/10.1038/s41567-022-01742-5}
  {\path{doi:10.1038/s41567-022-01742-5}}.

\bibitem{Mann2021}
Ryan~L. Mann and Tyler Helmuth.
\newblock Efficient algorithms for approximating quantum partition functions.
\newblock {\em Journal of Mathematical Physics}, 62(2), February 2021.
\newblock URL: \url{http://dx.doi.org/10.1063/5.0013689}, \href
  {https://doi.org/10.1063/5.0013689} {\path{doi:10.1063/5.0013689}}.

\bibitem{StatisticalMechanicsofLatticeSystems_friedli_velenik_2017}
Sacha Friedli and Yvan Velenik.
\newblock {\em Statistical Mechanics of Lattice Systems: A Concrete
  Mathematical Introduction}.
\newblock Cambridge University Press, 2017.
\newblock \href {https://doi.org/10.1017/9781316882603}
  {\path{doi:10.1017/9781316882603}}.

\bibitem{levin2017markovmixing}
David~A. Levin, Yuval Peres, and Elizabeth~L. Wilmer.
\newblock {\em {Markov Chains and Mixing Times}}.
\newblock American Mathematical Society, Providence, RI, 2 edition, 2017.

\bibitem{signpbm_Pan_2024}
Gaopei Pan and Zi~Yang Meng.
\newblock {\em The sign problem in quantum Monte Carlo simulations}, page
  879–893.
\newblock Elsevier, 2024.
\newblock URL: \url{http://dx.doi.org/10.1016/B978-0-323-90800-9.00095-0},
  \href {https://doi.org/10.1016/b978-0-323-90800-9.00095-0}
  {\path{doi:10.1016/b978-0-323-90800-9.00095-0}}.

\bibitem{pathcoupling_BubleyDyer1997PathCoupling}
Russ Bubley and Martin Dyer.
\newblock Path coupling: A technique for proving rapid mixing in markov chains.
\newblock In {\em Proceedings of the 38th Annual Symposium on Foundations of
  Computer Science (FOCS)}, pages 223--231. IEEE, 1997.
\newblock \href {https://doi.org/10.1109/SFCS.1997.646111}
  {\path{doi:10.1109/SFCS.1997.646111}}.

\bibitem{kuwahara2020longrange}
Tomotaka Kuwahara and Keiji Saito.
\newblock Lieb--robinson bounds and the speed of information propagation in
  long-range interacting quantum systems.
\newblock {\em Physical Review Letters}, 124(10):100601, 2020.
\newblock \href {https://doi.org/10.1103/PhysRevLett.124.100601}
  {\path{doi:10.1103/PhysRevLett.124.100601}}.

\bibitem{chen_continuous_2023}
Cheng Chen, Guillaume Bornet, Marcus Bintz, Gabriel Emperauger, Lucas Leclerc,
  Vincent~S. Liu, Pascal Scholl, Daniel Barredo, Johannes Hauschild, Shubhayu
  Chatterjee, Michael Schuler, Andreas~M. Läuchli, Michael~P. Zaletel, Thierry
  Lahaye, Norman~Y. Yao, and Antoine Browaeys.
\newblock Continuous symmetry breaking in a two-dimensional {Rydberg} array.
\newblock {\em Nature}, 616(7958):691--695, April 2023.
\newblock URL: \url{https://www.nature.com/articles/s41586-023-05859-2}, \href
  {https://doi.org/10.1038/s41586-023-05859-2}
  {\path{doi:10.1038/s41586-023-05859-2}}.

\bibitem{schauss2012rydberg}
P.~Schau{\ss} et~al.
\newblock Observation of spatially ordered structures in a two-dimensional
  rydberg gas.
\newblock {\em Nature}, 491(7422):87--91, 2012.
\newblock \href {https://doi.org/10.1038/nature11596}
  {\path{doi:10.1038/nature11596}}.

\bibitem{Richerme2014}
Philip Richerme, Zhe-Xuan Gong, Aaron Lee, Crystal Senko, Jacob Smith, Michael
  Foss-Feig, Spyridon Michalakis, Alexey~V. Gorshkov, and Christopher Monroe.
\newblock Non-local propagation of correlations in quantum systems with
  long-range interactions.
\newblock {\em Nature}, 511(7508):198--201, Jul 2014.
\newblock \href {https://doi.org/10.1038/nature13450}
  {\path{doi:10.1038/nature13450}}.

\bibitem{defenu_longrange}
Nicol\`o Defenu, Tobias Donner, Tommaso Macr\`{\i}, Guido Pagano, Stefano
  Ruffo, and Andrea Trombettoni.
\newblock Long-range interacting quantum systems.
\newblock {\em Rev. Mod. Phys.}, 95:035002, Aug 2023.
\newblock URL: \url{https://link.aps.org/doi/10.1103/RevModPhys.95.035002},
  \href {https://doi.org/10.1103/RevModPhys.95.035002}
  {\path{doi:10.1103/RevModPhys.95.035002}}.

\bibitem{evenbly2026loopseries}
Glen Evenbly, Nicola Pancotti, Ashley Milsted, Johnnie Gray, and Garnet Kin-Lic
  Chan.
\newblock Loop series expansions for tensor networks.
\newblock {\em Physical Review Research}, 8:013245, 2026.
\newblock URL: \url{https://arxiv.org/abs/2409.03108}, \href
  {https://arxiv.org/abs/2409.03108} {\path{arXiv:2409.03108}}, \href
  {https://doi.org/10.1103/vqks-cr6x} {\path{doi:10.1103/vqks-cr6x}}.

\bibitem{gray2026loopclusterexpansions}
Johnnie Gray, Gunhee Park, Glen Evenbly, Nicola Pancotti, Eirik~F. Kj{\o}nstad,
  and Garnet Kin-Lic Chan.
\newblock Tensor network loop cluster expansions for quantum many-body
  problems.
\newblock {\em Physical Review B}, 113:235135, 2026.
\newblock \href {https://doi.org/10.1103/r6mz-6q3g}
  {\path{doi:10.1103/r6mz-6q3g}}.

\bibitem{schuch2008mctensorcontractions}
Norbert Schuch, Michael~M. Wolf, Frank Verstraete, and J.~Ignacio Cirac.
\newblock Simulation of quantum many-body systems with strings of operators and
  monte carlo tensor contractions.
\newblock {\em Physical Review Letters}, 100:040501, 2008.
\newblock \href {https://doi.org/10.1103/PhysRevLett.100.040501}
  {\path{doi:10.1103/PhysRevLett.100.040501}}.

\bibitem{chen2025signproblemtn}
Jielun Chen, Jiaqing Jiang, Dominik Hangleiter, and Norbert Schuch.
\newblock Sign problem in tensor network contraction.
\newblock {\em PRX Quantum}, 6:010312, 2025.
\newblock URL: \url{https://arxiv.org/abs/2404.19023}, \href
  {https://arxiv.org/abs/2404.19023} {\path{arXiv:2404.19023}}, \href
  {https://doi.org/10.1103/PRXQuantum.6.010312}
  {\path{doi:10.1103/PRXQuantum.6.010312}}.

\bibitem{nelson2026polynomialtimeclassicalsamplernoisy}
Jon Nelson, Joel Rajakumar, Chao Yin, Yifan~F. Zhang, and Michael~J. Gullans.
\newblock A polynomial-time classical sampler for noisy quantum circuits from
  statistical mechanics, 2026.
\newblock URL: \url{https://arxiv.org/abs/2610.00548}, \href
  {https://arxiv.org/abs/2610.00548} {\path{arXiv:2610.00548}}.

\bibitem{coudron2026polynomialtimeadditiveerrorestimationoutput}
Matthew Coudron, Michael~J. Gullans, Jon Nelson, Joel Rajakumar, and Shi
  Jie~Samuel Tan.
\newblock Polynomial-time additive-error estimation of output probabilities for
  shallow quantum circuits, 2026.
\newblock URL: \url{https://arxiv.org/abs/2610.02146}, \href
  {https://arxiv.org/abs/2610.02146} {\path{arXiv:2610.02146}}.

\bibitem{Helmuth_2023}
Tyler Helmuth and Ryan~L. Mann.
\newblock Efficient algorithms for approximating quantum partition functions at
  low temperature.
\newblock {\em Quantum}, 7:1155, October 2023.
\newblock URL: \url{http://dx.doi.org/10.22331/q-2023-10-25-1155}, \href
  {https://doi.org/10.22331/q-2023-10-25-1155}
  {\path{doi:10.22331/q-2023-10-25-1155}}.

\bibitem{Rossi_2017}
R.~Rossi, N.~Prokof'ev, B.~Svistunov, K.~Van~Houcke, and F.~Werner.
\newblock Polynomial complexity despite the fermionic sign.
\newblock {\em Europhysics Letters}, 118(1):10004, may 2017.
\newblock \href {https://doi.org/10.1209/0295-5075/118/10004}
  {\path{doi:10.1209/0295-5075/118/10004}}.

\bibitem{Rossi_PhysRevLett.121.130405}
R.~Rossi, T.~Ohgoe, K.~Van~Houcke, and F.~Werner.
\newblock Resummation of diagrammatic series with zero convergence radius for
  strongly correlated fermions.
\newblock {\em Phys. Rev. Lett.}, 121:130405, Sep 2018.
\newblock URL: \url{https://link.aps.org/doi/10.1103/PhysRevLett.121.130405},
  \href {https://doi.org/10.1103/PhysRevLett.121.130405}
  {\path{doi:10.1103/PhysRevLett.121.130405}}.

\bibitem{nlce_review}
Baoming Tang, Ehsan Khatami, and Marcos Rigol.
\newblock A short introduction to numerical linked-cluster expansions.
\newblock {\em Computer Physics Communications}, 184(3):557--564, 2013.
\newblock URL:
  \url{https://www.sciencedirect.com/science/article/pii/S0010465512003414},
  \href {https://doi.org/10.1016/j.cpc.2012.10.008}
  {\path{doi:10.1016/j.cpc.2012.10.008}}.

\bibitem{QuantumDynamicsNLCE}
Ian~G. White, Bhuvanesh Sundar, and Kaden R.~A. Hazzard.
\newblock Quantum dynamics from a numerical linked cluster expansion, 2021.
\newblock URL: \url{https://arxiv.org/abs/1710.07696}, \href
  {https://arxiv.org/abs/1710.07696} {\path{arXiv:1710.07696}}.

\bibitem{chan2026efficientclassicalsimulationtwodimensional}
Jia-Lin Chan, Tao Xiang, and Yantao Wu.
\newblock Efficient classical simulation of two-dimensional long-range systems:
  Rydberg arrays and beyond, 2026.
\newblock URL: \url{https://arxiv.org/abs/2607.05178}, \href
  {https://arxiv.org/abs/2607.05178} {\path{arXiv:2607.05178}}.

\bibitem{Kai_PhysRevB.94.075156}
Sebastian Fey and Kai~Phillip Schmidt.
\newblock Critical behavior of quantum magnets with long-range interactions in
  the thermodynamic limit.
\newblock {\em Phys. Rev. B}, 94:075156, Aug 2016.
\newblock URL: \url{https://link.aps.org/doi/10.1103/PhysRevB.94.075156}, \href
  {https://doi.org/10.1103/PhysRevB.94.075156}
  {\path{doi:10.1103/PhysRevB.94.075156}}.

\bibitem{Kai_duft2026exactfastseriesexpansions}
Antonia Duft, Patrick Adelhardt, Jan~Alexander Koziol, Andreas~A. Buchheit, and
  Kai~Phillip Schmidt.
\newblock Exact and fast series expansions for quantum models with long-range
  interactions, 2026.
\newblock URL: \url{https://arxiv.org/abs/2609.18761}, \href
  {https://arxiv.org/abs/2609.18761} {\path{arXiv:2609.18761}}.

\bibitem{rayudu2026spectralgapleeyanghamiltonians}
Chaithanya Rayudu and Jun Takahashi.
\newblock Spectral gap of lee-yang hamiltonians, 2026.
\newblock URL: \url{https://arxiv.org/abs/2607.10765}, \href
  {https://arxiv.org/abs/2607.10765} {\path{arXiv:2607.10765}}.

\bibitem{wong2026leeyangtensorshamiltoniancomplexity}
Benjamin Wong, Sergey Bravyi, David Gosset, and Yinchen Liu.
\newblock Lee-yang tensors and hamiltonian complexity, 2026.
\newblock URL: \url{https://arxiv.org/abs/2602.03605}, \href
  {https://arxiv.org/abs/2602.03605} {\path{arXiv:2602.03605}}.

\bibitem{bourin2007matrixsubadditivityinequalityfab}
Jean-Christophe Bourin and Mitsuru Uchiyama.
\newblock A matrix subadditivity inequality for f(a+b) and f(a)+f(b), 2007.
\newblock URL: \url{https://arxiv.org/abs/math/0702475}, \href
  {https://arxiv.org/abs/math/0702475} {\path{arXiv:math/0702475}}.

\bibitem{bjorklund2008computingtuttepolynomialvertexexponential}
Andreas Björklund, Thore Husfeldt, Petteri Kaski, and Mikko Koivisto.
\newblock Computing the tutte polynomial in vertex-exponential time.
\newblock In {\em 2008 49th Annual IEEE Symposium on Foundations of Computer
  Science}, pages 677--686, 2008.
\newblock \href {https://doi.org/10.1109/FOCS.2008.40}
  {\path{doi:10.1109/FOCS.2008.40}}.

\end{thebibliography}

\newpage

\appendix
\onecolumngrid
\section{Exponential sign problem for sampling}\label{app:sign_problem}

Let us write the quantum partition function as a polymer model 

\begin{align}
    Z= \sum_{\Gamma \in \mathcal{G}} \prod_{\gamma \in \Gamma} w_\gamma.
\end{align}
If the weights are positive one can directly apply the framework of \cite{chen2021fastalgorithmslowtemperatures} to get efficient algorithms for observables and partition functions of polymer models when a condition akin to the Con. \ref{con:kotecky_preiss_criteria} is satisfied. One can also try to upgrade this Monte Carlo framework to complex-valued polymer models by defining the absolute value polymer model

\begin{align}
    Z'= \sum_{\Gamma \in \mathcal{G}} \prod_{\gamma \in \Gamma} \vert w_\gamma \vert,
\end{align}
such that any observable gets redefined as

\begin{align}
    \langle O \rangle =\frac{\sum_{\Gamma \in \mathcal{G}} \prod_{\gamma \in \Gamma}  w_\gamma  O(\Gamma)}{\sum_{\Gamma \in \mathcal{G}} \prod_{\gamma \in \Gamma}  w_\gamma }= \frac{\sum_{\Gamma \in \mathcal{G}} \prod_{\gamma \in \Gamma} \vert  w_\gamma \vert O(\Gamma) \cdot s(\Gamma) / \sum_{\Gamma \in \mathcal{G}} \prod_{\gamma \in \Gamma} \vert w_\gamma \vert}{\left(\sum_{\Gamma \in \mathcal{G}} \prod_{\gamma \in \Gamma} \vert  w_\gamma \vert \cdot s(\Gamma) \right) / \sum_{\Gamma \in \mathcal{G}} \prod_{\gamma \in \Gamma} \vert w_\gamma \vert }   = \frac{ \langle O s\rangle'}{\langle  s\rangle'},
\end{align}
where $s(\Gamma)= \prod_{\gamma \in \Gamma} w_\gamma / \vert w_\gamma \vert$ is the sign or phase associated to the compatible configuration $\Gamma \in \mathcal{G}$.  The mean sign can be written as

\begin{align}
    \langle s \rangle' = \frac{Z}{Z'} = \exp \left(\sum_{\Gamma \in \mathcal{G}_C }\varphi (H_\Gamma)  \prod_{\gamma \in \Gamma}\left[\frac{w_\gamma^{m_\gamma(\Gamma)}}{m_\gamma(\Gamma)!}\right] - \prod_{\gamma \in \Gamma} \left[\frac{\vert w_\gamma^{m_\gamma(\Gamma)} \vert}{m_\gamma(\Gamma)!}\right] \right).
\end{align}
When estimating observables via a Monte Carlo chain, the number of samples needed scales with the variance of the observable $\langle O^2 \rangle \equiv \langle O^2 \rangle' / \langle s \rangle'$. 

We thus see that here an exponential sign problem can emerge even if the Koteck\'y-Preiss criteria holds: while these criteria guarantee that the cluster expansion converges and $\log Z$ is extensive, the mean sign itself can still be an exponentially decaying quantity in $\vert V(G)\vert$ unless the weights are strictly positive (or at least positive up to sizes $M \propto \log (\vert V(G)\vert \varepsilon^{-1})$). In such cases, the number of samples required to estimate $\langle O \rangle$ grows exponentially as $\propto \exp(\vert V(G)\vert)$, resulting in an inefficient algorithm. This is analogous to the well-known sign problem, heavily studied in the QMC literature for thermal equilibrium \cite{troyer2005computational,signpbm_Pan_2024}.

The root of this sign problem in our context stems fundamentally from the non-commutativity of the local Hamiltonian terms. In the case of the quantum thermal partition function, the polymer weight is an alternating-sign sum of matrix exponential traces

\begin{equation}
    w_\gamma := (-1)^{\|\gamma\|} \sum_{T \subseteq E(\gamma)} (-1)^{|T|} \operatorname{Tr} \left[ \mathrm{e}^{-\beta \sum_{Z \in T} h_Z} \right].
\end{equation}

If the terms $h_Z$ commute, one can easily avoid the sign problem. By applying the shift $h_Z \rightarrow h_Z -2 \norm{h_Z} \mathbb{I}$, the exponents become positive semi-definite matrices. In the commuting case, the matrix exponentials factorize perfectly, and a straightforward rewriting turns the alternating sum into a manifestly positive quantity. When non-commutativity is introduced, we lose this global positivity guarantee. For polymers of size up to $\norm{\gamma} \leq 2$, the weights still remain positive. This is formalized in the following corollary:

\begin{cor}\label{cor_positiveweightN2}
    Given two positive semi-definite matrices $A, B \succeq 0$ 
    \begin{align}
        \Tr[\mathrm{e}^{A+B}]- \Tr[\mathrm{e}^A]- \Tr[\mathrm{e}^B]+ \Tr[\mathbb{I}] \geq 0
    \end{align}
\end{cor}
\begin{proof}
    It follows from applying Thm. 1.2 of \cite{bourin2007matrixsubadditivityinequalityfab} with the norm $\norm{\cdot}= \Tr[\cdot]$ and the convex function $f(x) = \mathrm{e}^{x}-1$.
\end{proof}

Unfortunately, this positivity breaks down for polymers of size $\norm{\gamma} \ge 3$. Even when the Hamiltonian terms are shifted to be positive definite matrices, the alternating sum of traces of exponentials can evaluate to a negative number due to non-commutativity. For instance, for $\norm{\gamma}=3$, the following set of three positive definite matrices serves as a counterexample

\begin{equation}
\begin{aligned}
A&=\begin{pmatrix}
0.2043328343 & -0.0830288469\\
-0.0830288469 & 0.0379359076
\end{pmatrix},\quad
B=\begin{pmatrix}
0.2021656944 & 0.1887743592\\
0.1887743592 & 0.1852969351
\end{pmatrix},\\[4pt]
C&=\begin{pmatrix}
0.0022878815 & -0.0131434550\\
-0.0131434550 & 0.0829966280
\end{pmatrix},
\end{aligned}
\end{equation}
as

\begin{align}
    \Tr[\mathrm{e}^{A+B+C}-\mathrm{e}^{A+B}-\mathrm{e}^{A+C}-\mathrm{e}^{B+C}+\mathrm{e}^{A}+\mathrm{e}^{B}+\mathrm{e}^{C}- \mathbb{I}] \approx -1.01 \cdot 10^{-3} 
\end{align}
The existence of this counterexample demonstrates that the sign problem is unavoidable for general non-commuting quantum systems at larger polymer sizes.

\section{Cluster Markov Chain}

\subsection{Convergence of the normalizing constant $A_M$}\label{app:A_convergence}

\begin{proof}(Lem. \ref{lemma_Aextensive})
    To establish the extensivity and convergence of $A_M$, we first bound the collective contributions of all cluster configurations connected to a reference root polymer $\gamma^*$. Specifically, by invoking Con. \ref{con:kotecky_preiss_criteria}, we bound these contributions in terms of polynomial weights scaled by the size $\norm{\gamma}$, reducing the global summation to locally controlled polymer expansions.
    
    \begin{figure}[t]
        \centering
        \includegraphics[width=0.65\linewidth]{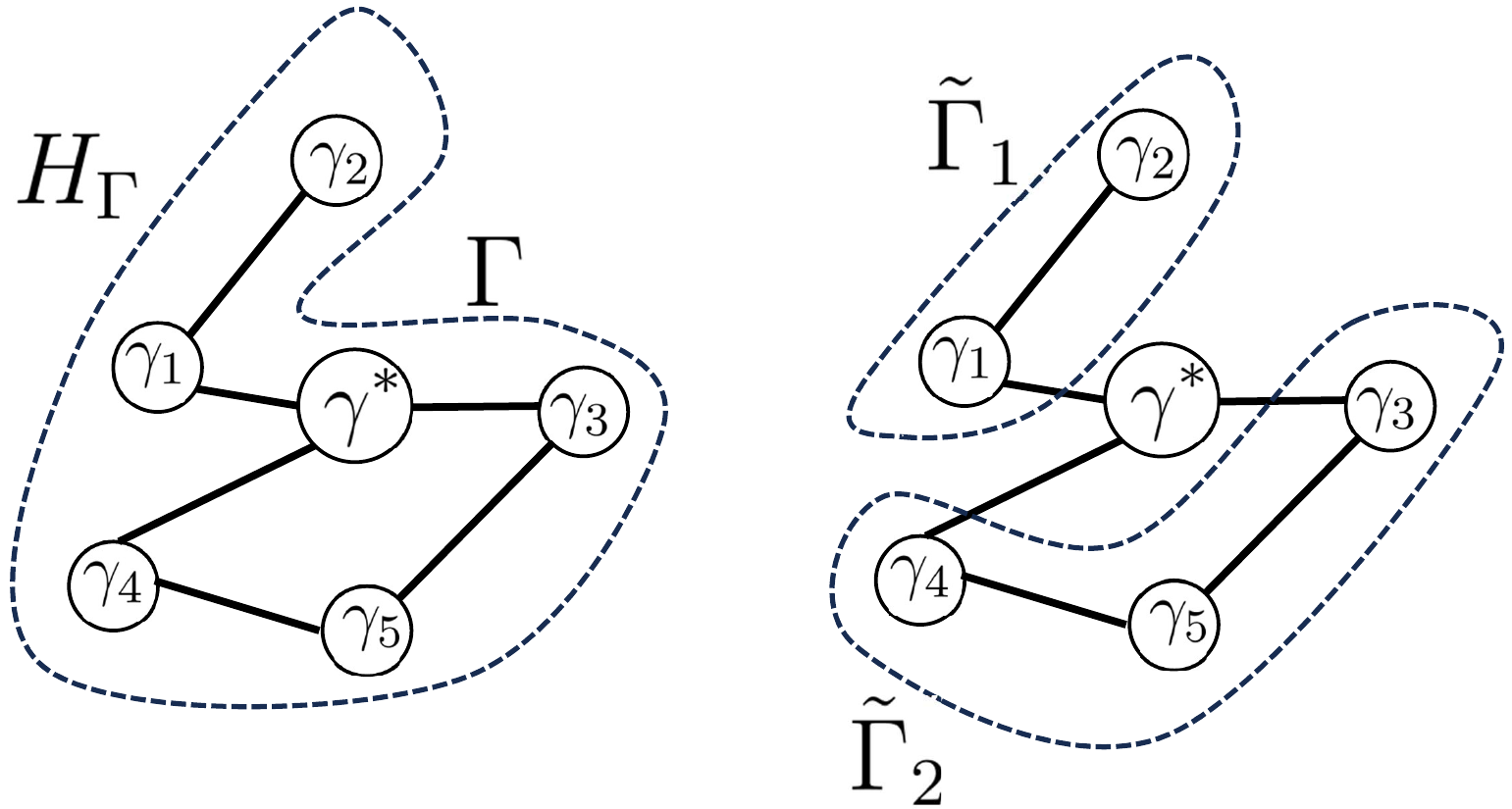}
        \caption{Example of decomposition of a cluster $\Gamma$ which contains $\gamma^*$ as a set of sub-clusters $\tilde{\Gamma}_i$.}
        \label{fig:proofA}
    \end{figure}
    To formalize this recursive structure, we first isolate the contribution of all sub-clusters that are pinned to a specific root. Let us define the following quantities 
    \begin{align}
        A'(\gamma^*) = \sum_{\substack{(\Gamma, m_\Gamma) \in \mathcal{G}_C\\ \Gamma \nsim \gamma^* }} \prod_{\gamma \in \Gamma} \left( \vert w_\gamma \vert \mathrm{e}^{ d(\gamma) } \right)^{m_\gamma (\Gamma)} ,
    \end{align}
    and 

    \begin{align}
        G_n(\gamma^*) = \sum_{\substack{(\Gamma, m_\Gamma) \in \mathcal{G}_C\\ \Gamma \ni \gamma^* \\ n_\gamma (\Gamma)\leq n }} \prod_{\gamma \in \Gamma} \left( \vert w_\gamma \vert \mathrm{e}^{ d(\gamma) } \right)^{m_\gamma (\Gamma)} ,
    \end{align}
    where $n_\gamma (\Gamma)$ is the number of polymers of the cluster.

    We prove first that for all $\gamma$ and n, $G_n(\gamma) \leq \vert w_\gamma \vert \mathrm{e}^{d(\gamma) }$. We proceed by induction; for $n=1$ the statement is trivial

    \begin{align}
        G_1(\gamma^*) = \sum_{\substack{(\Gamma, m_\Gamma) \in \mathcal{G}_C\\ \Gamma \ni \gamma^* \\ n_\gamma (\Gamma)\leq 1 }} \prod_{\gamma \in \Gamma} \left( \vert w_\gamma \vert \mathrm{e}^{ d(\gamma) } \right)^{m_\gamma (\Gamma)}  = \vert w_{\gamma^*} \vert   \mathrm{e}^{d(\gamma^*)  }
    \end{align}
    We assume, by the inductive hypothesis, that $G_n (\gamma^*) \leq \vert w_\gamma \vert \mathrm{e}^{d(\gamma)}$. For $n+1$ we realize that by extracting $\gamma^*$ contribution, one can always decompose $\Gamma$ in a set of sub-clusters $\{\tilde{\Gamma}_i\}$ connected to $\gamma^*$, as in Fig. \ref{fig:proofA}. This leads to the following upper bound 

    \begin{align}
        G_{n+1}(\gamma^*) =& \sum_{\substack{(\Gamma, m_\Gamma) \in \mathcal{G}_C\\ \Gamma \ni \gamma^* \\ n_\gamma (\Gamma)\leq n }} \prod_{\gamma \in \Gamma} \left( \vert w_\gamma \vert \mathrm{e}^{ d(\gamma) } \right)^{m_\gamma (\Gamma)} , \\
        &\leq \vert w_{\gamma^*} \vert \sum_{k=0}^{n}\sum_{\tilde{\Gamma}_0, \dots \tilde{\Gamma}_k } \prod_{M}  \prod_{\gamma \in \tilde{\Gamma}_k} \left( \vert w_\gamma \vert \mathrm{e}^{ d(\gamma) } \right) ^{m_\gamma (\tilde{\Gamma}_k)} \\
        &\leq \vert w_{\gamma^*} \vert \sum_{k=0}^{n} \frac{1}{k!} \left( \sum_{e \nsim \gamma^*}G_{n}(e) \right)^k \\
        &\leq \vert w_{\gamma^*} \vert \sum_{k=0}^{n+1} \frac{1}{k!} \left( \sum_{e \nsim \gamma^*}\vert w_e \vert \mathrm{e}^{d(e)} \right)^k \leq  \vert w_{\gamma^*} \vert \mathrm{e}^{d(\gamma^*) }
    \end{align}
    where we use in the last step the inductive hypothesis and Con. \ref{con:kotecky_preiss_criteria}. The factor $1/k!$ avoids over-counting when summing over sets of sub-clusters.  Finally 

    \begin{align}
        A'(\gamma^*) \leq \sum_{\gamma \nsim \gamma^*} G(\gamma) \leq \sum_{\gamma \nsim \gamma^*} \vert w_\gamma \vert \mathrm{e}^{a(\gamma)+d(\gamma)} \leq a(\gamma^*),
     \end{align}
    where we use again Con. \ref{con:kotecky_preiss_criteria}. It follows that 

    \begin{align}
        A' = \sum_{(\Gamma, m_\Gamma) \in \mathcal{G}_C} \prod_{\gamma \in \Gamma} \left( \vert w_\gamma \vert \mathrm{e}^{ d(\gamma) } \right)^{m_\gamma (\Gamma)} \leq \sum_{v \in \tilde{V}}A'(v) \leq \vert \tilde{V}(\tilde{G}) \vert,
    \end{align}
    and then for some $M'\leq M$ given our choice of $a(\gamma)=\norm{\gamma}/2$ and $d(\gamma)=2 \cdot \log 2 \norm{\gamma}$

    \begin{align}
        \vert A_M- A_{M'} \vert \leq \sum_{\Gamma , \norm{\Gamma} >M'} \prod_{\gamma \in \Gamma} \vert w_\gamma \vert^{m_\gamma} \leq 4^{-M'}\sum_{\Gamma , \norm{\Gamma} >M'} \prod_{\gamma \in \Gamma} \vert w_\gamma \vert^{m_\gamma} 4^{\norm{\gamma}} \leq 4^{-M'}\sum_{v \in \tilde{V}(\tilde{G})}A'(v)  \leq  4^{-M'}  \vert \tilde{V}(\tilde{G}) \vert ,
    \end{align}
    concluding the proof.

\end{proof}

\subsection{Proof of Lem. \ref{lemma_finitevariance}}\label{app:variance_convergence}

First, note that
\begin{align}
        \mathbb{E}_{p(\Gamma)}[\vert\xi (\Gamma)\vert^2] \leq \frac{1}{A_M} \sum_{v\in \tilde{V}} \sum_{\Gamma \nsim v} \left\vert \left(\frac{  \varphi(\Gamma) }{m_\gamma(\Gamma)!}\right)^2 \prod_{\gamma \in \Gamma}  w_\gamma \right \vert,  
    \end{align}
and we can simply rewrite
\begin{align}
         \sum_{\Gamma \nsim v} \left\vert \left(\frac{  \varphi(\Gamma) }{\prod_{\gamma \in \Gamma}m_\gamma(\Gamma)!}\right)^2 \prod_{\gamma \in \Gamma}  w_\gamma \right \vert &\leq 
         \sum_{\Gamma, \Gamma' \nsim v}   \left \vert   \varphi(\Gamma) \prod_{\gamma \in \Gamma}  \frac{   w_\gamma^{1/2}}{m_\gamma(\Gamma)!}\right \vert \cdot  \left \vert   \varphi(\Gamma') \prod_{\gamma \in \Gamma'}  \frac{   w_\gamma^{1/2}}{m_\gamma(\Gamma')!}\right \vert \nonumber\\
         &\leq \left( \sum_{\Gamma \nsim v} \left \vert   \varphi(\Gamma) \prod_{\gamma \in \Gamma}  \frac{   w_\gamma^{1/2}}{m_\gamma(\Gamma)!}\right \vert \right)^2 .
    \end{align}
Note that the summed quantity inside the square,
\begin{align}\label{eq:kpcriteria_raizcuadrada}
    \sum_{\Gamma \nsim v} \left \vert   \varphi(\Gamma) \prod_{\gamma \in \Gamma}  \frac{   w_\gamma^{1/2}}{m_\gamma(\Gamma)!}\right \vert,
\end{align}
can be regarded as a modified cluster expansion with weights $\tilde{w}_\gamma \equiv w_\gamma^{1/2}$. Consequently, its convergence can be bounded using Con. \ref{con:clustersamplingcondition}, which is equivalent to Con. \ref{con:kotecky_preiss_criteria} for this modified cluster expansion (see Lemma \ref{lem:kotecky_preiss}). Given our choices of $a(\gamma)$ and $d(\gamma)$ in Con. \ref{con:clustersamplingcondition} immediately implies:
\begin{align}
    \mathbb{E}_{p(\Gamma)}[\vert \xi (\Gamma)\vert^2] \leq& \frac{1}{A_M} \sum_{v \in \tilde{V}(\tilde{G})} \sum_{\Gamma \nsim v} \left \vert   \varphi(\Gamma) \prod_{\gamma \in \Gamma}  \frac{   w_\gamma^{1/2}}{m_\gamma(\Gamma)!}\right \vert\\
    \leq&  \frac{1}{A_M} \cdot  \sum_{v \in \tilde{V}(\tilde{G})} \sum_{\Gamma \nsim v} \left \vert   \varphi(\Gamma) \prod_{\gamma \in \Gamma}  \frac{   w_\gamma^{1/2}}{m_\gamma(\Gamma)!}\right \vert  \leq \frac{1}{A_M}  \sum_{v\in \tilde{V}(\tilde{G})} a(v) \leq \frac{\vert \tilde{V}(\tilde{G}) \vert}{A_M},
\end{align}
concluding the proof.

\subsection{Mixing of the cluster Markov chain}\label{app_rapidmixingproof}

To bound the mixing time of the chain, we first construct a deterministic canonical path $\mathcal{P}_{XY}$ between any two valid states $X, Y \in \mathcal{G}_C$ of the cluster Markov chain defined in Sec. \ref{subsec::cluster_markov_chain}. The state space is represented by the transition graph $G_{MC} = (\mathcal{G}_{C, \vert{}\Gamma\vert{}\leq M}, E_{MC})$, where edges correspond to valid MCMC transitions (the addition or deletion of a single polymer that maintains cluster connectivity). To ensure the path is uniquely defined, we establish a strict total order, denoted by $\prec$, on the set of all allowed polymers $\mathcal{C}_{\vert{}\gamma\vert{}\leq M}$. This ordering is defined primarily by the absolute weights of the polymers. Specifically, for any two distinct polymers $\gamma_a, \gamma_b$, we define $\gamma_a \succ \gamma_b$ if $\vert{}w_{\gamma_a}\vert{} > \vert{}w_{\gamma_b}\vert{}$. If $\vert{}w_{\gamma_a}\vert{} = \vert{}w_{\gamma_b}\vert{}$, we break the tie using a secondary lexicographical spatial ordering of their constituent vertices to guarantee that $\prec$ is a strict total order. We emphasize that this ordering is a theoretical construct for the congestion proof and is not evaluated by the algorithm itself.

The canonical path $\mathcal{P}_{XY}$ is divided into three sequential phases, illustrated in Fig. \ref{fig:congestion_proof}

\begin{figure*}[t]
    \centering
    \includegraphics[width=0.785\textwidth]{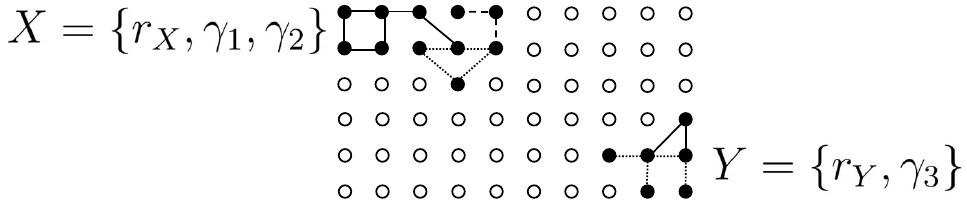}\\[2ex]
    \begin{subfigure}{0.31\textwidth}
        \centering
        \includegraphics[width=\textwidth]{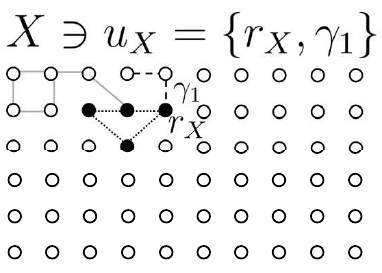}
        \caption{}
        \label{fig:congestion_peeling}
    \end{subfigure}\hfill
    \begin{subfigure}{0.31\textwidth}
        \centering
        \includegraphics[width=\textwidth]{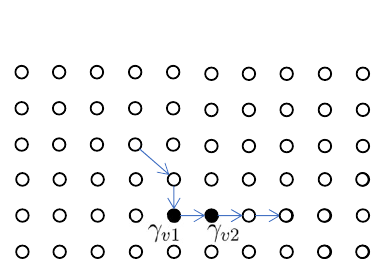}
        \caption{}
        \label{fig:congestion_transport}
    \end{subfigure}\hfill
    \begin{subfigure}{0.31\textwidth}
        \centering
        \includegraphics[width=\textwidth]{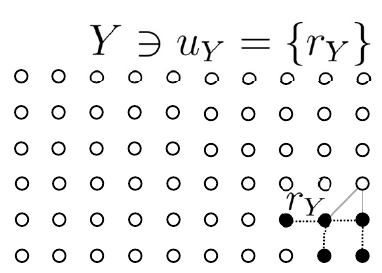}
        \caption{}
        \label{fig:congestion_construction}
    \end{subfigure}
 
    \caption{\justifying \textbf{Illustration of the canonical path
    $\mathcal{P}_{XY}$ between an initial cluster $X$ and a target cluster
    $Y$.} \textbf{Top:} the two clusters on the host graph, anchored at their
    roots $r_X$ and $r_Y$. \textbf{(a)} A step in the peeling phase, where the
    intermediate cluster $u_X$ remains strictly contained within $X$
    ($u_X \subseteq X$). \textbf{(b)} A step in the transport phase, where the
    intermediate state contains at most two polymers propagating along the
    geodesic in the connected covering graph $\tilde{G}(\tilde{V}, \tilde{E})$
    between the roots $r_X$ and $r_Y$. \textbf{(c)} A step in the construction
    phase, where the intermediate cluster $u_Y$ is assembled such that it
    remains strictly contained within $Y$ ($u_Y \subseteq Y$). In panels (a)
    and (c), the inactive polymers---those already removed from $X$ or yet to
    be added to $Y$---are shaded in gray to illustrate the full extent of the
    initial and target clusters, respectively.}
    \label{fig:congestion_proof}
\end{figure*}

\begin{itemize}
    \item \textit{1. Peeling Phase (Deconstruction)} Given the initial cluster $X$, we reduce it to a single base polymer $r_X$ through a unique sequence of deletions. At any step, let $C$ be the current connected cluster. We define the set of removable polymers $L(C) \subseteq C$ as the set of polymers whose deletion does not disconnect the cluster (graph theoretically, these are the non-cut vertices of the incompatibility graph of $C$, which are mathematically guaranteed to exist). We uniquely select the polymer $\gamma^* = \max_{\prec} L(C)$ and delete it. This ensures we always remove the heaviest possible valid polymer first. We iterate this deterministic removal process until the cluster is reduced to a single polymer, denoted as $r_X$.
    \item \textit{2. Transport Phase (Translation)} This phase moves the single polymer $r_X$ to the target base polymer $r_Y$ (where $r_Y$ is the result of applying the peeling procedure to $Y$). To do this uniquely, we select a representative spatial vertex for both polymers. Let $v_X \in r_X$ and $v_Y \in r_Y$ be their respective minimal vertices according to the fixed spatial ordering of the underlying connected covering set of vertices $ \tilde{G}(\tilde{V},\tilde{E})$ of the polymer model. We then select the unique shortest path $p_{X \to Y} = (v_X, v_1, \dots, v_Y)$ in $ \tilde{G}(\tilde{V},\tilde{E})$ (breaking ties via the spatial ordering). The sequence of transitions in $G_{MC}$ is uniquely determined by walking this geodesic: we alternately add the minimal polymer associated with the next vertex in the path, and delete the previous polymer, maintaining connectivity at all times until $r_Y$ is reached.
    \item \textit{3. Growing Phase (Reconstruction)} Having reached $r_Y$, we reconstruct the target cluster $Y$. Because the peeling phase on $Y$ generates a strict, unique sequence of single-polymer deletions from $Y$ down to $r_Y$, we simply apply the exact time-reversed sequence of polymer additions. By definition, each addition maintains connectivity, generating a valid path in $G_{MC}$ that successfully terminates at state $Y$.
\end{itemize}
\begin{rem}[Path Uniqueness and Invertibility]\label{rem:path_invertibility}
    By construction, the canonical path $\mathcal{P}_{XY}$ is deterministic and strictly unique for any given pair of initial and target clusters $X, Y \in \mathcal{G}_C$. Furthermore, the path generation is mathematically invertible: given any intermediate transition $e=(u,v) \in \mathcal{P}_{XY}$ and the sets of inactive polymers (those already deleted from $X$ or yet to be added to $Y$), the mapping back to the exact pair $(X, Y)$ is injective.
\end{rem}

\begin{lemma}
    The maximum congestion $\bar{\rho}$ (See Eq. \eqref{eq:bar_rho_congestion}) of the Cluster Markov chain $P_{MC}$ is bounded as

    \begin{align}
        \bar{\rho} = \max_{e=(u,v) \in E} \frac{1}{p(u) P(u,v)} \sum_{\mathcal{P}_{xy} \ni e} p(x) p(y) \leq \mathcal{\tilde{O}}(\vert \tilde{V}(\tilde{G})\vert \varepsilon^{-1})
    \end{align}
    \label{lemma_congestionbound}
\end{lemma}
\begin{proof}(Lem. \ref{lemma_congestionbound})
    To bound the maximum congestion $\bar{\rho}$ we employ the canonical paths method. Given any two clusters $X$ and $Y$ we consider the deterministic canonical path written above and evaluate $\bar{\rho}$ at peeling, transport and growing phases.

    \textit{Peeling phase:}
    At this stage any transition $(u,v)$ is the deletion of a polymer from $u$, yielding $v$, both clusters contained in $u,v \subseteq X$. This yields

        \begin{align}
        \bar{\rho} \leq& \max_{e=(u,v) \in E(G_{MC})} \frac{1}{p(u) P(u,v)} \cdot \left( \sum_{X \ni u} \prod_{\gamma \ni u} \vert w_\gamma^{m_\gamma (X)}   \vert \cdot \frac{1}{A_M} \right) \cdot \cancelto{1}{\left( \sum_{Y} p(Y) \right)} \leq   \max_{u \in \mathcal{G}_C} \frac{ \sum_{X \ni u} \prod_{\gamma \in X} \vert w_\gamma^{m_\gamma (X)}   \vert }{\prod_{\gamma \in u} \vert w_\gamma \vert^{m_\gamma (u)}   \cdot \frac{1}{ 2 \vert u \vert^2}} \cdot \\
        \leq& 2 \log^2 (\vert \tilde{V}(\tilde{G})\vert \varepsilon^{-1}) \cdot  \max_{u \in \mathcal{G}_C} \sum_{X \setminus u } \prod_{\gamma \in X \setminus u} \vert w_\gamma \vert^{m_\gamma (X \setminus u)} , 
    \end{align}
    where we used $P(u,v) \geq 1/\vert u \vert^2$, being $\vert u \vert \leq \mathcal{O}(\log (\Vt \varepsilon^{-1}))$ due to the truncation (recall the cluster dynamics of Sec. \ref{subsec::cluster_markov_chain}). 
    At this point we define the polymer $\gamma_u$ whose support is the same of $u$, i.e. supp$(\gamma_u)=$ supp$(u)$. We realize that this last sum can be upper bounded similarly as how we bound $G_{n}(\gamma)$ in App. \ref{app:A_convergence}; decomposing $X\setminus u$ in sets of clusters $\{\tilde{\Gamma}_i\}$ connected to $\gamma_u$, and summing over all possible sets), yielding

        \begin{align}
        \bar{\rho} \leq& 2 \log^2 (\vert \tilde{V}(\tilde{G})\vert \varepsilon^{-1}) \cdot  \max_{u \in \mathcal{G}_C} \sum_{m=0}^\infty  \frac{1}{m!}  \sum_{\tilde{\Gamma}_0 \dots \tilde{\Gamma}_m} \prod_i  \prod_{\gamma \in \tilde{\Gamma}_i} \vert w_\gamma \vert \leq 2 \log^2 (\vert \tilde{V}(\tilde{G})\vert \varepsilon^{-1}) \cdot \sum_{m=0}^\infty  \frac{1}{m!}  \left( \sum_{e \nsim \gamma_u} \vert w_e \vert \mathrm{e}^{a(e)} \right)^m \\
        \leq& 2 \log^2 (\vert \tilde{V}(\tilde{G})\vert \varepsilon^{-1}) \cdot \mathrm{e}^{a(\gamma_u)} \leq 4 \log^2 (\vert \tilde{V}(\tilde{G})\vert \varepsilon^{-1}) \cdot \vert \tilde{V}(\tilde{G}) \vert \varepsilon^{-1} \equiv \mathcal{\tilde{O}}(\vert \tilde{V}(\tilde{G}) \vert \varepsilon^{-1} ),
    \end{align}
where in the last equation we used that because of the truncation $a(\gamma_v) \leq \frac{1}{2} \cdot 2  \log (2\vert  \tilde{V}(\tilde{G}) \vert \varepsilon^{-1} )=\log (2\vert  \tilde{V}(\tilde{G}) \vert \varepsilon^{-1} )$.
\begin{rem}\label{rem_achoice}
Since $\gamma_u$ is supported on $\text{supp}(u)$, its size is only limited by
the cluster cutoff, $\norm{\gamma_u} \leq 2M$, and the bound above carries a
factor $\mathrm{e}^{a(\gamma_u)}$. Our choice $a(\gamma) = \norm{\gamma}/2$
keeps it at $\mathrm{e}^{M} = \vert \tilde{V}(\tilde{G})\vert \varepsilon^{-1}$,
a factor already present in the congestion bound, whereas the conventional
choice $a(\gamma) = \norm{\gamma}$ would square it. More general choices
$a(\gamma) = c \norm{\gamma}$ are of course admissible, trading the convergence
threshold of Con. \ref{con:kotecky_preiss_criteria} against a polynomial
overhead $(\vert \tilde{V}(\tilde{G})\vert \varepsilon^{-1})^{2c}$ in the mixing
time.
\end{rem}

\textit{Transport phase:} We now realize that, due to Con. \ref{ass::uniform_local_weights}, we have connected to each $v\in \tilde{V}(\tilde{G})$ a $\mathcal{O}(1)$ size polymer $\gamma$ with $\vert w_\gamma \vert \geq w_{\text{c}} \equiv \mathcal{O}(1)$. As the canonical path chooses the largest polymer anchored at each vertex through the geodesic we conclude that we never use a polymer in the transport phase with $\vert w_\gamma \vert \leq w_{\text{c}} \equiv \mathcal{O}(1)$. Thus

\begin{align}
    \bar{\rho} \leq \frac{\cancelto{1}{\left(\sum_{X} p (X) \sum_{Y} p (Y) \right)} }{w_{\text{c}}^2 \cdot\frac{1}{A_M} \cdot \frac{1}{2|u|^2}} \leq \mathcal{\tilde{O}}(\vert \tilde{V}(\tilde{G}) \vert)
\end{align}
Note that we could have relaxed Con. \ref{ass::uniform_local_weights}, and use the minimum provided by truncation $\vert  w_\gamma  \vert \geq \kappa (\gamma)$, which would incur a poly$(\vert \tilde{V} \vert, \varepsilon^{-1})$ overhead dictated by $\kappa (\gamma)$. However, for many cases of interest, (such as ours in Sec. \ref{sec::aplication}) Con. \ref{ass::uniform_local_weights} will be satisfied.

\textit{Construction phase:} In this last phase, the canonical path traverses polymer additions rather than deletions. While the forward addition probability $P(u,v)$ carries a weight factor $\vert w_\gamma \vert$ that can be as small as the algorithmic threshold $\kappa(\gamma)$, we seamlessly bypass this dependence by exploiting the detailed balance of the Markov chain. Because the chain is strictly reversible, the probabilistic flow is symmetric: $p(u)P(u,v) = p(v)P(v,u)$. By evaluating the congestion denominator in the reverse (deletion) direction, the relevant transition probability is governed by $P(v,u) \ge 1/\vert v \vert^2$. Consequently, bounding the congestion for the construction phase becomes mathematically identical to the peeling phase, yielding the exact same upper bound.

\end{proof}

\noindent Given Lem. \ref{lemma_congestionbound} we are in position to bound the mixing time

\begin{proof}(Thm. \ref{thm_mixingtime_clusterexp})

We recall Eq. \eqref{eq:mixing_time_bound};

\begin{align}
    t_{\text{mix}} \leq \bar{\rho} \cdot L \cdot \log \left(\frac{1}{p_{\text{min}} \, \varepsilon_{\text{mix}}}\right).
\end{align}
    The maximum length of the canonical path, $L$, can be bounded as
    \begin{align}
        L \leq& (\text{max. steps of peeling phase}) + (\text{max. steps of transport phase}) + (\text{max. steps of construction phase}) \\
        \leq& \mathcal{O}(\log (\vert \tilde{V}(\tilde{G}) \vert \varepsilon^{-1} ))+ \text{diam}(\tilde{G}(\tilde{V},\tilde{E})) + \mathcal{O}(\log (\vert \tilde{V}(\tilde{G}) \vert \varepsilon^{-1} ))  \leq \mathcal{\tilde{O}}(\text{diam}(\tilde{G}(\tilde{V},\tilde{E}))),
    \end{align}
    where $\text{diam}(\tilde{G}(\tilde{V},\tilde{E}))$ is the maximum diameter of the graph $\tilde{G}(\tilde{V},\tilde{E})$, which upper bounds the length of any geodesic from any point to any other.

    On the other hand, a valid lower bound on the minimum stationary probability corresponds to a valid cluster reaching the maximum truncation size $\norm{\Gamma} = M$, where every constituent polymer attains the minimum algorithmic weight threshold $\kappa(\gamma)$. Recalling the target distribution $p(\Gamma) = \frac{1}{A_M} \prod_{\gamma \in \Gamma} \vert w_\gamma \vert$, we bound this worst-case state as:
\begin{align}
    p_{\text{min}} = \min_{\Gamma \in \mathcal{G}_{\text{alg}} } p(\Gamma) \geq \frac{1}{A_M} \kappa_{\text{min}}^{2M},
\end{align}
where $\kappa_{\text{min}}$ absorbs the polynomial lower bounds of the algorithmic threshold (recall Sec. \ref{subsection_algorithmicpolymermodel}). While $p_{\text{min}}$ itself is vanishingly small, the mixing time penalty depends strictly on its logarithm, which neatly collapses this exponential dependence
\begin{align}
    \log \left( \frac{1}{p_{\text{min}}} \right) \leq \log(A_M) + 2M \log \left( \frac{1}{\kappa_{\text{min}}} \right). \label{eq:pmin_log_collapse}
\end{align}
Since the truncation size is $M = \mathcal{O}(\log(\vert \tilde{V}(\tilde{G}) \vert \varepsilon^{-1}))$ and $\kappa_{\text{min}}$ scales at worst as $\left(N^{-1} \varepsilon\right)^{\mathcal{O}(1)}$, this entire expression scales strictly as $\mathcal{O}(\log^3(\vert \tilde{V}(\tilde{G}) \vert \varepsilon^{-1})) \equiv \mathcal{\tilde{O}}(1)$. This, together with Lem. \ref{lemma_congestionbound} yields

\begin{align}
    t_{\text{mix}} \leq \mathcal{\tilde{O}}( \text{diam}(\tilde{G}(\tilde{V},\tilde{E})) \cdot  \vert \tilde{V}(\tilde{G}) \vert \varepsilon^{-1} ).
\end{align}
     
\end{proof}

\subsection{Auxiliary results for annealing procedure}

\subsubsection{Samples from modified distributions $p_\theta$}
 
\begin{lemma}[Annealed samples by filtering]
\label{lem_app_AnnealingFilter}
Let the polymer model satisfy Con. \ref{con:kotecky_preiss_criteria} and Con. \ref{ass::uniform_local_weights}, and let $\theta\ge0$. Consider
the procedure that draws $\Gamma$ from $p_0$ using the Cluster Markov chain of
Sec.\ref{subsec::cluster_markov_chain}, accepts it with probability
\begin{equation}
  f_\theta(\Gamma)\;=\;\mathrm{e}^{-\theta\left(\|\Gamma\|-r\right)},
  \label{eq:filterdef}
\end{equation}
and on rejection draws a fresh sample and repeats. Then the accepted cluster is
distributed exactly as $p_\theta$, and the expected number of draws per accepted
sample is at most $1/a_*$, where
\begin{equation}
  a_*\;=\;\wc\;=\;\Omega(1)
  \label{eq:astar}
\end{equation}
is independent of $\theta$, of $M$ and of $\Vt$.
\end{lemma}

\begin{proof}
The laws $p_0$ and $p_\theta$ differ only through the factor
$\mathrm{e}^{-\theta\|\Gamma\|}$,
\begin{align}
    \frac{p_\theta(\Gamma)}{p_0(\Gamma)} = \frac{A_M(0)}{A_M(\theta)}\,
    \mathrm{e}^{-\theta\|\Gamma\|} .
\end{align}
By Con. \ref{ass::uniform_local_weights} every cluster of $\GC$ with
non-vanishing weight satisfies $\|\Gamma\|\ge r$, so this ratio is maximal on
the clusters of minimal size. Normalising by that maximum yields
\eqref{eq:filterdef}, which therefore takes values in $(0,1]$ and is a valid
acceptance probability.

The mean acceptance is
\begin{align}
  a_\theta&=\sum_{\Gamma\in\GC}p_0(\Gamma)\,\mathrm{e}^{-\theta(\|\Gamma\|-r)}\notag\\
  &=\frac{\mathrm{e}^{r\theta}}{A_M(0)}\sum_{\Gamma\in\GC}\prod_{\gamma\in\Gamma}
  \left|w_\gamma^{m_\gamma(\Gamma)}(\theta)\right|
  =\frac{\mathrm{e}^{r\theta}A_M(\theta)}{A_M(0)} ,
  \label{eq:arho}
\end{align}
so that, conditioning on acceptance,
\begin{align}
   \Pr[\Gamma \mid \text{accepted}] = \frac{p_0(\Gamma)\,f_\theta(\Gamma)}{a_\theta}
    = \frac{\prod_{\gamma\in\Gamma}\left|w_\gamma^{m_\gamma(\Gamma)}(\theta)\right|}
           {A_M(\theta)} = p_\theta(\Gamma),
\end{align}
the factors $\mathrm{e}^{r\theta}$ and $A_M(0)$ cancelling between numerator and
denominator. The procedure thus returns exact samples from $p_\theta$.

It remains to bound the acceptance probability from below. Discarding in
\eqref{eq:arho} all clusters of size larger than $r$, whose contributions are
non-negative, and using that $f_\theta\equiv1$ on those of size exactly $r$,
\begin{align}
    a_\theta \;\geq\; \sum_{\Gamma\,:\,\|\Gamma\|=r} p_0(\Gamma)
    \;=\; \frac{B_r}{A_M(0)} \;\geq\; \frac{\wc\Vt}{\Vt} \;=\; a_* ,
\end{align}
where we used $B_r \geq \wc \Vt$ from \eqref{eq:Brchain} and
$A_M \leq \Vt$ from Lem. \ref{lemma_Aextensive}. The number of draws per
accepted sample is geometric with success probability $a_\theta\ge a_*$, hence has
mean at most $1/a_*$.
\end{proof}

\subsubsection{Proof of Lemma \ref{lemm:Wi_expectedvalues_v2}}\label{app_annealingproofs}
Recall $W_i = \prod_{\gamma \in \Gamma} \left\vert \frac{w_\gamma(i)}
{w_\gamma(i+1)} \right\vert^{m_\gamma(\Gamma)} = \mathrm{e}^{\norm{\Gamma}}$
with $\Gamma \sim p_{i+1}$. It follows from the definition of $p_\theta$ that
\begin{align*}
\mathbb{E}_{p_{i+1}}[W_i] &= \sum_{\Gamma \in \mathcal{G}_C} p_{i+1}(\Gamma)
   W_i(\Gamma) = \frac{1}{A_M(i+1)} \sum_{\Gamma \in \mathcal{G}_C} \prod_{\gamma \in \Gamma}
   \vert w_\gamma \vert^{m_\gamma(\Gamma)} \mathrm{e}^{-(i+1)\norm{\Gamma}}
   \mathrm{e}^{\norm{\Gamma}} = \frac{A_M(i)}{A_M(i+1)} ,
\end{align*}
and that
\begin{align}
\mathbb{E}_{p_{i+1}}[W_i^2] &= \sum_{\Gamma \in \mathcal{G}_C} p_{i+1}(\Gamma) W_i(\Gamma)^2 = \frac{1}{A_M(i+1)} \sum_{\Gamma \in \mathcal{G}_C} \prod_{\gamma \in \Gamma} \vert w_\gamma \vert \mathrm{e}^{-(i+1)\norm{\gamma}} \prod_{\gamma \in \Gamma} \mathrm{e}^{2\norm{\gamma}} \\
&= \frac{1}{A_M(i+1)} \sum_{\Gamma \in \mathcal{G}_C} \prod_{\gamma \in \Gamma} \vert w_\gamma \vert \mathrm{e}^{-(i-1)\norm{\gamma}} = \frac{A_M(i-1)}{A_M(i+1)}.
\end{align}
Since $W_0, \dots, W_{\ell-1}$ are mutually independent, we obtain

\begin{align}
    \mathbb{E}[W] = \prod_{i=0}^{\ell-1} \mathbb{E}[W_i] = \prod_{i=0}^{\ell-1} \frac{A_M(i)}{A_M(i+1)} = \frac{A_M(0)}{A_M(\ell)}
\end{align}
and

\begin{align}
    \mathbb{E}[W^2] = \prod_{i=0}^{\ell-1} \mathbb{E}[W_i^2] = \prod_{i=0}^{\ell-1} \frac{A_M(i-1)}{A_M(i+1)} = \frac{A_M(-1)A_M(0)}{A_M(\ell)A_M(\ell-1)}.
\end{align}

\subsection{Computation of Ursell functions}\label{app:ursell_calculation_cost}

     \begin{lemma}\label{lemm::calculation_ursell}
Let the polymer model  $(\mathcal{C}, w, \sim)$ satisfy Con. \ref{con:kotecky_preiss_criteria}. Then the Ursell function $\varphi(\Gamma)$ can be computed in expected runtime $\mathcal{O}(\vert \tilde{V}(\tilde{G}) \vert /A_M)$.
\end{lemma}

\begin{proof}
Our proof follows the one of \cite{mann2024algorithmic}. For a connected graph $H$, we have
\begin{equation}
    \varphi(\Gamma) =  \sum_{\substack{S \subseteq E(H_\Gamma) \\ \text{spanning} \\ \text{connected}}} (-1)^{|S|} = (-1)^{|V(H_\Gamma)|-1} T_H(1,0),
\end{equation}
where $T_H(x, y)$ denotes the Tutte polynomial of $H$ defined by
\begin{equation}
    T_H(x, y) := \sum_{S \subseteq E} (x - 1)^{k(S) - k(E)} (y - 1)^{k(S) + |S| - |H_\Gamma|}.
\end{equation}
Here $k(S)$ denotes the number of connected components of the subgraph with edge set $S$. The Ursell function can then be computed in time $3^{\vert H_\Gamma \vert} \vert H_\Gamma \vert^{\mathcal{O}(1)}$ by evaluating the Tutte polynomial via the algorithm of \cite{bjorklund2008computingtuttepolynomialvertexexponential}. Then, by Lem. \ref{lemma_Aextensive}
\begin{align*}
    \mathbb{E}_{p(\Gamma)} [\text{runtime of computing }\varphi(\Gamma)] \leq \frac{1}{A_M} \sum_{\Gamma \in \mathcal{G}_C} 4^{\norm{\gamma}} \prod_{\gamma \in \Gamma } \vert w_\gamma \vert^{m_\gamma (\Gamma)} \leq \frac{\vert \tilde{V}(\tilde{G}) \vert}{A_M}.
\end{align*}
Our choice of scaling parameter $d(\gamma) = 2 \log 2 \cdot \norm{\gamma}$ (yielding $\mathrm{e}^{d(\gamma)} = 4^{\norm{\gamma}}$) is designed to exponentially dominate the $\mathcal{O}(3^{\vert H_\Gamma \vert} \vert H_\Gamma \vert^{\mathcal{O}(1)})$ runtime of the Tutte polynomial evaluation, exploiting the bound $\vert H_\Gamma \vert \le \norm{\Gamma}$. We note that the strict algebraic inequality $3^n n^{\mathcal{O}(1)} \le 4^n$ holds only for polymer sizes $n$ exceeding a constant threshold $n_0$. For small polymers ($n < n_0$), the ratio between the exact runtime and our $4^n$ bound is strictly maximized by a finite universal constant. This $\mathcal{O}(1)$ prefactor is trivially absorbed into the overarching asymptotic notation ($\mathcal{O}$ or $\mathcal{\tilde{O}}$) of our final algorithmic complexity.
\end{proof}

\subsection{Confidence amplification}\label{app:amplification}

The proof of Thm.\ref{thm:sampling_cexp} fixes the internal failure
probabilities to constants and then boosts the resulting constant confidence to
an arbitrary $1-\delta$. We record the standard argument here, and explain why a
Chernoff bound is the appropriate tool at this particular step.
 
\begin{lemma}[Median amplification]\label{lem:median_amplification}
    Let $\mathcal{A}$ be a randomized algorithm that, for a target precision
    $\varepsilon > 0$, outputs an estimate $\widehat{x}$ of a quantity
    $x \in \mathbb{C}$ satisfying $\vert \widehat{x} - x \vert \leq \varepsilon$
    with probability at least $2/3$, and whose expected runtime is
    $\bar{T}(\varepsilon)$. Then, for every $\delta \in (0,1)$, there is an
    algorithm with expected runtime
    $\mathcal{O}\big(\bar{T}(\varepsilon/\sqrt{2}) \log \delta^{-1}\big)$ whose
    output satisfies $\vert \widehat{x} - x \vert \leq \varepsilon$ with
    probability at least $1-\delta$.
\end{lemma}
 
\begin{proof}
    Run $\mathcal{A}$ independently $K$ times at precision $\varepsilon/\sqrt2$,
    with $K$ odd and $K \geq 18 \log (2\delta^{-1})$, obtaining estimates
    $\widehat{x}^{(1)}, \dots, \widehat{x}^{(K)}$, and output 
    \begin{align}
        \widehat{x} = \operatorname{med}_j \big\{ \operatorname{Re}
        \widehat{x}^{(j)} \big\} + i \operatorname{med}_j \big\{
        \operatorname{Im} \widehat{x}^{(j)} \big\},
    \end{align}
being $\operatorname{med}$ the median. Since the repetitions are independent executions of the same algorithm,
    linearity of expectation gives an expected runtime
    $K \cdot \bar{T}(\varepsilon/\sqrt{2})$.
 
    Consider the real parts. The indicator variables for a successful estimate,
    $X_j = \mathbb{I} \big[ \vert \operatorname{Re}(\widehat{x}^{(j)} - x) \vert \leq \varepsilon/\sqrt{2} \big]$,
    are independent Bernoulli variables with success probability $p \geq 2/3$. Let $S = \sum_j X_j$ be the number of successful estimates, with expected value $\mu = pK$. The standard Chernoff--Hoeffding bound states that $\Pr[S \le \mu - t] \le \mathrm{e}^{-\frac{2t^2}{K}}$. The condition for the median to fall within the error interval is $S > K/2$. Setting the deviation $t = \mu - K/2 = K(p-1/2)$, we obtain
    \begin{align}
        \Pr\Big[ \textstyle\sum_j X_j \leq K/2 \Big]
        \leq \mathrm{e}^{-2K(p-1/2)^2} \leq \mathrm{e}^{-K/18}
        \leq \delta/2 .
    \end{align}
    Whenever more than half of a finite set of reals lie in an interval, so does
    their median: were the median below the target interval, at least $\lceil K/2 \rceil$ of the
    values would lie below it as well, leaving fewer than $K/2$ inside it, and
    symmetrically from above. Hence,
    $\vert \operatorname{med}_j \{\operatorname{Re}\widehat{x}^{(j)}\} -
    \operatorname{Re} x \vert \leq \varepsilon/\sqrt{2}$ except with probability at most
    $\delta/2$. The same bound applies to the imaginary parts. Since each coordinate fails with probability at most $\delta/2$, the total probability that either median exceeds its allowed error is bounded by $\delta$. When both coordinates succeed, the total error in the complex plane is bounded by
    $\vert \widehat{x} - x \vert^2 \leq 2 (\varepsilon/\sqrt{2})^2 = \varepsilon^2$, yielding the claim.
\end{proof}
We apply Lem.\ref{lem:median_amplification} to the final estimator
$\widehat{T}_M = \widehat{\xi} \cdot \widehat{A_M}$.
 
\begin{rem}[Why a Chernoff bound at this step]\label{rem:why_chernoff}
    Throughout this work sample complexities are controlled with Chebyshev's
    inequality, Eq.~\eqref{eq:sample_complexity}, which requires only a bound on
    the variance. The reason for a different tool here is structural. The
    variables $X_j$ above are indicator functions, hence bounded in $\{0,1\}$ by
    construction, which is precisely the hypothesis a Chernoff bound requires,
    and the exponential tail it provides is what converts a $\delta^{-1}$
    dependence into $\log \delta^{-1}$.
    Applying a Chernoff--Hoeffding bound directly to $\xi$ instead would be
    substantially weaker, because such bounds are governed by the \emph{range}
    of the variable rather than by its variance. While
    Lem.\ref{lemma_finitevariance} guarantees a second moment
    $\mathbb{E}[\vert \xi \vert^2] = \mathcal{O}(1)$, the range of $\xi$ is not
    $\mathcal{O}(1)$: for a cluster of size $\norm{\Gamma} \leq 2M$ the
    incompatibility graph $H_\Gamma$ has at most $\binom{M}{2}$ edges, so in general
    $\vert \varphi(\Gamma)\vert \leq 2^{\mathcal{O}(M^2)}$, which for our choice
    of $M$ is quasi-polynomial in $\vert V(G)\vert \varepsilon^{-1}$. However, using
    Chebyshev's inequality within each repetition and a Chernoff bound across
    repetitions yields the desired exponential decrease in runtime $\delta^{-1} \rightarrow \log \delta^{-1}$.
\end{rem}

    \section{Long-range quantum interactions: Proofs}

      \subsection{Long-range interactions: Proof of Koteck\'y-Preiss criteria} \label{app:kp_longrange}

	\begin{proof}(Lem. \ref{lem::koteckypreiss_longrange})
	    Let $a(\gamma)= \norm{\gamma}/2$ and $d(\gamma)= 2\log2 \cdot \norm{\gamma}$. To prove Koteck\'y-Preiss criteria it suffices to show that for any $v\in G(V,E)\equiv Z$

        \begin{equation}\label{eq:app_koteckypreiss_longrange_1st}
            \sum_{\gamma \nsim Z^*}  \vert w_\gamma \vert \mathrm{e}^{\frac{1}{2}\norm{\gamma}  + 2\log2\norm{ \gamma}} \leq \frac{1}{2},
            \end{equation}

            so it is necessary to obtain a bound on $\vert w_\gamma \vert $.
	
	\begin{lemma}\label{lemma_boundweight}
		Let $H$ be a Hamiltonian as in Sec. \ref{sec:quantumHamilton}, then

        \begin{equation}\label{eq:app_weightbound_longrange}
			\vert w_\gamma \vert \leq \sum_{n= \norm{\gamma}}^\infty\stirlingII{n}{\norm{\gamma}} \frac{\norm{\gamma}!}{n!} \vert \beta g \vert^n \prod_{Z \in \gamma} \frac{\norm{h_Z}}{g} \leq 
			\left(\mathrm{e}^{\vert \beta g \vert}-1\right)^{\norm{\gamma}} \prod_{Z \in \gamma} \frac{\norm{h_Z}}{g},
		\end{equation}
	\end{lemma}

	\begin{proof}
        The proof proceeds similarly to the one of Lemma 23 of \cite{mann2024algorithmic}. Let $P$ denote the set of all sequences of edges in $\gamma$.
        By the Taylor series
		\begin{equation}
			\vert w_\gamma \vert \leq \norm{\sum_{T \subseteq E(\gamma)} (-1)^{\vert T \vert} \mathrm{e}^{-\beta \sum_{Z \in T } h_Z}} \leq \sum_{\substack{\omega \in P \\ \text{supp}(\omega)= \gamma}} \frac{\vert \beta \vert^{\vert \omega \vert}}{\vert \omega \vert !} \prod_{Z \in \omega} \norm{h_Z} = \sum_{\substack{\omega \in P \\ \text{supp}(\omega)= \gamma}} \frac{\vert \beta g \vert^{\vert \omega \vert}}{\vert \omega \vert !} \prod_{Z \in \omega} \frac{\norm{h_Z}}{g}.
		\end{equation}
		There are precisely $\stirlingII{n}{\norm{\gamma}}\,\norm{\gamma}!$ sequences $\omega$ of length $n$ whose support is $\gamma$, where $\stirlingII{n}{k}$ denotes the Stirling number of the second kind. Hence
		
		\begin{equation}
			\vert w_\gamma \vert \leq \sum_{n= \norm{\gamma}}^\infty\stirlingII{n}{\norm{\gamma}} \frac{\norm{\gamma}!}{n!} \vert \beta g \vert^n \prod_{Z \in \gamma} \frac{\norm{h_Z}}{g} \leq 
			\left(\mathrm{e}^{\vert \beta g \vert}-1\right)^{\norm{\gamma}} \prod_{Z \in \gamma} \frac{\norm{h_Z}}{g},
		\end{equation}
where in the last inequality we used the fact that $\norm{h_Z}/ g \leq 1$.
	\end{proof}

    Plugging Lem. \ref{lemma_boundweight} in Eq. \eqref{eq:app_koteckypreiss_longrange_1st} yields

	\begin{equation}\label{eq:app_koteckypreiss_longrange_2nd}
		\sum_{\gamma \nsim Z^*}  \vert w_\gamma \vert \mathrm{e}^{\norm{\gamma}  + 2\log2\norm{ \gamma}} \leq \sum_{m=1}^{\infty} \left(\mathrm{e}^{\left(\frac{1}{2}+ 2\log2 \right)}\left(\mathrm{e}^{\vert \beta g \vert}-1\right)\right)^{m}   \underbrace{\sum_{\substack{\gamma \nsim Z^* \\ \norm{\gamma} =m}}  \prod_{Z \in \gamma} \frac{\norm{h_Z}}{g}}_{\mathbf{S}_{Z^*}} 
	\end{equation}
To upper bound $\mathbf{S}_{Z^*}$, it is convenient to divide the product into $k$ subpolymers $\{\gamma_j\}_{j=1}^k$ with multiplicities $\{m_j\}$, each connected to vertices $\{i_j\}_{j=1}^k$ . To sum over all connected polymers of cardinality $M$ connected to $Z^*$ it suffices to count over all the sets of polymers connected to the vertices $\{i_j\}$ such that $M=\sum_j m_j$. This can always be done for all $\gamma$ (see Appendix A.1 of \cite{Segovia_PRXQuantum2025} for a more detailed explanation). This leads to a (tolerable) over-counting, as different combinations of the sub-polymers $\{\gamma_j\}_{j=1}^k$ can correspond to the same polymer $\gamma$. With this in mind, we write 
\begin{align}
    \mathbf{S}_{Z^*} \leq \sum_{\substack{m_1, \dots m_j, \dots m_k \\ \sum_j m_j =M}} \prod_{j}  \sum_{ \substack{\gamma \in \Upsilon_{i_j}\\ \norm{\gamma}=m_j }}  \prod_{Z \in \gamma} \frac{\norm{h_Z}}{g}, \\
    \label{eq::S_Z*_beforebound}
\end{align}
where we now sum over all the elements of the set $\Upsilon_{i_j}$ of polymers $\gamma$ connected to site $i_j$. To continue we thus need to control the sum over all polymers, or equivalently tuples of hyperedges, of certain multiplicity $m_i$ which are connected to a given site $i$,
\begin{align}
    \mathbf{S}_i  \equiv \sum_{ \substack{\gamma \in \Upsilon_{i}\\ \norm{\gamma}=m_i }}  \prod_{Z \in \gamma}\frac{\norm{h_Z}}{g}.
    \label{eq::S_i}
\end{align}
At this point it is important to see that each polymer $\gamma$ can be associated with an auxiliary graph, labeled as $\tilde{\gamma}$, which has an edge for each hyperedge of $\gamma$, and such that it connects all the hyperedges $Z$. We refer to it as the \emph{connectivity polymer}. As it is always possible to make this association we can, using that $\sum_{Z\ni e} \norm{h_Z}\leq g/(1+d_e)^{\alpha}$, write

\begin{align}
\label{eq:app_longrangekp_conpoly}
    \mathbf{S}_{i} \leq \sum_{\substack{\tilde{\gamma} \in \tilde{\Upsilon}_i\\ \norm{\tilde{\gamma}}=m_i}} \prod_{e \in \tilde{\gamma}} \sum_{Z \nsim e} \frac{\norm{h_Z}}{g} \leq \sum_{\substack{\tilde{\gamma} \in \tilde{\Upsilon}_i\\ \norm{\tilde{\gamma}}=m_i}} \prod_{e \in \tilde{\gamma}} (1+d_e)^{-\alpha}.
\end{align}
where $\tilde{\Upsilon}_i$ is analogously the set of all connectivity polymers $\tilde{\gamma}$ connected to site $i$. Since in Eq.~\eqref{eq:app_longrangekp_conpoly} we sum over all possible connectivity polymers and, within each, over all its associated polymers, this leads to an overestimation. The benefit is that we turn the sum into one over polymers defined on graphs rather than hypergraphs, which simplifies our task considerably. For simplicity, we refer to connectivity polymers $\tilde{\gamma}$ simply as polymers from now on.

At this point we use the long-range decay properties for $\alpha >D$ and a mapping from polymers to trees (For more details see \cite{Segovia_PRXQuantum2025}) so we can bound
	
	\begin{equation}
		\sum_{\gamma \nsim Z^*}  \vert w_\gamma \vert \mathrm{e}^{\left(1/2+ 2\log2 \right)\norm{\gamma}} \leq \sum_{m=1}^{\infty} \left(\mathrm{e}^{\left(1/2+ 2\log2 \right)} \cdot 4uk\cdot\left(\mathrm{e}^{\vert \beta g \vert}-1\right)\right)^{m},
	\end{equation}
	
	and consequently for $\beta < \beta' = g^{-1} \log(1+ ( 24\sqrt{\mathrm e}\,uk)^{-1})$
	
	\begin{equation}
		\sum_{\gamma \nsim Z^*}  \vert w_\gamma \vert \mathrm{e}^{\left(1+ 2\log2 \right)\norm{\gamma}} \leq \frac{1}{2}
	\end{equation}

	so taking an arbitrary $\gamma^*$ and choosing $a(\gamma) =\norm{\gamma}/2$ and $d(\gamma) = 2\log 2 \norm{\gamma}$ yields

    \begin{align}
        \sum_{\gamma \nsim \gamma^*} \vert w_\gamma \vert \mathrm{e}^{a(\gamma)+d(\gamma)} \leq \sum_{Z \in \gamma^*} \sum_{\gamma \nsim Z} \vert w_\gamma \vert \mathrm{e}^{a(\gamma)+d(\gamma)} \leq  \frac{\norm{\gamma}}{2} = a(\gamma^*),
    \end{align}

    concluding that Lem. \ref{lem:kotecky_preiss} is satisfied. Thereby for $\beta \leq \beta^*= (2g)^{-1} \log(1+ ( 24\sqrt{\mathrm e}\,uk)^{-1})$
    \begin{align}
        \vert \log Z - T_M \vert = \sum_{\substack{(\Gamma, m_\gamma ) \in \mathcal{G}_C  \\ \norm{\Gamma}\geq M}} \vert \varphi(\Gamma) \prod_{\gamma \in \Gamma} \frac{w_\gamma^{m_\gamma(\Gamma)}}{m_\gamma(\Gamma)!} \vert \leq  \sum_{\substack{(\Gamma, m_\gamma ) \in \mathcal{G}_C  \\ \norm{\Gamma}\geq M}} \left| \frac{2 \beta}{\beta'} \right|^{\norm{\Gamma}}  \cdot  \vert \varphi(\Gamma) \prod_{\gamma \in \Gamma} \frac{w_\gamma^{m_\gamma(\Gamma)}}{m_\gamma(\Gamma)!} [\beta = \beta'/2]  \vert \leq N \left| \frac{ \beta}{\beta^*} \right|^M .
    \end{align}

	\end{proof}

    \subsection{Long-range interactions: Proof of cluster sampling condition (Con. \ref{con:clustersamplingcondition})} \label{app:clustersamplingcondition_longrange}

	\begin{proof}(Cor. \ref{coro::clustersampling_longrange})
	    Let $a(\gamma)= \norm{\gamma}/2$ and $d(\gamma)= 2\log2 \cdot \norm{\gamma}$. Proving Con. \ref{con:clustersamplingcondition} is very similar to proving Con. \ref{con:kotecky_preiss_criteria}. More concretely it suffices to show that for any $v\in G(V,E)\equiv Z$

        \begin{equation}\label{eq:app_csc_longrange_1st}
            \sum_{\gamma \nsim Z^*}  \vert w_\gamma \vert^{1/2} \mathrm{e}^{\left(1+ 2\log2 \right)\norm{\gamma}} \leq 1,
            \end{equation}

            so it is necessary to obtain a bound on $\vert w_\gamma \vert $. Plugging Lem. \ref{lemma_boundweight} in Eq. \eqref{eq:app_csc_longrange_1st} yields

	\begin{equation}\label{eq:app_csc_longrange_2nd}
		\sum_{\gamma \nsim Z^*}  \vert w_\gamma \vert^{1/2} \mathrm{e}^{\left(1+ 2\log2 \right)\norm{\gamma}} \leq \sum_{m=1}^{\infty} \left(\mathrm{e}^{\left(1+ 2\log2 \right)}\left(\mathrm{e}^{\vert \beta g \vert}-1\right)^{m/2}\right)   \underbrace{\sum_{\substack{\gamma \nsim Z^* \\ \norm{\gamma} =m}}  \prod_{Z \in \gamma}  \left(\frac{\norm{h_Z}}{g}\right)^{1/2}}_{\mathbf{S'}_{Z^*}} .
	\end{equation}
We can repeat the steps from Eq. \eqref{eq:app_koteckypreiss_longrange_2nd} to Eq. \eqref{eq:app_longrangekp_conpoly}, with the difference that in this case the more strict definition of long range from Eq.~\eqref{eq:longrange_definition} and $\alpha >2D$ are needed. This yields
	\begin{equation}
		\sum_{\gamma \nsim Z^*}  \vert w_\gamma \vert^{1/2} \mathrm{e}^{\left(\frac{1}{2}+ 2\log2 \right)\norm{\gamma}} \leq \sum_{m=1}^{\infty} \left(\mathrm{e}^{2\left(\frac{1}{2}+ 2\log2 \right)} \cdot 16u^2k^2\cdot\left(\mathrm{e}^{\vert \beta g \vert}-1\right)^{1/2}\right)^{m},
	\end{equation}
	and consequently for $\beta < \beta''= g^{-1} \log(1+ ( 24\sqrt{\mathrm e}\,uk)^{-2})$
	
	\begin{equation}
		\sum_{\gamma \nsim Z^*}  \vert w_\gamma \vert^{1/2} \mathrm{e}^{\left(1+ 2\log2 \right)\norm{\gamma}} \leq 1,
	\end{equation}
    yielding Con. \ref{con:clustersamplingcondition};

    \begin{equation}
        \sum_{\gamma \nsim \gamma^*}  \vert w_\gamma \vert^{1/2} \mathrm{e}^{\left(1+ 2\log2 \right)\norm{\gamma}} \leq \sum_{Z \in \gamma^*} \frac{1}{2} = \frac{\norm{\gamma}}{2} = a(\gamma^*)
    \end{equation}

	\end{proof}

\subsection{Long range polymer sampler}

\subsubsection{Proof of Lem. \ref{lem:longrange_polymer_sampler}}
\label{app:longrangesampler}

The algorithm is given as input a vertex $v\in V(G,E)$ and it gives as output a polymer according the following distribution
\begin{align}
    \nu_v (\gamma)= \begin{cases}
         \vert w_\gamma \vert \quad  \forall \gamma \nsim v\\ 1- \sum_{\gamma \nsim v}  \vert w_\gamma \vert \quad  \text{if } \gamma = \emptyset
    \end{cases}
\end{align}
It proceeds as follows:

	\begin{widetext}
\begin{itemize}
		\item \textit{Step 1 (Maximum radius choice)} Sample an interaction radius $R$ according to the discrete probability distribution
\begin{equation}
    \Pr(R) = \frac{R^{-\tilde{\alpha}}}{\zeta(\tilde{\alpha})},
\end{equation}
where $\tilde\alpha$ is a free algorithmic parameter, denoted with a tilde to
distinguish it from the physical interaction exponent $\alpha$.

The admissible window for this parameter is restricted to $\bigl(D(k-1)+1,\ \alpha+1-D\bigr)$ due to algorithmic reasons (See Eqs. \eqref{eq:step23}, \eqref{eq_constrainAlphaTilde2}). It is non-empty if and only if $\alpha>kD$, which for two-body interactions ($k=2$) reads
$\alpha>2D$. Thus, while Con. \ref{con:clustersamplingcondition} impose the restriction $\alpha>2D$, our polymer sampler is more restricted for the $k-$local case $\alpha > k D$.

After $R$ is fixed, for the subsequent steps, we restrict our algorithm to consider only Hamiltonian terms $h_Z$ whose maximum spatial extent satisfies $d(Z)=\max_{(i,j)\in Z} d_{ij} \leq R$. This truncation reduces our model to an ``effectively local'' subgraph $G_{\text{int},R}\subset G_{\text{int}}$ with a strictly bounded maximum degree $\Delta_R$.

To rigorously bound $\Delta_R$ for a $k$-local Hamiltonian, note that any term $Z'$ incompatible (i.e., overlapping) with a fixed term $Z$ must share at least one of its $\le k$ physical sites. Once this anchoring site is fixed, the remaining $j \le k-1$ sites of $Z'$ must be selected from the $c_D R^D$ sites available within a ball of radius $R$. Summing over all valid sizes of $Z'$ yields the upper bound
\begin{equation} \label{eq:DeltaR}
    \Delta_R \le k \sum_{j=0}^{k-1} \binom{c_D R^D}{j} = \mathcal{O}\!\left(R^{D(k-1)}\right).
\end{equation}

		\item \textit{Step 2 (Seed choice)} Let $\mathcal{R}_R(v) = \{ v' \in V(G_{\text{int},R}) : v' \nsim v \}$ denote the neighborhood of $v$ within the bounded interaction graph $G_{\text{int},R}$. For each term $Z$, we introduce the probability weight 
$q_Z=  \left( \mathrm{e}^{\beta g}-1 \right)\frac{\mathrm{e}\norm{h_Z}}{g \cdot \text{Pr}(d_Z \leq R)}$. We select a term $u \in \mathcal{R}_R(v)$ to act as a seed with probability $\text{Pr}(u) = q_u$, returning $\emptyset$ with the residual probability $1 - \sum_{u \in \mathcal{R}_R(v)} q_u$. This is achieved by drawing $U \sim \text{Unif}[0,1)$ and selecting $u$ via the cumulative distribution of $q_u$. Lemma \ref{lem:appPolymerSampler_qwell} guarantees that this probability distribution is strictly well-defined.

		\item \textit{Step 3: (weighted) BFS search:} Let $Q$ be an initially empty queue used to perform a Breadth-First Search (BFS) in $G_{\text{int},R}$. Set $\gamma \leftarrow\{u\} ,\, Q \leftarrow \{u\}$, and let $u$ be our seed.
        While the queue $Q \neq \emptyset$, repeat the following:
\begin{enumerate}
    \item Extract a vertex $v'$ from $Q$. Calculate on the fly $q_{v''}$ for all neighbors $v''$.
    \item For each unexplored neighbor $v''$ of $v'$ (restricted to the effective local graph defined by $R$), mark $v''$ as explored. Then, with probability $q_{v''}$, add $v''$ to $Q$.
\end{enumerate}\textbf{}
       
Once the queue is empty, the proposed polymer $\gamma$ is precisely the connected subgraph induced by all the vertices that were successfully added to $Q$ during this process. The whole frontier $\partial_R \gamma$; all vertices $v\in V(G_{\text{int},R}) $ connected to $\gamma$, is explored but not added to the final subgraph $\gamma$. Thereby, given the seed $u$ and maximum radius $R$, the probability of obtaining a certain polymer by the BFS search is

        \begin{align}\label{eq_ap3c_proposolalaw}
            \mu_{u |R}(\gamma)=\prod_{Z \in \gamma \setminus\{ u \}}q_Z \prod_{Z' \in \partial_R \gamma}(1-q_{Z'}).
        \end{align}

		\item \textit{Step 4 (Final filter):} We compute $d^*(\gamma)=\max_{Z\in \gamma}d(Z)$, and the number of vertices of $\gamma$ connected to the vertex $v$;  $m_v(\gamma)=\vert \gamma \cap \mathcal{R}_R(v)\vert $.
        Let $\gamma$ be the proposed polymer generated in the previous step. Note that $\gamma \in S(G,v|R)\cup\{\emptyset\}$, where $S(G,v|R)$ denotes the set of all valid polymers connected to $v$ and restricted to interaction distances $d(Z) \leq R$. We accept this proposed polymer with probability
        \begin{align}\label{eq_filterfunction}
        f(\gamma |R) = \frac{\vert w_\gamma \vert}{\text{Pr}(R\geq d^*(\gamma))\cdot m_v(\gamma)\prod_{Z\in \gamma}q_Z \prod_{Z' \in \partial_R \gamma }(1-q_{Z'})}, 
\end{align}
and if the filter is not succesfully surpased return $\gamma = \emptyset$. Notice that for a probability filter to be valid it must be $\leq 1$. This is true because of Lemma \ref{lem:auxiliarProof_longRangeSampler_filter}.

Having outlined the algorithm, we must now verify that it generates polymers $\gamma$ exactly according to the target distribution $\nu_v(\gamma)$. 

For a given initial vertex $v$ and a fixed interaction radius $R$, the probability $\nu_{v|R}(\gamma)$ of obtaining $\gamma$ as the output of the BFS is obtained by marginalizing $\mu_{u' |R}(\gamma)$ over all possible seeds
\begin{align}
    \nu_{v|R}(\gamma)=\sum_{v' \in \mathcal{R}_R(v)} q_{v'} \mu_{v' |R}(\gamma) = m_{v}(\gamma) \prod_{Z \in \gamma }q_Z \prod_{Z' \in \partial_R \gamma}(1-q_{Z'}).
\end{align}
Subsequently, applying the rejection filter yields the conditional acceptance probability
\begin{align}
    \Pr[\gamma \text{ as output} |R]=m_{v}(\gamma) \prod_{Z \in \gamma }q_Z \prod_{Z' \in \partial_R \gamma}(1-q_{Z'}) \cdot  \frac{\vert w_\gamma \vert}{\text{Pr}(R\geq d^*(\gamma))\cdot m_v(\gamma)\prod_{Z\in \gamma}q_Z \prod_{Z' \in \partial_R \gamma }(1-q_{Z'})} = \frac{\vert w_\gamma  \vert}{\text{Pr}(R\geq d^*(\gamma))}.
\end{align}
Finally, we marginalize over the random choice of $R$. Noting that the probability of outputting $\gamma$ strictly vanishes unless $R \geq d^*(\gamma)$, this yields
\begin{align}
    \Pr[\gamma \text{ as output} ]= \sum_R \Pr(R) \Pr[\gamma \text{ as output} |R] = \sum_{R\geq d^*(\gamma)} \Pr(R) \frac{\vert w_\gamma  \vert}{\text{Pr}(R\geq d^*(\gamma))} = \vert w_\gamma \vert,
\end{align}
which confirms that the procedure acts as a perfect sampler for the desired distribution
\begin{align}
    \nu_v (\gamma)= \begin{cases}
         \vert w_\gamma \vert \quad  \forall \gamma \nsim v\\ 1- \sum_{\gamma \nsim v}  \vert w_\gamma \vert \quad  \text{if } \gamma = \emptyset
    \end{cases}
\end{align}

\end{itemize}

	\begin{proof}(Lem. \ref{lem:longrange_polymer_sampler})
	Once we designed the algorithm and proved it is well defined, we examine the runtime step by step.
        
        \textit{Step 1:}
         The computational complexity is strictly $\mathcal{O}(1)$, as sampling the interaction radius $R \sim \Pr(R) \propto R^{-\tilde{\alpha}}$ via standard scalar techniques (e.g., inverse transform sampling) requires strictly $\mathcal{O}(1)$ time, being completely independent of the system size $N$ and the current polymer state.
\noindent\emph{Steps 2 and 3.} The cost of these two steps is incurred on
\emph{every} call to the sampler, irrespective of the outcome of the filter of
Step~4; it must therefore be averaged over the law of the \emph{proposal},
conditioned on $R$ and on the seed; $\mu_{u |R}(\gamma)$. 
 
Drawing the seed requires evaluating $q_u$ for the $|\mathcal R_R(v)|\le\Delta_R$
candidates, i.e.\ $O(\Delta_R)$ operations. With probability at least $\tfrac12$
the outcome is null (Lem. \ref{lem:appPolymerSampler_qwell}) and the algorithm stops there. Otherwise the
search of Step 3 extracts $\norm{\gamma}$ vertices and, for each of them,
evaluates $q_{Z}$ at its at most $\Delta_R$ neighbours, for a total of
$O(\Delta_R\norm{\gamma})$ operations; this also accounts for the frontier,
since $|\partial_R\gamma|\le\Delta_R\norm{\gamma}$. Each evaluation of $q_Z$
reduces to computing $\norm{h_Z}$, the diagonalisation of a matrix of fixed
dimension $d_\mathcal{H}$, hence $\mathcal{O}(1)$.
 
By Lem.\ref{lem:branch}(i) the expected number of extractions is at most $\mathcal{O}(1)$
\emph{uniformly in $R$ and in the seed}. It is precisely this uniformity that
allows for separating the two random factors:
\begin{equation}\label{eq:step23}
  \mathbb{E}_{\mu_{u |R}}\bigl[\Delta_R\norm{\gamma}\bigr]
  \;\le\;2\sum_{R\ge1}\Pr(R)\, \mathcal{O}\!\left(R^{D(k-1)}\right)\;=\;O(1),
\end{equation}
the series converging exactly when $\tilde\alpha>D(k-1)+1$, the lower end of the
admissible window of Step~1. Steps~2 and~3 therefore run in $\mathcal{O}(1)$ expected
time.
 
\medskip

\noindent\emph{Step 4.} Here one must diagonalise $h_\gamma=\sum_{Z\in\gamma}h_Z$
exactly. The total dimension of $h_\gamma$ is at most $d_\mathcal{H}^{\norm{\gamma}}$, so
the cost obeys $C_w(\gamma)=O\bigl(d_\mathcal{H}^{3\norm{\gamma}}\bigr)$. As in Steps 2
and 3 this is paid on every proposal, but the situation is now sharper: the cost
grows \emph{exponentially} in $\norm{\gamma}$, so the bound of
Lem.\ref{lem:branch}(i) on the mean size is of no use, and what is needed is a
bound on an exponential moment.
 
Set $x:=d_\mathcal{H}^{3}$. By Lem.\ref{lem:branch}(ii), $\mathbb{E}_{\mu_{u |R}}[x^{\norm{\gamma}}\mid u,R]
\le\mathrm ex=O(1)$ as soon as $\mathcal{B}\le(\mathrm e\,d_\mathcal{H}^{3})^{-1}$. By (C45) we have
$\mathcal{B}\le\mathrm{e}(\mathrm e^{\beta g}-1)ku'$, so this holds whenever
\begin{equation}\label{eq:bdiag}
  \beta\;\le\;\beta_{\rm diag}
  \;:=\;\frac{1}{g}\,
  \log\!\left(1+\frac{1}{\,\mathrm{e^2}\,k\,u'\,d_\mathcal{H}^{3}}\right),
\end{equation}
with $u'$ as in Eq. \eqref{eq:Ldef}. Under this
condition $\mathbb{E}[C_w]=O(1)$, uniformly in $R$ and in the seed, so the weight
evaluation introduces no asymptotic overhead.
 
\medskip

In conclusion, we have demonstrated that the polymer sampler correctly targets the distribution $\nu_v (\gamma)$. By defining the threshold
\begin{align}\label{eq:def_betasampling}
  \beta\;\le\;\beta_{\rm sampling}
  &=\min\bigl(\beta_{\rm diag},\,\beta^{**}\bigr)\;\equiv\;g^{-1}c_\beta,
  \nonumber\\[2pt]
  c_\beta&=\min\left\{
     \log\!\Bigl(1+\tfrac{1}{\mathrm e^2 k u' d_\mathcal{H}^{3}}\Bigr),\;
     \tfrac12\log\!\Bigl(1+\tfrac{1}{(24\sqrt{\mathrm e}\,uk)^{2}}\Bigr)
     \right\}.
\end{align}
\end{proof}
\bigskip
\hrule
\bigskip

\begin{rem}
\label{rem_kdependence}
    For $k$-local Hamiltonians, the local dimension scales as $d_\mathcal{H} \propto \exp(k)$, which restricts the efficient sampling regime to $\beta_{\text{sampling}} \propto \exp(-k)$. In contrast, the fundamental convergence threshold is significantly broader, scaling as $\beta^{**} \le \mathcal{O}(k^{-2})$. We note, however, that one could bypass this randomized sampling bottleneck by accepting the worst-case deterministic cost of the polymer sampler, which scales as $\exp(M)$. While this alternative introduces an additional $\text{poly}(N, \varepsilon^{-1})$ overhead to the total runtime, it successfully enables the overall algorithm to operate at the tighter convergence threshold $\beta^{**} \le \mathcal{O}(k^{-2})$.
\end{rem}

\subsubsection{Auxiliary lemmas for the polymer sampler}

Let us introduce some additional notation for this subsection. The radius $R$ of Step 1 is random, so we write
\begin{equation}
  F(d)\;:=\;\Pr\bigl(R\ge d\bigr)\;=\;\frac{1}{\zeta(\tilde\alpha)}\sum_{s\ge d}s^{-\tilde\alpha}.
  \label{eq:Fdef}
\end{equation}
With it, the local weights of Step 2 read
\begin{equation}
  q_Z\;=\;\frac{\mathrm{e}\left( \mathrm{e}^{\beta g}-1 \right)}{F\bigl(d(Z)\bigr)}\cdot\frac{\lVert h_Z\rVert}{g}.
  \label{eq:qdef}
\end{equation}
Let us introduce the constant $u'$, the reweighted analogue of the geometric constant $u$ of Eq.\ref{eq:u}
\begin{equation}
  u'\;:=\;\max_{i\in\Lambda}\ \sum_{Z\ni i}\ \frac{\lVert h_Z\rVert}{g\,F\bigl(d(Z)\bigr)} .
  \label{eq:Ldef}
\end{equation}

\begin{lemma}[Local weights are well defined]\label{lem:appPolymerSampler_qwell}
Let $\tilde\alpha<\alpha+1-D$ be the radial exponent of Step 1.
Then $u'<\infty$. Moreover, if the inverse temperature satisfies 
\begin{equation}
  \beta\;\leq \beta_{\text{loc}}= \frac{1}{g}\log\!\Bigl(1+\frac{1}{2\,\mathrm{e}\,k\,u'}\Bigr) \leq \beta_{\text{diag}},
  \label{eq:betaloc}
\end{equation}
then
\begin{equation}
  \mathcal B\;:=\;\max_{Z\in V(G_{\text{int}})}\ \sum_{Z'\nsim Z} q_{Z'}
  \label{eq:Bdef}
\end{equation}
obeys $\mathcal B\le\tfrac12$. In particular, for every term $Z$, every vertex $v$ and every radius $R$,

\begin{align}
    \sum_{u\in\mathcal R_R(v)}q_u\le\tfrac12\le1
\end{align}
so Step 2 is a well-defined sub-probability distribution on $\mathcal R_R(v)$, with null outcome of probability
        $1-\sum_{u\in\mathcal R_R(v)}q_u\ge\tfrac12$.
\end{lemma}

\begin{proof}
\textit{Finiteness of $u'$.} Comparing the tail with an integral,
$\sum_{s\ge d}s^{-\tilde\alpha}\ge\int_d^{\infty}s^{-\tilde\alpha}\,ds=d^{1-\tilde\alpha}/(\tilde\alpha-1)$,
so that
\begin{equation}
  \frac{1}{F(d)}\;\le\;(\tilde\alpha-1)\,\zeta(\tilde\alpha)\,d^{\,\tilde\alpha-1}.
  \label{eq:Ftail}
\end{equation}
Fix a site $i$ and a term $Z\ni i$ with $\lvert Z\rvert\ge2$, and let $j^\ast\in Z$ maximise
$d_{ij}$. By the triangle inequality every pair of sites of $Z$ is within
$2d_{ij^\ast}$ of each other, so $d(Z)\le 2d_{ij^\ast}$ and, $F$ being non-increasing,
$1/F(d(Z))\le 1/F(2d_{ij^\ast})$. Assigning each such $Z$ to the edge $(i,j^\ast)$ and
summing over all edges incident to $i$ can only over-count, whence by Eq. \eqref{eq:longrange_definition}
\begin{align}\label{eq_constrainAlphaTilde2}
  \sum_{\substack{Z\ni i\\ \lvert Z\rvert\ge2}}\frac{\lVert h_Z\rVert}{g\,F(d(Z))}
  \;\leq &\;\sum_{j\ne i}\frac{1}{F(2d_{ij})}\sum_{Z\ni(i,j)}\frac{\lVert h_Z\rVert}{g} \\
  \;\le\;\sum_{j\ne i}\frac{(1+d_{ij})^{-\alpha}}{F(2d_{ij})} .
\end{align}
Using \eqref{eq:Ftail} and grouping the sites $j$ by their distance $r=d_{ij}$, of which there are
at most $c_D\,r^{D-1}$ at each $r$ ,
\begin{align}
  \sum_{\substack{Z\ni i\\ \lvert Z\rvert\ge2}}&\frac{\lVert h_Z\rVert}{g\,F(d(Z))}
  \;\le \;c_D\,2^{\tilde\alpha-1}(\tilde\alpha-1)\zeta(\tilde\alpha)
  \sum_{r\ge1} r^{\,\tilde\alpha-1}\,r^{D-1}\,(1+r)^{-\alpha},
\end{align}
and the series converges since $\tilde{\alpha}<\alpha+1-D$. Adding the single-site term, bounded by
$\lVert h_{\{i\}}\rVert/g\le1$ since $F(0)=1$, gives $u'<\infty$.

\smallskip
\textit{Bound on $\mathcal B$.} Two terms are incompatible exactly when their supports on
$\Lambda$ overlap, so every $Z'\nsim Z$ contains at least one of the at most $k$ sites of
$\text{supp}(Z)$. Hence
\begin{align}
  \sum_{Z'\nsim Z} q_{Z'}
  \;\le& \;\sum_{i\in\text{supp}(Z)}\ \sum_{Z'\ni i} q_{Z'}
  \; \\
  =&\;\mathrm{e}\left( \mathrm{e}^{\beta g}-1 \right)\sum_{i\in\text{supp}(Z)}\ \sum_{Z'\ni i}\frac{\lVert h_{Z'}\rVert}{g\,F(d(Z'))}\\
  \;\le&\;\mathrm{e}\,\left( \mathrm{e}^{\beta g}-1 \right)\,k\,u' ,
\end{align}
and $\beta\le \beta_{\text{loc}}$ is precisely $\left( \mathrm{e}^{\beta g}-1 \right)\le(2\mathrm{e} ku')^{-1}$, giving $\mathcal B\le\tfrac12$.

\smallskip
\textit{(i) and (ii).} Every term is incompatible with itself, so $q_Z$ is one of the
summands defining $\mathcal B$ and therefore $q_Z\le\mathcal B\le\tfrac12$. For (ii),
$\mathcal R_R(v)=\{v'\in V(G_{\text{int, R}}):v'\nsim v\}$ is a subset of the incompatibility
neighbourhood of $v$ in $G_{\text{int}}$, so $\sum_{u\in\mathcal R_R(v)}q_u\le\mathcal B\le\tfrac12$,
uniformly in $R$.
\end{proof}

\begin{lemma}\label{lem:frontier}
Assume $\beta\le\beta_{\text{loc}}$, so that $\mathcal B\le\tfrac12$ by Lem.\ \ref{lem:appPolymerSampler_qwell}. Then, for every
polymer $\gamma$ and every radius $R$,
\begin{equation}
  \prod_{Z'\in\partial \gamma_R}\bigl(1-q_{Z'}\bigr)\;\ge\;\mathrm{e}^{-2\mathcal B\lVert\gamma\rVert}\;\ge\;\mathrm{e}^{-\lVert\gamma\rVert}.
  \label{eq:frontier}
\end{equation}
\end{lemma}

\begin{proof}
For $x\in[0,\tfrac12]$ one has $1-x\ge e^{-2x}$: the function $h(x)=\log(1-x)+2x$ vanishes at
$x=0$ and has $h'(x)=2-(1-x)^{-1}\ge0$ on $[0,\tfrac12]$. By Lem.\ \ref{lem:appPolymerSampler_qwell}(i) every
$q_{Z'}$ lies in that range, so
\begin{equation}
  \prod_{Z'\in\partial \gamma_R}(1-q_{Z'})\;\ge\;\exp\Bigl(-2\sum_{Z'\in\partial \gamma_R}q_{Z'}\Bigr).
\end{equation}
The frontier $\partial \gamma_R$ consists of vertices of $G_{\text{int, R}}$ adjacent to $\gamma$, hence incompatible
with at least one of its terms, so $\partial \gamma_R\subseteq\bigcup_{Z\in\gamma}\{Z':Z'\nsim Z\}$ and
\begin{equation}
  \sum_{Z'\in\partial \gamma_R}q_{Z'}\;\le\;\sum_{Z\in\gamma}\ \sum_{Z'\nsim Z}q_{Z'}\;\le\;\mathcal B\,\lVert\gamma\rVert .
  \qedhere
\end{equation}
\end{proof}

\begin{lemma}[The filter $f(\gamma |R)$ is a valid probability]\label{lem:auxiliarProof_longRangeSampler_filter}
Let $\tilde\alpha<\alpha+1-D$, and let $\beta\le\beta_{\text{loc}}$. Then the filter function of Eq. \ref{eq_filterfunction}
satisfies $0\le f(\gamma\mid R)\le1$ for every polymer $\gamma\nsim v$ and every radius
$R\ge d^\ast(\gamma)$.
\end{lemma}

\begin{proof}
Non-negativity is clear. For the upper bound we treat the three factors of the denominator in turn.

\smallskip
\textit{(a) } By Lem. \ref{lemma_boundweight}, Eq.\ (C2),
\begin{equation}
  \lvert w_\gamma\rvert\;\le\;\left( \mathrm{e}^{\beta g}-1 \right)^{\lVert\gamma\rVert}\prod_{Z\in\gamma}\frac{\lVert h_Z\rVert}{g}
  \;=\;\frac{1}{\mathrm{e}^{\lVert\gamma\rVert}}\ \prod_{Z\in\gamma} q_Z\,F\bigl(d(Z)\bigr),
\end{equation}
where the equality is just the definition \eqref{eq:qdef} of $q_Z$ solved for
$\left( \mathrm{e}^{\beta g}-1 \right)\lVert h_Z\rVert/g$.

\smallskip
\textit{(b)} Every factor $F(d(Z))$ lies in
$(0,1]$, and the term $Z_0\in\gamma$ attaining $d(Z_0)=d^\ast(\gamma)$ contributes exactly
$F(d^\ast(\gamma))$. Hence
\begin{equation}
  \prod_{Z\in\gamma}F\bigl(d(Z)\bigr)\;\le\;F\bigl(d^\ast(\gamma)\bigr)\;=\;\Pr\bigl(R\ge d^\ast(\gamma)\bigr),
\end{equation}
and combining with (a),
\begin{equation}
  \frac{\lvert w_\gamma\rvert}{\Pr\bigl(R\ge d^\ast(\gamma)\bigr)\prod_{Z\in\gamma}q_Z}
  \;\le\;\frac{1}{\mathrm{e}^{\lVert\gamma\rVert}} .
  \label{eq:stepb}
\end{equation}

\smallskip
\textit{(c)} By Lem.\ \ref{lem:frontier},
$\prod_{Z'\in\partial \gamma_R}(1-q_{Z'})\ge e^{-2\mathcal B\lVert\gamma\rVert}$.

\smallskip
Putting (a)--(c) together with $m_v(\gamma)\ge1$,
\begin{align}
  f(\gamma\mid R)\;=&\;\frac{\lvert w_\gamma\rvert}{\Pr(R\ge d^\ast(\gamma))\,m_v(\gamma)
  \prod_{Z\in\gamma}q_Z\prod_{Z'\in\partial \gamma_R}(1-q_{Z'})} \\
  \;\le&\;\Bigl(\frac{e^{2\mathcal B}}{\mathrm{e}}\Bigr)^{\lVert\gamma\rVert}\;\le\;1,
\end{align}
where the last inequality follows from $\mathcal{B}\leq \tfrac{1}{2}$
\end{proof}

\noindent For a fixed radius $R$ and a fixed non-null seed $u$, the search of
Step~3 is a stochastic process; we write $\mathbb P_{u,R}$ for its law on
trajectories, $\mathbb E_{u,R}$ for the expectation under it, and
$(\mathcal F_t)_{t\ge0}$ for the filtration it generates. The polymer $\gamma$
is a function of the trajectory, and the proposal law $\mu_{v|R}$ of
Eq.~\eqref{eq_ap3c_proposolalaw} is the mixture over the seed of its law under
$\mathbb P_{u,R}$.

\begin{lemma}[The exploration is subcritical]\label{lem:branch}
Assume $\beta\le\beta_{\rm loc}$, so that $\mathcal B\le\tfrac12$ by
Lem.~\ref{lem:appPolymerSampler_qwell}. Fix a radius $R$ and a non-null seed
$u$, and let $\gamma$ be the polymer produced by the breadth-first search of
Step~3. Then
\begin{enumerate}
\item[(i)] $\displaystyle \mathbb E_{u,R}\bigl[\norm{\gamma}\bigr]\;\le\;\frac{1}{1-\mathcal B}\;\le\;2$;
\item[(ii)] for every $x\ge1$ and every $\psi\ge1$ with
      $x\,\mathrm e^{\mathcal B(\psi-1)}\le\psi$,
      $\displaystyle \mathbb E_{u,R}\bigl[x^{\norm{\gamma}}\bigr]\le\psi$. In
      particular, if $\mathcal B\le(\mathrm e x)^{-1}$ then $\psi=\mathrm e x$
      is admissible.
\end{enumerate}
Both bounds are uniform in $u$ and in $R$.
\end{lemma}

\begin{proof}
Every expectation below is under $\mathbb P_{u,R}$ for the fixed pair $(u,R)$,
so we drop the subscript. Describe Step~3 by its queue: let $P_t$ be the number
of vertices already extracted after $t$ iterations and $Q_t:=|Q|$ at that
moment, so that $P_0=0$, $Q_0=1$ (the seed) and, writing
$T:=\min\{t:Q_t=0\}$, on $\{T>t\}$,
\begin{equation}\label{eq:queue}
  P_{t+1}=P_t+1,\qquad Q_{t+1}=Q_t-1+K_t ,
\end{equation}
where $K_t$ is the number of neighbours of the $t$-th extracted vertex that
Step~3.2 adds to the queue; on $\{T\le t\}$ the process is frozen,
$P_{t+1}:=P_t$ and $Q_{t+1}:=Q_t=0$. These are all random variables on the
single space carrying $\mathbb P_{u,R}$ --- the index $t$ labels the variables,
not a family of measures --- and $T$ is a stopping time for
$(\mathcal F_t)$. Counting the queue over the whole run gives
$0=Q_T=1-T+\sum_{t<T}K_t$, so the number of vertices ever enqueued equals the
number ever extracted and $\norm{\gamma}=T$.

The whole argument rests on one observation. Conditionally on $\mathcal F_t$
the extracted vertex and its unexplored neighbours are known, and $K_t$ is a sum
of \emph{independent} Bernoulli variables, one per such neighbour $Z$ with
parameter $q_Z$; every such $Z$ is incompatible with the extracted vertex, so
those parameters sum to at most $\mathcal B$ by \eqref{eq:Bdef}. Hence, almost
surely on $\{T>t\}$ and for every $s\ge1$,
\begin{equation}\label{eq:pgf}
  \mathbb E[K_t\mid\mathcal F_t]\le \mathcal B,
  \qquad
  \mathbb E\bigl[s^{K_t}\mid\mathcal F_t\bigr]
  =\prod_{Z}\bigl(1+q_Z(s-1)\bigr)\le \mathrm e^{\mathcal B(s-1)} ,
\end{equation}
using $1+y\le\mathrm e^{y}$. Restricting the search to $G_{{\rm int},R}$ and
skipping already-explored vertices only removes summands, so \eqref{eq:pgf}
holds for every radius, every seed and every history --- this is the source of
the claimed uniformity.

\emph{(i).} Taking expectations in \eqref{eq:queue} and using
$\{T>t\}\in\mathcal F_t$,
\begin{equation}
  \mathbb E[Q_{t+1}]=\mathbb E[Q_t]-\Pr[T>t]
  +\mathbb E\bigl[K_t\mathbf 1_{\{T>t\}}\bigr]
  \le \mathbb E[Q_t]-(1-\mathcal B)\Pr[T>t].
\end{equation}
Summing over $t=0,\dots,n-1$ and using $Q_n\ge0$, $Q_0=1$,
\begin{equation}
  (1-\mathcal B)\sum_{t<n}\Pr[T>t]\;\le\;1-\mathbb E[Q_n]\;\le\;1 .
\end{equation}
Letting $n\to\infty$ and using $\mathbb E[T]=\sum_{t\ge0}\Pr[T>t]$ gives
$\mathbb E[\norm{\gamma}]=\mathbb E[T]\le(1-\mathcal B)^{-1}\le2$; in
particular $T<\infty$ a.s.
\emph{(ii).} Put $M_t:=x^{P_t}\psi^{Q_t}\ge0$, so that $M_0=\psi$ is the bound
claimed and $M_T=x^{\norm{\gamma}}$ the quantity to be bounded. Here $x^{P_t}$
is the cost already incurred and $\psi$ a price charged in advance to each
vertex still queued, and the hypothesis
$x\,\mathrm e^{\mathcal B(\psi-1)}\le\psi$ says precisely that this price
covers processing one such vertex, a factor $x$, together with the price of
the children it releases.

On $\{T>t\}$, \eqref{eq:queue} gives
$M_{t+1}=x^{P_t+1}\psi^{Q_t-1}\,\psi^{K_t}$, whose first two factors are
$\mathcal F_t$-measurable --- after $t$ iterations they have been counted ---
so they come out of the conditional expectation and only $K_t$ is averaged:
\begin{equation}\label{eq:super}
	\mathbb E[M_{t+1}\mid\mathcal F_t]
	=x^{P_t+1}\psi^{Q_t-1}\,\mathbb E\bigl[\psi^{K_t}\mid\mathcal F_t\bigr]
	\le M_t\,\frac{x}{\psi}\,\mathrm e^{\mathcal B(\psi-1)}\le M_t ,
\end{equation}
the last inequality being the hypothesis on $\psi$. On $\{T\le t\}$ the process
is frozen and $M_{t+1}=M_t$, so $\mathbb E[M_t]\le M_0=\psi$ for every $t$.
Since $M_t=x^{\norm{\gamma}}$ once $T\le t$, we get
$\mathbb E\bigl[x^{\norm{\gamma}}\mathbf 1_{\{T\le t\}}\bigr]\le\psi$; and
since $T<\infty$ a.s.\ by (i), monotone convergence as $t\to\infty$ yields
$\mathbb E[x^{\norm{\gamma}}]\le\psi$.
\end{proof}

\end{widetext}

\end{document}